\documentclass{lmcs} 
\pdfoutput=1
\usepackage[utf8]{inputenc}

\usepackage{lastpage}
\lmcsdoi{21}{4}{17}
\lmcsheading{}{\pageref{LastPage}}{}{}%
{Dec.~02,~2024}{Oct.~30,~2025}{}


\newif\ifshowcomments
\showcommentstrue

\usepackage{amssymb, amsthm, amsmath, mathtools}
\usepackage{comment}
\usepackage[linesnumbered,lined,ruled, commentsnumbered, noend]{algorithm2e}
\usepackage{multirow}
\usepackage{cite}
\usepackage{multicol}
\usepackage{graphicx}
\usepackage{aliascnt}
\usepackage{balance}
\usepackage{mathtools}
\usepackage[normalem]{ulem}
\usepackage{booktabs}
\usepackage{fancyvrb}
\usepackage{bm}
\usepackage{colortbl}
\usepackage{footnote}
\usepackage{xspace}
\usepackage{url}
\usepackage{textcomp}
\usepackage{booktabs}
\usepackage{hyperref}
\usepackage{color}
\usepackage{xspace}
\usepackage{caption}
\usepackage{float}
\usepackage{wrapfig,lipsum}
\renewcommand{\paragraph}[1]{\medskip \noindent {\bf #1}}

\newcommand{\polylog}{\operatorname{polylog}}

\newcommand{\eps}{\varepsilon}
\newcommand{\T}{\mathcal{T}}

\renewcommand{\O}{\tilde{O}}

\renewcommand{\Re}{\mathbb{R}}

\newcommand{\new}[1]{{#1}}
\newcommand{\renyi}{R\'enyi\xspace}
\newcommand{\ren}{H}
\newcommand{\matrixmul}{\mathcal{M}}
\newcommand{\out}{\xi}
\theoremstyle{plain} 
\DeclareMathOperator*{\argmax}{arg\,max}

\begin{document}

\title{Range (R\'enyi) Entropy Queries and Partitioning}

\author[A.~Esmailpour]{Aryan Esmailpour\lmcsorcid{0009-0000-3798-9578}}[a]
\author[S.~Krishnan]{Sanjay Krishnan\lmcsorcid{0000-0001-6968-4090}}[b]
\author[S.~Sintos]{Stavros Sintos\lmcsorcid{0000-0002-2114-8886}}[a]

\address{University of Illinois Chicago}	
\email{aesmai2@uic.edu, stavros@uic.edu}  

\address{University of Chicago}	
\email{skr@uchicago.edu}  

\begin{abstract}
Data partitioning that maximizes/minimizes the Shannon entropy, or more generally the \renyi entropy is a crucial subroutine in data compression, columnar storage, and cardinality estimation algorithms.
These partition algorithms can be accelerated if we have a data structure to compute the entropy in different subsets of data when the algorithm needs to decide what block to construct. Such a data structure will also be useful for data analysts exploring different subsets of data to identify areas of interest. For example, subsets with high entropy might correspond to dirty data in data cleaning or areas with high biodiversity in ecology.
While it is generally known how to compute the Shannon or the \renyi entropy of a discrete distribution in the offline or streaming setting efficiently, we focus on the query setting where we aim to efficiently derive the entropy among a subset of data that satisfy some linear predicates. We solve this problem in a typical setting when we deal with real data, where data items are geometric points and each requested area is a query (hyper)rectangle. More specifically, we consider a set $P$ of $n$ weighted and colored points in $\Re^d$, where $d$ is a constant.
For the range S-entropy (resp. R-entropy) query problem, 
the goal is to construct a low space data structure, such that given a query (hyper)rectangle $R$, it computes the Shannon (resp. \renyi) entropy based on the colors and the weights of the points in $P\cap R$, in sublinear time.
We show conditional lower bounds proving that we cannot hope for data structures with near-linear space and near-constant query time for both the range S-entropy and R-entropy query problems. Then, we propose exact data structures for $d=1$ and $d>1$ with $o(n^{2d})$ space and $o(n)$ query time for both problems. We also provide a tuning parameter $t$ that the user can choose to bound the asymptotic space and query time of the new data structures. Next, we propose near linear space data structures for returning either an additive or a multiplicative approximation of the Shannon (resp. \renyi) entropy in $P\cap R$. Finally, we show how we can use the new data structures to efficiently partition time series and histograms with respect to the Shannon entropy.
\end{abstract}


\maketitle

\section{Introduction}
\label{sec:intro}
Discrete Shannon entropy is defined as the expected amount of information needed to represent an event drawn from a probability distribution. That is, given a probability distribution $\mathcal{D}$ over the set $\mathcal{X}$, the Shannon entropy is defined as\footnote{We use $\log(\cdot)$ for the logarithmic function with base $2$.}
$$
H(\mathcal{D}) = -\sum_{x \in \mathcal{X}} \mathcal{D}(x) \cdot \log \mathcal{D}(x).$$
In information theory, the \renyi entropy is a quantity that generalizes Shannon entropy and various other notions of entropy, including Hartley entropy, collision entropy, and min-entropy.
The \renyi entropy of order $\alpha>1$ for a distribution $\mathcal{D}$ is defined as\footnote{Although the \renyi entropy can be defined for any order $\alpha>0$, for simplicity we focus on the case where $\alpha>1$, as was also done in~\cite{obremski2017renyi}. Most of our methods and data structures can be extended to the range $\alpha\in(0,1)$.} $$\ren_\alpha(\mathcal{D})=-\frac{1}{\alpha-1}\log\left(\sum_{x\in \mathcal{X}}(\mathcal{D}(x))^\alpha\right).$$
It is known that $\lim_{\alpha\rightarrow 1}\ren_{\alpha}(\mathcal{D})=H(\mathcal{D})$.
\new{Some other common values of $\alpha$ that are used in the literature are: $\alpha=2$ (Collision entropy~\cite{bosyk2012collision}) and $\alpha\rightarrow \infty$ (Min entropy~\cite{konig2009operational}).}

The Shannon and \renyi entropy have a few different interpretations in information theory, statistics, and theoretical computer science such as:
\begin{itemize}
    \item (Compression) Entropy is a lower bound on data compressibility for datasets generated from the probability distribution via the Shannon source coding theorem.
    \item (Probability) Entropy measures a probability distribution's similarity to a uniform distribution over the set $\mathcal{X}$ on a scale of $[0, \log |\mathcal{X}| ]$.
    \item (Theoretical computer science) Entropy is used in the context of randomness extractors~\cite{vadhan2012pseudorandomness}.
\end{itemize}
Because of these numerous interpretations, entropy is a highly useful optimization objective. Various algorithms, ranging from columnar compression algorithms to histogram construction and data cleaning, maximize or minimize (conditional) entropy as a subroutine. 
These algorithms try to find high or low entropy data subsets.
Such algorithms can be accelerated if we have a data structure to efficiently calculate the entropy of different subsets of data.
While it is known how to compute the entropy of a distribution efficiently, there is little work on such ``range entropy queries'', where we want to derive efficiently the entropy among the data items that lie in a specific area. 
To make this problem more concrete, let us consider a few examples.
\begin{exa}[Columnar Compression]
\label{ex1}
An Apache Parquet file is a columnar storage format that first horizontally partitions a table into row groups, and then applies columnar compression along each column within the row group. A horizontal partitioning that minimizes the Shannon entropy within each partition can allow for more effective columnar compression~\cite{hansert2024partition}.
\end{exa}

\begin{exa}[Histogram Construction]
\label{ex2}
Histogram estimation often uses a uniformity assumption, where the density within a bucket is modeled as roughly uniform.  A partitioning that maximizes the (Shannon or \renyi) entropy within each partition can allow for more accurate estimation under uniformity assumptions~\cite{to2013entropy, markl2007consistent, jizba2014multifractal}. 
\end{exa}
\begin{exa}[Data Cleaning]
\label{ex3}
As part of data exploration, a data analyst explores different subsets of data to find areas with high Shannon entropy, i.e., high uncertainty. Usually, subsets of data or items in a particular area of the dataset with high entropy contain dirty data so they are good candidates for applying data cleaning methods. For example, Chu et al.~\cite{chu2015katara} used a (Shannon) entropy-based scheduling algorithm to maximize the uncertainty reduction of candidate table patterns. Table patterns are used to identify errors in data.
\end{exa}

\begin{exa}[Diversity index]
    \label{ex3a}
    The \renyi  entropy is used in ecology as a diversity index to measure how many different types (e.g., species) there exist in an area~\cite{chao2016phylogenetic, chao2015estimating}.
    An ecologist might explore different subsets of data to find areas with high or low entropy, corresponding to areas with high or low biodiversity.
\end{exa}

\new{
\begin{exa}[Network-Traffic Anomaly Detection]
\label{ex4}
The \renyi entropy has been employed to detect sudden distribution shifts in high-volume network traffic. 
By monitoring Rényi entropy over sliding windows of flow features, one can spot anomalies such as
DDoS bursts or malware beacons more sensitively than with Shannon entropy. 
For instance, Yu et al.~\cite{yu2024renyi, yu2020multiple} design a Rényi-entropy–driven detector that automatically
sets dynamic thresholds and achieves higher precision and recall than state-of-the-art statistical
baselines on real backbone-trace datasets.
\end{exa}
}


\begin{figure}
    \includegraphics[scale=0.4]{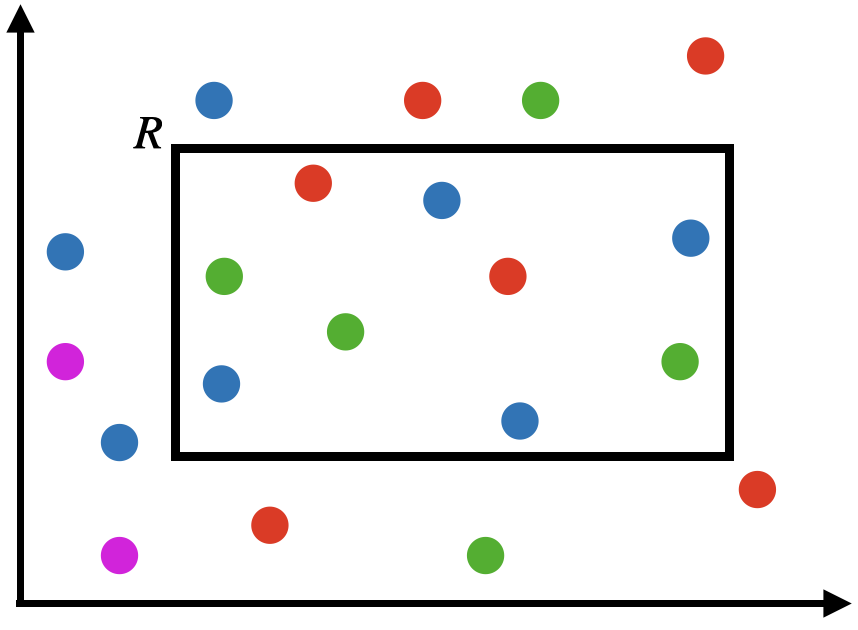}
   \caption{A set $P$ of $20$ points in $\Re^2$. For simplicity, assume that the weight of every point is $1$, i.e., $w(p)=1$ for every $p\in P$. Each point is associated with one color (or category) red, green, blue, or purple. There are three different colors among the points in $P\cap R$, namely red, green, and blue. The distribution $\mathcal{D}_R$ is defined over $3$ outcomes: red, green, and blue. The probability of red is $\mathcal{D}_R(\mathsf{red})=\frac{2}{9}$ because there are $2$ red points and $9$ total points in $P\cap R$. Similarly, the probability of green is $\mathcal{D}_R(\mathsf{green})=\frac{3}{9}$ and the probability of blue is $\mathcal{D}_R(\mathsf{blue})=\frac{4}{9}$. We have $H(P\cap R)=H(\mathcal{D}_R)=\frac{2}{9}\log\frac{9}{2}+\frac{3}{9}\log\frac{9}{3}+\frac{4}{9}\log\frac{9}{4}\approx 1.53$ and $\ren_2(P\cap R)=\ren_2(\mathcal{D}_R)=-\log\left((2/9)^2+(3/9)^2+(4/9)^2\right)\approx 1.48$.\label{fig:motiv}}
\end{figure}

The first two examples above have a similar structure, where an outer algorithm leverages a subroutine that identifies data partitions that minimize or maximize entropy. In the last two examples, we aim to explore areas with high or low entropy by running arbitrary range entropy queries.
We formulate the problem of range entropy query in a typical and realistic setting when we deal with real data: We assume that each item is represented as a point in Euclidean space.
More specifically, we consider a set $P$ of $n$ weighted and colored points in $\Re^d$. Each point $p\in P$ has a color (category) $u(P)$ and a weight $w(p)\in \Re$.
We aim to compute the Shannon or \renyi entropy of the points in $P\cap R$.
The entropy of $P\cap R$ is defined as the entropy of a discrete distribution $\mathcal{D}_R$ over the colors in $P\cap R$: Let $U_R$ be the set of all colors of the points in $P\cap R$. For each color $u_j\in U_R$, we define a value (we can also refer to it as an independent event or outcome) $\out_j$ with probability $\mathcal{D}_R(\out_j)$ equal to the sum of weights of points with color $u_j$ in $P\cap R$ divided by the sum of the weights of all points in $P\cap R$.
In other words, the discrete distribution $\mathcal{D}_R$ has $|U_R|$ outcomes corresponding to the points' colors in $P\cap R$, and each outcome $\out_j=u_j\in U_R$ has probability $\mathcal{D}_R(\out_j)=\frac{\sum_{p\in P\cap R, u(p)=u_j}w(p)}{\sum_{p\in P\cap R}w(p)}$.
Notice that $\sum_{u_j\in U_R}\mathcal{D}_R(\out_j)=1$.
The Shannon entropy of $P\cap R$ is denoted by $H(P\cap R)=H(\mathcal{D}_R)$, and the \renyi entropy of $P\cap R$ is denoted by $\ren_\alpha(P\cap R)=\ren_\alpha(\mathcal{D}_R)$.
See Figure~\ref{fig:motiv} for an example.
The goal is to construct a data structure on $P$ such that given a region (for example a rectangle) $R$, it computes the Shannon (or \renyi) entropy of the points in $P\cap R$, i.e., the Shannon (or \renyi) entropy of the distribution $\mathcal{D}_R$.
Unfortunately, we do not have direct access to distribution $\mathcal{D}_R$; we would need $\Omega(n)$ time to construct the entire distribution $\mathcal{D}_R$ in the query phase. Using the geometry of the points along with key properties from information theory we design data structures such that after some pre-processing of $P$, given any query rectangle $R$, we compute $H(\mathcal{D}_R)$,  $\ren_\alpha(\mathcal{D}_R)$ without constructing $\mathcal{D}_R$ explicitly.



\begin{defi}[Range S-entropy query problem]
Given a set $P$ of $n$ weighted and colored points in $\Re^d$, the goal is to construct a data structure with low space such that given any query rectangle $R$, it returns $H(P\cap R)$ in sub-linear time $o(n)$.
\end{defi}
\begin{defi}[Range R-entropy query problem]
Given a set $P$ of $n$ weighted and colored points in $\Re^d$, and a parameter $\alpha>1$, the goal is to construct a data structure with low space such that given any query rectangle $R$, it returns $\ren_{\alpha}(P\cap R)$ in sub-linear time $o(n)$.
\end{defi}
We assume throughout that the dimension $d$ is constant.

As we show later, both query problems can be solved by constructing near linear \new{size} data structures with query time that depends linearly on the number of colors (see Section~\ref{sec:prelim}).
However, these are efficient data structures with $o(n)$ query time because in the worst case the number of different colors is $O(n)$. Our goal is to construct data structures whose query time is always sublinear with respect to $n$. We study both exact and approximate data structures. Exact data structure return $H(P\cap R)$ (resp. $\ren_\alpha(P\cap R)$) exactly, while approximated data structure return either an additive or multiplicative approximation of $H(P\cap R)$ (resp. $\ren_\alpha(P\cap R)$).

We note that known algorithms for estimating the Shannon entropy usually do not work for estimating the \renyi entropy and vice versa. Hence, different data structures are needed to solve the range S-entropy query problem and the range R-entropy query problem.

\new{
Our range S-entropy (equivalently R-entropy) query is essentially a \emph{range colored query}, as commonly defined in the literature. Range colored queries have been extensively studied, both in theory and in the database community. Typically, they are modeled as follows:
Given a set $P$ of colored points in $\mathbb{R}^d$ with $n = |P|$, and a real-valued function $f$ defined over the colored points of $P$, the goal is to construct a data structure that efficiently computes $f(P \cap R)$ for any query range $R$.
Various functions $f$ have been studied in the past, including counting, reporting, and ratio computations.
We discuss the connection between our problems and range colored queries in the related work, later in this section.
}

\paragraph{Useful notation.}
Throughout the paper we use the following notation.
Let $P$ be a set of $n$ points in $\Re^d$ and let $U$ be a set of $m$ colors $U=\{u_1, \ldots, u_m\}$.
Each point $p\in P$ is associated with a color from $U$, i.e., $u(p)=u_i$ for $u_i\in U$. Furthermore, each point $p\in P$ is associated with a non-negative weight $w(p)\geq 0$.
For a subset of points $P'\subseteq P$, let $P'(u_i)=\{p\in P'\mid u(p)=u_i\}$, for $i\leq m$, be the set of points having color $u_i$.
Let $u(P')=\{u_i\mid \exists p\in P', u(p)=u_i\}$ be the set of colors of the points in $P'$. Finally, let $w(P')=\sum_{p\in P'}w(p)$.

\paragraph{Summary of Results.}
One of the main challenges with range entropy queries is that entropy is not a \emph{decomposable} quantity. Let $P_1, P_2$ be two sets of points such that $P_1\cup P_2=P$ and $P_1\cap P_2=\emptyset$. If we know $H(P_1), H(P_2)$ there is no straightforward way to compute $H(P_1\cup P_2)$. Similarly,  if we know $\ren_{\alpha}(P_1), \ren_{\alpha}(P_2)$ there is no straightforward way to compute $\ren_{\alpha}(P_1\cup P_2)$. In this paper, we build low space data structures such that given a rectangle $R$, we visit points or subsets of points in $P\cap R$ in a particular order and carefully update the overall entropy. All results for the S-entropy query can be seen in Table~\ref{tab:results}, while all results for R-entropy query can be seen in Table~\ref{tab:Renresults}.
\begin{table}[t]
    \centering
    \resizebox{\columnwidth}{!}{
    \begin{tabular}{|c|c|c|c|}
    \hline
    Type & Space & Query Time & Preprocessing\\\hline
    \ifshowcomments
    \new{Lower bound (Space/Query), $d=1$}& \new{$\widetilde{\Omega}\left(\frac{n^2}{(Q(n))^4}\right)$} & \new{$Q(n)$} & \new{--} \\\hline
    \fi
    Lower bound (Space/Query), $d\geq 2$& $\widetilde{\Omega}\left(\left(\frac{n}{Q(n)}\right)^2\right)$ & $Q(n)$ & -- \\\hline
    Lower bound (Prep./Query), $d\geq 1$& -- & $Q(n)$ & $\Omega(\max\{\matrixmul(\sqrt{n})-nQ(n),1\})$ \\\hline
    $d=1$, exact & $O\left(n^{2(1-t)}\right)$ & $\O\left(n^t\right)$ & $O\left(n^{2-t}\right)$\\\hline
    $d>1$, exact & $\O\left(n^{(2d-1)t+1}\right)$ & $\O\left(n^{1-t}\right)$ & $\O\left(n^{(2d-1)t+1}\right)$\\\hline
    $d\geq 1$, $\Delta$-additive approx. & $\O\left(n\right)$ & $\O\left(\frac{1}{\Delta^2}\right)$ & $\O\left(n\right)$\\\hline
    $d\geq 1$, $(1+\eps)$-multiplicative approx. & $\O\left(n\right)$ & $\O\left(\frac{1}{\eps^2}\right)$ & $\O\left(n\right)$\\\hline
    $d=1$, $\eps$-additive and &\multirow{2}{*}{$\O\left(\frac{n}{\eps}\right)$} & \multirow{2}{*}{$\O\left(1\right)$} & \multirow{2}{*}{$\O\left(\frac{n}{\eps}\right)$} \\
    $(1+\eps)$-multiplicative approx.& & &\\\hline
    \end{tabular}
    }
    \caption{New results for the S-entropy query problem (lower bounds in the first two rows and data structures with their complexities in the next rows). $t\in[0,1]$ is a tune parameter. $\O(\cdot)$ and $\widetilde{\Omega}(\cdot)$ notation hides a $\log^{O(1)}n$ factor, where the $O(1)$ exponent is at most linear on $d$. $Q(n)$ is any function of $n$ that represents the query time of a data structure for S-entropy queries over $n$ points. $\matrixmul(\sqrt{n})$ is a function of $\sqrt{n}$ that represents the running time of the fastest algorithm to multiply two $\sqrt{n}\times\sqrt{n}$ boolean matrices.
    }
    \label{tab:results}
\end{table}

\begin{table}[t]
    \centering
    \resizebox{\columnwidth}{!}{
    \begin{tabular}{|c|c|c|c|}
    \hline
    Type & Space & Query Time & Preprocessing\\\hline
     \ifshowcomments
    \new{Lower bound (Space/Query), $d=1$}& \new{$\widetilde{\Omega}\left(\frac{n^2}{(Q(n))^4}\right)$} & \new{$Q(n)$} & \new{--} \\\hline
    \fi
    Lower bound (Space/Query), $d\geq 2$& $\widetilde{\Omega}\left(\left(\frac{n}{Q(n)}\right)^2\right)$ & $Q(n)$ & -- \\\hline
    Lower bound (Prep./Query), $d\geq 1$& -- & $Q(n)$ & $\Omega(\max\{\matrixmul(\sqrt{n})-nQ(n),1\})$ \\\hline
    $d=1$, exact & $O\left(n^{2(1-t)}\right)$ & $\O\left(n^t\right)$ & $O\left(n^{2-t}\right)$\\\hline
    $d>1$, exact & $\O\left(n^{(2d-1)t+1}\right)$ & $\O\left(n^{1-t}\right)$ & $\O\left(n^{(2d-1)t+1}\right)$\\\hline
    $d\geq 1$, $\alpha\in (1,2]$, $\Delta$-add. approx. & $\O\left(n\right)$ & $\O\left(\min\left\{\frac{\alpha}{(\alpha-1)^2\Delta^2}, \frac{1}{(1-2^{(1-\alpha)\Delta})^2}\right\}\cdot n^{1-1/\alpha}\right)$ & $\O\left(n\right)$\\\hline
     $d\geq 1$, $\alpha>2$, $\Delta$-add. approx. & $\O\left(n\right)$ & $\O\left(\min\left\{\frac{\alpha}{\Delta^2}, \frac{1}{(1-2^{(1-\alpha)\Delta})^2}\right\}\cdot n^{1-1/\alpha}\right)$ & $\O\left(n\right)$\\\hline
    $d= 1$, $\eps\cdot\frac{\alpha+1}{\alpha-1}$-add. approx. & $\O\left(\frac{\alpha \cdot n}{\eps}\right)$ & $\O\left(\log \alpha\right)$ & $\O\left(\frac{\alpha \cdot n}{\eps}\right)$\\\hline
    $d\geq 1$, $\alpha\in(1,2]$,$(1+\eps)$-mult. approx.&$\O(n)$&$\O\left(\frac{\alpha}{(\alpha-1)^2\eps^2}\cdot n^{1-1/\alpha}\right)$&$\O(n)$\\\hline
    $d\geq 1$, $\alpha>2$,$(1+\eps)$-mult. approx.&$\O(n)$&$\O\left(\frac{\alpha}{\eps^2}\cdot n^{1-1/\alpha}\right)$&$\O(n)$\\\hline
    \end{tabular}
    }
    \caption{New results for the R-entropy query problem (lower bounds in the first two rows and data structures with their complexities in the next rows). $t\in[0,1]$ is a tune parameter. $\O(\cdot)$ and $\widetilde{\Omega}(\cdot)$ notation hides a $\log^{O(1)}n$ factor, where the $O(1)$ exponent is at most linear on $d$. $Q(n)$ is any function of $n$ that represents the query time of a data structure for R-entropy queries over $n$ points. $\matrixmul(\sqrt{n})$ is a function of $\sqrt{n}$ that represents the running time of the fastest algorithm to multiply two $\sqrt{n}\times\sqrt{n}$ boolean matrices.
    }
    \label{tab:Renresults}
\end{table}

\begin{itemize}
    \item In Section~\ref{sec:prelim} we introduce some useful notation and we revisit a way to update the Shannon entropy of the union of two sets with no color in common in $O(1)$ time. Similarly, we show how to update the \renyi entropy of the union of two sets with no color in common in $O(1)$ time.
    \item In Section~\ref{sec:lowerbound}, we propose space-query and preprocessing-query tradeoff lower bound proofs for the S-entropy and R-entropy queries.
    First, we study the preprocessing-query tradeoff of our queries for $d\geq 1$. We reduce the problem of multiplying two $\sqrt{n}\times \sqrt{n}$ boolean matrices to the range S-entropy query problem (resp. R-entropy query problem) in $\Re^1$ over $n$ points. We prove a conditional lower bound showing that if we have a data structure with $P(n)$ preprocessing time and $Q(n)$ query time then the multiplication of two $\sqrt{n}\times\sqrt{n}$ boolean matrices can be done in $O(P(n)+n\cdot Q(n))$ time. Equivalently,
    any data structure for the range S-entropy (resp. R-entropy) query problem with $Q(n)$ query time must have $\Omega(\max\{\matrixmul(\sqrt{n})-n\cdot Q(n),1\})$ preprocessing time.
    Second, we study the space-query tradeoff of our queries for $d\geq 2$. We
    reduce the set intersection problem to the range S-entropy query problem (resp. R-entropy query problem) in $\Re^2$. We prove a conditional lower bound showing that
    any data structure with $Q(n)$ query time must have $\widetilde{\Omega}\left(\left(\frac{n}{Q(n)}\right)^2\right)$ space. Hence,
    we cannot hope for $O(n\polylog n)$ space and $O(\polylog n)$ query time data structures for the range S-entropy (resp. R-entropy) query problems.
    \ifshowcomments
    \new{Using ideas from the lower bound with preprocessing-query tradeoff, we also show a space-query tradeoff of our queries for $d=1$, which is weaker than the lower bound we got for $d\geq 2$. In particular, for $d=1$, we prove a conditional lower bound showing that any data structure with $Q(n)$ query time must have $\widetilde{\Omega}\left(\frac{n^2}{(Q(n))^4}\right)$ space.}
    \fi
    \item Exact data structures for $d=1$. In Section~\ref{subsec:DS1}, we efficiently partition the input points with respect to their $x$ coordinates into buckets, where each bucket contains a bounded number of points. 
    Given a query interval $R$, we visit the bounded number of points in buckets that are partially intersected by $R$ and we update the overall Shannon entropy (resp. \renyi entropy) of the buckets that lie completely inside $R$. For any parameter $t\in [0,1]$ chosen by the user, we construct a data structure in $O(n^{2-t})$ time, with $O(n^{2(1-t)})$ space and $O(n^t\log n)$ query time. The same guarantees hold for both S-entropy and R-entropy queries.
    \item In Section~\ref{subsec:DSd}, instead of partitioning the points with respect to their geometric location, we partition the input points with respect to their colors. We construct $O(n^{1-t})$ blocks where two sequential blocks contain at most one color in common. Given a query rectangle, we visit all blocks and carefully update the overall Shannon entropy (resp. \renyi entropy).
    For any tune parameter $t\in [0,1]$ chosen by the user, we construct a data structure in $O(n\log^{2d}n + n^{(2d-1)t+1}\log^{d+1} n)$ time with $O(n\log^{2d-1}n + n^{(2d-1)t+1})$ space and $O(n^{1-t}\log^{2d} n)$ query time. The same guarantees hold for both S-entropy and R-entropy queries.
    \item Additive approximation --- S-entropy. In Subsection~\ref{subsec:ApproxAdd} we use known results for estimating the Shannon entropy of an unknown distribution by sampling in the \emph{dual access model}. We propose efficient data structures that apply sampling in a query range in the dual access model.
    We construct a data structure in $O(n\log^{d}n)$ time, with $O(n\log^{d-1}n)$ space and $O\left(\frac{\log^{d+3} n}{\Delta^2}\right)$ query time. The data structure returns an additive $\Delta$-approximation of the Shannon entropy in a query hyper-rectangle, with high probability. It also supports dynamic updates in $O(\log^d n)$ time.
    \item Multiplicative approximation --- S-entropy. In Subsection~\ref{subsec:ApproxMult}
    we propose a multiplicative approximation of the entropy using the results for estimating the entropy in a streaming setting. One significant difference with the previous result is that in information theory at least $\Omega\left(\frac{\log n}{\eps^2\cdot H'}\right)$ sampling operations are needed to find get an $(1+\eps)$-multiplicative approximation, where $H'$ is a lower bound of the entropy. Even if we have efficient data structures for sampling (as we have in additive approximation) we still do not have an efficient query time if the real entropy $H$ is extremely small.
    We overcome this technical issue by considering two cases: i) there is no color with a total weight of more than $2/3$, and ii) there exists a color with a total weight of at most $2/3$.
    While in the latter case, the entropy can be extremely small, an additive approximation is sufficient in order to get a multiplicative approximation. In the former one, the entropy is large so we apply the standard sampling method to get a multiplicative approximation. 
    We construct a data structure in $O(n\log^{d}n)$ time, with $O(n\log^{d}n)$ space and $O\left(\frac{\log^{d+3}}{\eps^2}\right)$ query time. The data structure returns a multiplicative $(1+\eps)$-approximation of the Shannon entropy in a query hyper-rectangle, with high probability. It also supports dynamic updates in $O(\log^d n)$ time.
    \item Additive and multiplicative approximation --- S-entropy. In Subsection~\ref{subsec:ApproxAddMult1}, we propose a new data structure for approximating the entropy in the query range for $d=1$. We get the intuition from data structures that count the number of colors in a query interval. Such a data structure finds a geometric mapping to a different geometric space, such that if at least a point with color $u_i$ exists in the original $P\cap R$, then there is a unique point with color $u_i$ in the corresponding query range in the new geometric space. Unfortunately, this property is not sufficient for finding the entropy. Instead, we need to know more information about the weights of the points and the entropy in canonical subsets of the new geometric space, which is challenging to do.
    We construct a data structure in $O\left(\frac{n}{\eps}\log^5 n\right)$ time, with $O\left(\frac{n}{\eps}\log^2 n\right)$ space and $O\left(\log^2 n \log\frac{\log n}{\eps}\right)$ query time. The data structure returns an $(1+\eps)$-multiplicative and $\eps$-additive approximation of the entropy.
    \item Additive approximation --- R-entropy. In Subsection~\ref{subsec:RenApprox1}, we use results for estimating the \renyi entropy of an unknown distribution by sampling in the \emph{samples-only model} and the dual access model.
    We construct a data structure in $O(n\log^{d}n)$ time, with $O(n\log^{d-1}n)$ space and $O\left(\min\left\{\frac{\alpha}{\Delta^2}, \frac{1}{(1-2^{(1-\alpha)\Delta})^2}\right\}\cdot n^{1-1/\alpha}\log^{d+1} n\right)$ query time if $\alpha>2$ and $O\left(\min\left\{\frac{\alpha}{(\alpha-1)^2\Delta^2}, \frac{1}{(1-2^{(1-\alpha)\Delta})^2}\right\}\cdot n^{1-1/\alpha}\log^{d+1} n\right)$ query time if $\alpha\in(1,2]$. The data structure returns an additive $\Delta$-approximation of the \renyi entropy with high probability. It also supports dynamic updates in $O(\log^d n)$ time. The data structure works for any $d\geq 1$. In Subsection~\ref{subsec:RenApprox2}, for $d=1$, we construct a faster and deterministic data structure using ideas from the additive and multiplicative approximation data structure we designed for the range S-entropy query problem. In particular, for the range R-entropy query problem in $\Re$ we design a data structure in $O(\frac{\alpha\cdot n}{\eps}\log^2 n)$ time, with $O(\frac{\alpha\cdot n}{\eps}\log^2 n)$ space and $O(\log^2 n \log\frac{\alpha\cdot \log n}{\eps})$ query time. The data structure returns an 
    $\eps\cdot\frac{\alpha+1}{\alpha-1}$-additive approximation of the \renyi entropy.

    \item Multiplicative approximation -- R-entropy. In Subsection~\ref{subsec:Rmult}, we propose a multiplicative approximation of the \renyi entropy modifying a known algorithm in~\cite{harvey2008sketching} for estimating the \renyi entropy in the streaming setting. Interestingly, there is no known multiplicative approximation algorithm of the \renyi entropy in the streaming setting for every $\alpha>1$. The multiplicative approximation in~\cite{harvey2008sketching} works for $\alpha\in(1,2]$. Similarly, to the best of our knowledge, there is no known multiplicative approximation in the samples-only or dual access model given an unknown distribution. Taking advantage of the query setting and the geometry of the input points, we are able to design a data structure that returns a multiplicative approximation for every $\alpha>1$. More specifically, we construct a data structure in  $O(n\log^d n)$ time, with $O(n\log^d n)$ space and $O\left(\frac{\alpha}{(\alpha-1)^2\eps^2}\cdot n^{1-1/\alpha}\log^d n\right)$ query time if $\alpha\in (1,2]$, and $O\left(\frac{\alpha}{\eps^2}\cdot n^{1-1/\alpha}\log^d n\right)$ time if $\alpha>2$. The data structure returns a multiplicative $(1+\eps)$-approximation of the \renyi entropy in a query hyper-rectangle, with high probability. It also supports dynamic updates in $O(\log^d n)$ time.
    
    \item Partitioning using entropy. In Section~\ref{sec:partition} we show how our new data structures for the range S-entropy query problem can be used to run partitioning algorithms over time series, histograms, and points efficiently.
\end{itemize}

\paragraph{Comparison with the conference version.}
An earlier version of this work~\cite{krishnan2024range} appeared in ICDT 2024. There are multiple new results in this new version of our work. The main differences from the previous version are:
\begin{itemize}
    \item In Section~\ref{sec:lowerbound} we propose a new (conditional) lower bound proof with preprocessing-query tradeoff for any $d\geq 1$. In the previous version, we only had a (conditional) lower bound with space-query tradeoff for $d\geq 2$.
    \ifshowcomments
    \new{Furthermore, we added a new lower bound proof with space-query tradeoff for $d=1$.}
    \fi
    All lower bounds hold for both range S-entropy and R-entropy queries. 
    \item We extended all results from the range S-entropy query to the range R-entropy query problem. In the ICDT version, we only considered the Shannon entropy. In the new version, we design new data structures to compute the \renyi entropy of any order $\alpha>1$ in a query hyper-rectangle constructing near-linear size data structures with sublinear query time. While the exact data structures for the range R-entropy queries share similar ideas with the exact data structures for the range S-entropy queries, new techniques and novel ideas are required for the approximate data structures. All results in Table~\ref{tab:Renresults} are new.
    \item We included all the missing proofs and details from the ICDT version. More specifically, in the new version, we included: an efficient construction algorithm of the exact data structure in Subsection~\ref{subsec:DS1}, an efficient construction algorithm of the exact data structure in Subsection~\ref{subsec:DSd}, the construction of a range tree to sample a point excluding the points of a specific color in Subsection~\ref{subsec:ApproxMult}, the full correctness proof of the multiplicative algorithm in Subsection~\ref{subsec:ApproxMult}, the proof of Lemma~\ref{lem:monotone}, and the construction algorithm of the data structure in Subsection~\ref{subsec:ApproxAddMult1}.
\end{itemize}

\paragraph{Related work.}
Shannon entropy has been used a lot for partitioning to create histograms in databases.
For example, To et al.~\cite{to2013entropy} use entropy to design histograms for selectivity estimation queries.
In particular, they aim to find a partitioning of $k$ buckets in $1$d such that the cumulative entropy is maximized.
They consider a special case where they already have a histogram (so all items of the same color are accumulated to the same location) and the goal is to partition the histogram into $k$ buckets. They propose a greedy algorithm that finds a local optimum solution. However, there is no guarantee on the overall optimum partitioning. Using our new data structures, we can find the entropy in arbitrary range queries, which is not supported in~\cite{to2013entropy}. Our data structures can also be used to accelerate partitioning algorithms with theoretical guarantees (see Subsection~\ref{sec:partition}) in a more general setting, where points of the same color have different locations.

In addition, there are a number of papers that use the Shannon entropy to find a clustering of items.
Cruz et al.~\cite{cruz2011entropy} use entropy for the community detection problem in augmented social networks. They describe a greedy algorithm that exchanges two random nodes between two random clusters if the entropy of the new instance is lower.
Barbar{\'a} et al.~\cite{barbara2002coolcat} use the \emph{expected entropy} for categorical clustering. They describe a greedy algorithm that starts with a set of initial clusters, and for each new item decides to place it in the cluster that has the lowest entropy.
Li et al.~\cite{li2004entropy} also use the expected entropy for categorical clustering but they extend it to probabilistic clustering models.
Finally, Ben-Gal et al.~\cite{ben2019clustering} use the expected entropy to develop an entropy-based clustering measure that measures the homogeneity of mobility patterns within clusters of users.
All these methods do not study the problem of finding the entropy in a query range efficiently.
While these methods perform well in practice, it is challenging to derive theoretical guarantees. In spatial databases, items are represented as points in $\Re^d$, so our new data structures could be used to find faster and better entropy-based clustering techniques. For example, we could run range entropy queries with different radii around a center until we find a cluster with a small radius and small (or large) expected entropy.

There is a lot of work on computing an approximation of the Shannon and \renyi entropy in the streaming setting~\cite{bhuvanagiri2006estimating, chakrabarti2006estimating, guha2006streaming, li2011new}. For a stream of $m$ distinct values ($m$ colors in our setting) Chakrabarti et al.~\cite{chakrabarti2007near} compute an $(1+\eps)$-multiplicative approximation of the entropy in a single pass using $O(\eps^{-2}\log (\delta^{-1})\log m)$ words of space, with probability at least $1-\delta$.
For a stream of size $n$ ($n$ points in our setting) Clifford and Cosma~\cite{clifford2013simple} propose a single-pass $\eps$-additive algorithm using $O(\eps^{-2}\log n \log (n\eps^{-1}))$ bits with bounded probability.
Harvey et al.~\cite{harvey2008sketching} allow deletions in the streaming setting and they propose a single-pass $(1+\eps)$-multiplicative algorithm using $\O(\eps^{-2}\log^2 m)$ words of space with bounded probability. Furthermore, they propose a single-pass $\eps$-additive approximation using $\O(\eps^{-2}\log m)$ words of space. Finally, they design a streaming algorithm for multiplicative approximation of the \renyi entropy using $O(\frac{\log m}{|1-\alpha|\eps^2})$ bits of space, for $\alpha\in(1,2]$.
While some techniques from the streaming setting are useful in our query setting, the two problems are fundamentally different. In the streaming setting, preprocessing is not allowed, all data are processed one by one and an estimation of the entropy is maintained. In our setting, the goal is to construct a data structure such that given any query range, the entropy of the items in the range should be computed in sublinear time, i.e., without processing all items in the query range during the query phase.

Let $\mathcal{D}$ be an unknown discrete distribution over $n$ values.
There is an interesting line of work on approximating the Shannon and the \renyi entropy of $\mathcal{D}$ by applying oracle queries in the \emph{dual access model}.\footnote{In the dual access model we are given an oracle  to sample and an oracle to evaluate the probability of an outcome from an unknown distribution.  A more formal definition is given in Section~\ref{sec:approx}.}
Batu et al.~\cite{batu2002complexity} give an $(1+\eps)$-multiplicative approximation of the Shannon entropy of $\mathcal{D}$ with oracle complexity $O(\frac{(1+\eps)^2\log^2 n}{\eps^2\cdot H'})$, where $H'$ is a lower bound of the actual entropy $H(\mathcal{D})$.
Guha et al.~\cite{guha2006streaming} improve the oracle complexity to $O(\frac{\log n}{\eps^2\cdot H'})$, matching the lower bound $\Omega(\frac{\log n}{(2+\eps)\eps^2\cdot H'})$ found in~\cite{batu2002complexity}.
Canonne and Rubinfeld~\cite{canonne2014testing} describe a $\Delta$-additive approximation of the Shannon entropy with oracle complexity $O(\frac{\log^2\frac{n}{\Delta}}{\Delta^2})$. Caferov et al.~\cite{caferov2015optimal} show that $\Omega(\frac{\log^2 n}{\Delta^2})$ oracle queries are necessary to get $\Delta$-additive approximation. They also describe a $\Delta$-additive approximation of the \renyi entropy with oracle complexity $O(\frac{n^{1-1/\alpha}}{(1-2^{(1-\alpha)\Delta})^2}\log n)$.
Finally, Obremski and Skorski~\cite{obremski2017renyi} use $O(2^{\frac{\alpha-1}{\alpha}\ren_\alpha(\mathcal{D})}\frac{\log n}{\Delta^2})$ random samples from the unknown distribution $\mathcal{D}$ (\emph{samples-only model}) to get an additive $\Delta$ approximation of the \renyi entropy.
All these algorithms return the correct approximations with constant probability. If we want to guarantee the result with high probability then the sample complexity is multiplied by a $\log n$ factor.

As pointed out earlier, our range S-entropy and R-entropy queries are essentially range colored queries. Next, we discuss known data structures for range colored queries, including range colored counting and reporting.

For range colored counting, the goal is to return the number of colors in $P\cap R$, i.e., $|u(P\cap R)|$. For range colored reporting, the goal is to report all colors in $P\cap R$.
For $d\leq 3$, Gupta et al.~\cite{gupta1995further} study the range colored counting/reporting queries. For $d=1$, where the query range is an interval, they design a data structure for the range colored reporting query with $O(n)$ space and $O(\log n + \mathsf{OUT})$, where $\mathsf{OUT}$ is the output size. For the range colored counting query, the data structure has $O(n)$ space and $O(\log n)$ query time. For $d=2$, where the query range is a rectangle, they derive a data structure for the range colored counting query with $O(n^2\log^2 n)$ space and $O(\log^2 n)$ query time. For the range colored reporting query the data structure has $O(n\log^2 n)$ space and $O(\log n + \mathsf{OUT})$ query time. For $d=3$, where the query range is a box, they design a data structure for the range colored reporting problem with $O(n\log^4 n)$ space and $O(\log^2 n +\mathsf{OUT})$ query time. They extend their result to dynamic data structure and other range queries such as open rectangles.
Chan et al.~\cite{chan2020further} study range colored reporting queries for $d=3$. When the query range is a box, they design a randomized data structure with $O(n\polylog (n))$ space and $O(\mathsf{OUT}\cdot \polylog(n))$ expected query time. 
See~\cite{gupta2018computational} for a survey on range colored queries. 
Kaplan et al.~\cite{kaplan2007counting} study the range colored counting problem for any constant $d\geq 2$. Their data structure has $O(n^d\log^{2d-2}n)$ space and $O(\log^{2d-2})$ query time. More generally, for any threshold parameter $1\leq X\leq n$, they obtain a data structure with $O\left(\frac{n^d}{X^{d-1}}\log^{2d-1}n\right)$ space and $O(X\log^d n + \log^{2d-1}n)$ query time.
Since exact range colored counting queries are generally challenging, there are also papers in the literature~\cite{nekrich2014efficient, rahul2017approximate} proposing near optimal  data structures for approximate range colored counting queries for $d\leq 3$, over various query ranges.
To the best of our knowledge, none of these data structures cannot be extended to handle the more complex range S-entropy and R-entropy queries.

A different type of range colored queries has been studied in~\cite{rahul2009data, rahul2010range}. Give a set $P$ of $n$ (weighted) colored points in $\Re^d$, they design efficient data structures such that, given a query hyper-rectangle $R$, for every color $u_i\in u(P\cap R)$ they report the weighted sum of $P(u_i)\cap R$, the maximum weight of a point in $P(u_i)\cap R$, or the bounding box of $P(u_i)\cap R$.
Finally, a new type of range colored queries, which is related to \emph{data discovery}, have been studied in~\cite{afshani2023range, esmailpour2025theoretical}. More specifically, for $d\leq 3$, Afshani et al.~\cite{afshani2023range} design am efficient data structure such that given a query halfspace (or open box) $R$ and a parameter $\eps\in (0,1)$, it returns all colors that contain at least $\eps\cdot|P\cap R|$ points in $P\cap R$, i.e., it returns a color $u_i$ if $|P(u_i)\cap R|\geq \eps\cdot |P\cap R|$, along with their frequencies with an additive error of $\eps|P\cap R|$. Furthermore, for any constant $d$,  Esmailpour et al.~\cite{esmailpour2025theoretical} design an efficient data structure such that given a query hyper-rectangle $R$ and an interval $\theta\subseteq [0,1]$, it returns all colors whose fraction of points in $R$ lies in $\theta$, i.e., it returns a color $u_i$ if $\frac{|P(u_i)\cap R|}{|P(u_i)|}\in \theta$.
While these queries are more complex than range colored reporting queries, their objectives are fundamentally different than the objectives in S-entropy and R-entropy queries. Furthermore, they focus on reporting colors that satisfy a condition, so in the worst case their query time depends on $|U|=m$. We aim for data structures with sublinear query time with respect to both $n$ and $m$.

\section{Preliminaries}
\label{sec:prelim}
Let $P$ be a set of $n$ colored points in $\Re^d$ and let $P'\subseteq P$.
The Shannon entropy of set $P'$ is defined as
$$H(P')=\sum_{i=1}^m\frac{w(P'(u_i))}{w(P')}\log\left(\frac{w(P')}{w(P'(u_i))}\right),$$
while the \renyi entropy of order $\alpha$ of $P'$ is defined as
$$\ren_{\alpha}(P')=\frac{1}{\alpha-1}\log\left(\frac{1}{\sum_{i=1}^m \left(\frac{w(P'(u_i))}{w(P')}\right)^{\alpha}}\right).$$


For simplicity, and without loss of generality, we can consider throughout the paper that $w(p)=1$ for each point $p\in P$. All the results, proofs, and properties we show hold for the weighted case straightforwardly. Hence, from now on, we assume $w(p)=1$ and the definition of Shannon entropy becomes,
\begin{equation}
\label{eq:def}
H(P')=\sum_{i=1}^m\frac{|P'(u_i)|}{|P'|}\log\left(\frac{|P'|}{|P'(u_i)|}\right)=\sum_{u_i\in u(P')}\frac{|P'(u_i)|}{|P'|}\log\left(\frac{|P'|}{|P'(u_i)|}\right).
\end{equation}
If $|P'(u_i)|=0$, then we consider that $\frac{|P'(u_i)|}{|P'|}\log\left(\frac{|P'|}{|P'(u_i)|}\right)= 0$.

The definition of \renyi entropy becomes,
$$\ren_{\alpha}(P')=\frac{1}{\alpha-1}\log\left(\frac{1}{\sum_{i=1}^m \left(\frac{|P'(u_i)|}{|P'|}\right)^{\alpha}}\right).$$

\paragraph{Updating the Shannon entropy.}
Let $P_1, P_2 \subset P$ be two subsets of $P$ such that $u(P_1)\cap u(P_2)=\emptyset$.
The next formula for the entropy of $P_1\cup P_2$ is known (see~\cite{to2013entropy})
\begin{equation}
\label{eq:entropyupdate}
H(P_1\cup P_2)=\frac{|P_1|H(P_1)+|P_2|H(P_2)+|P_1|\log\left(\frac{|P_1|+|P_2|}{|P_1|}\right)+|P_2|\log\left(\frac{|P_1|+|P_2|}{|P_2|}\right)}{|P_1|+|P_2|}.    
\end{equation}

If $|u(P_2)|=1$ then,
\begin{equation}
\label{eq:entropyupdateinsert}
H(P_1\cup P_2)=\frac{|P_1|H(P_1)}{|P_1|+|P_2|}+ \frac{|P_1|}{|P_1|+|P_2|}\log\left(\frac{|P_1|+|P_2|}{|P_1|}\right)+\frac{|P_2|}{|P_1|+|P_2|}\log\left(\frac{|P_1|+|P_2|}{|P_2|}\right).      
\end{equation}

Finally, if $P_3\subset P_1$ with $|u(P_3)|=1$ and $u(P_1\setminus P_3)\cap u(P_3)=\emptyset$ then
\begin{equation}
\label{eq:entropyupdatedelete}
H(P_1\setminus P_3)=\frac{|P_1|}{|P_1|-|P_3|}\left(H(P_1)-\frac{|P_3|}{|P_1|}\log\frac{|P_1|}{|P_3|}-\frac{|P_1|-|P_3|}{|P_1|}\log \frac{|P_1|}{|P_1|-|P_3|}\right).
\end{equation}
We notice that in all cases, if we know $H(P_1), H(P_2), |P_1|, |P_2|, |P_3|$ we can update the entropy in $O(1)$ time.
\new{If we consider the weighted case, where the points may have different weights, then we replace $|P_1|, |P_2|, |P_3|$ in the formulas with $w(P_1), w(P_2), w(P_3)$, respectively.}

\paragraph{Updating the \renyi entropy.}
Let $P_1, P_2 \subset P$ be two subsets of $P$ such that $u(P_1)\cap u(P_2)=\emptyset$.
The next formula for the \renyi entropy of order $\alpha$ of $P_1\cup P_2$ follows from basic algebraic operations.
\new{For completeness, we show proofs are shown in Appendix~\ref{sec:appndx:renyi}.}
\begin{equation}
\label{eq:Renentropyupdate}
\ren_{\alpha}(P_1\cup P_2)=\frac{1}{\alpha-1}\log\left(\frac{(|P_1|+|P_2|)^{\alpha}}{|P_1|^{\alpha}\cdot 2^{(1-\alpha)\ren_{\alpha}(P_1)}+|P_2|^{\alpha}\cdot 2^{(1-\alpha)\ren_{\alpha}(P_2)}}\right).    
\end{equation}

If $|u(P_2)|=1$ then,
\begin{equation}
\label{eq:Renentropyupdateinsert}
\ren_{\alpha}(P_1\cup P_2)=  \frac{1}{\alpha-1}\log\left(\frac{(|P_1|+|P_2|)^{\alpha}}{|P_1|^{\alpha}\cdot 2^{(1-\alpha)\ren_{\alpha}(P_1)}+|P_2|^{\alpha}}\right).     
\end{equation}

Finally, if $P_3\subset P_1$ with $|u(P_3)|=1$ and $u(P_1\setminus P_3)\cap u(P_3)=\emptyset$ then
\begin{equation}
\label{eq:Renentropyupdatedelete}
\ren_{\alpha}(P_1\setminus P_3)=\frac{1}{\alpha-1}\log\left(\frac{(|P_1|-|P_3|)^{\alpha}}{|P_1|^{\alpha}\cdot 2^{(1-\alpha)\ren_{\alpha}(P_1)}-|P_3|^{\alpha}}\right).  
\end{equation}

We notice that in all cases, if we know $\ren_{\alpha}(P_1), \ren_{\alpha}(P_2), |P_1|, |P_2|, |P_3|$ we can update the \renyi entropy in $O(1)$ time. \new{Similarly to the Shannon entropy, if we consider the weighted case, where the points may have different weights, then we replace $|P_1|, |P_2|, |P_3|$ in the formulas with $w(P_1), w(P_2), w(P_3)$, respectively.}

\paragraph{Range queries.}
In some data structures we need to handle range reporting or range counting problems. Given $P$, we need to construct a data structure such that given a query rectangle $R$, the goal is to return $|R\cap P|$, or report $R\cap P$. We use range trees~\cite{berg1997computational}. 
A range tree can be constructed in $O(n\log^{d})$ time, it has $O(n\log^{d-1}n)$ space and can answer an aggregation query (such as count, sum, max etc.) in $O(\log^{d}n)$ time. A range tree can be used to report $R\cap P$ in $O(\log^{d}n + |R\cap P|)$ time. Using \emph{fractional cascading} the $\log^d n$ term can be improved to $\log^{d-1} n$ in the query time. \new{However, for simplicity, we consider the simple version of a range tree without using fractional cascading. In this way, it is easy to extend to the weighted case of the problem where fractional cascading is not applied}.
Furthermore, a range tree can be used to return a uniform sample point from $R\cap P$ in $O(\log^{d} n)$ time. We give more details about range trees and sampling in the next paragraph.
There is also lot of work on designing data structures for returning $k$ independent samples in a query range efficiently~\cite{martinez2020parallel, tao2022algorithmic, wang2015spatial, xie2021spatial, afshani2017independent, afshani2019independent, hu2014independent}. For example, if the input is a set of points in $\Re^d$ and the query range is a query hyper-rectangle, then there exists a data structure~\cite{martinez2020parallel} with space $O(n\log^{d-1}n)$ and query time $O(\log^d n + k\log n)$. For our purposes, it is sufficient to run $k$ independent sampling queries in a (modified) range tree with total query time $O(k\log^d n)$.

\paragraph{Range tree and sampling.}
Next, we formally describe the construction of the range tree and we show how it can be used for range sampling queries.

For $d=1$, the range tree on $P$ is a balanced binary search tree $T$ of $O(\log n)$ height. The points of $P$ are stored at the leaves of $T$ in increasing order, while each internal node $v$ stores the smallest and the largest values/coordinates, $\alpha_v^-$ and $\alpha_v^+$, respectively, contained in its subtree.
The node $v$ is associated with an interval $I_v=[\alpha_v^-, \alpha_v^+]$ and the subset $P_v=I_v\cap P$.
For $d>1$, $T$ is constructed recursively: 
We build a $1$D range tree $T_d$ on the $x_d$-coordinates of points in $P$. Next, for each node $v\in T_d$, we recursively construct a $(d-1)$-dimensional range tree $T_v$ on $P_v$, which is defined as the projection of $P_v$ onto the hyperplane $x_d=0$, and attach $T_v$ to $v$ as its secondary tree. The size of $T$ in $\Re^d$ is $O(n\log^{d-1} n)$ and it can be constructed in $O(n\log^d n)$ time.

For a node $v$ at a level-$i$ tree, let $p(v)$ denote its parents in that tree. If $v$ is the root of that tree, $p(v)$ is undefined.
For each node $v$ of the $d$-th level of $T$, we associate a $d$-tuple $\langle v_1, v_2, \ldots, v_d=u\rangle$, where $v_i$ is the node at the $i$-th level tree of $T$ to which the level-$(i+1)$ tree containing $v_{i+1}$ is connected.
We associate the rectangle $\square_v=\prod_{j=1}^d I_{v_j}$ with the node $v$.
For a rectangle $R=\prod_{i=1}^d \delta_i$
, a $d$-level node $v$ is called a \emph{canonical node} if for every $i\in [1,d]$, $I_{v_i}\subseteq \delta_i$ and $I_{p(v_i)}\not\subseteq \delta_i$.
For any rectangle $R$, there are $O(\log^d n)$ canonical nodes in $\T$, denoted by $\mathcal{N}(R)$, and they can be computed in $O(\log^d n)$ time~\cite{bentley1978decomposable,de1997computational,lueker1978data, agarwal2017range, agarwal1999geometric}.
$\T$ can be maintained dynamically, as points are inserted into $P$ or deleted from $P$ using the standard partial-reconstruction method, which periodically reconstructs various bottom subtrees. The amortized time is $O(\log^d n)$; see~\cite{overmars1983design} for details.

A range tree can be used to answer range (rectangular) aggregation queries, such as range counting queries, in $O(\log^d n)$ time and range reporting queries in $O(\log^d n + K)$ time, where $K$ is the output size. The query time can be improved to $O(\log^{d-1} n)$ using fractional cascading. See~\cite{lueker1978data,de1997computational, agarwal1999geometric} for details. However, for simplicity, in this work we use the simpler version of it with the term $\log^d n$ in the query time.

\new{
A range tree can be used to return a uniform sample in a query rectangle. More formally, the goal is to construct a data structure such that given a query rectangle $R$, a uniform sample in $P\cap R$ is returned in $O(\log^d n)$ time. We construct a standard range tree $T$ on the point set $P$. For each $d$-level node $v$ of the tree we precompute and store $c(v)=|P\cap \square_v|$, i.e., the number of points stored in the subtree with root $v$. The space of $T$ remains $O(n\log^{d-1}n)$ and the construction time $O(n\log^d n)$. We are given a query rectangle $R$. We run the query procedure in the range tree $T$ and we find the set of canonical nodes $\mathcal{N}(R)$. For each node $v\in \mathcal{N}(R)$, we define the weight $w_v=\frac{c(v)}{\sum_{v'\in \mathcal{N}(R)} c(v')}$. We sample one node from $\mathcal{N}(R)$ with respect to weights $\{w_v\mid v\in \mathcal{N}(R)\}$, using reservoir sampling~\cite{efraimidis2006weighted}. 
Let $v$ be the node that is sampled.
If $v$ is a leaf node then we return the point that is stored in node $v$. Otherwise, assume that $v$ has two children $x, y$. We move to the node $x$ with probability $\frac{c(x)}{c(x)+c(y)}$ and to node $y$ with probability $\frac{c(y)}{c(x)+c(y)}$. We recursively repeat this process until we reach a leaf node of the range tree. We return the point stored in the leaf node.

\textit{Analysis.}
As we discussed above, we can get the set $\mathcal{N}(R)$ in $O(\log^d n)$ time. Then, we sample one node from $\mathcal{N}(R)$ in $O(\log^d n)$ time using reservoir sampling. Finally, the recursive method takes $O(\log n)$ time because the height of the level-$d$ tree is $O(\log n)$. Overall, the query procedure takes $O(\log^d n)$ time.

Next, we show that the sampled point is chosen uniformly at random, i.e., with probability $\frac{1}{|P\cap R|}$. 
Let $v\rightarrow v_1\rightarrow\ldots\rightarrow v_k$ be the path of nodes followed by the algorithm to sample a point $p$. Thus $p$ is stored in the leaf node $v_k$. Let $\bar{v}_1,\ldots, \bar{u}_k$ be the siblings of nodes $v_1, \ldots, v_k$, respectively. The probability that $p$ is selected is 
\vspace{-0.1em}
$$\frac{c(v)}{\sum_{v'\in \mathcal{N}(R)}c(v')}\cdot \frac{c(v_1)}{c(v_1)+c(\bar{v}_1)}\cdot \ldots\cdot \frac{c(v_k)}{c(v_k)+c(\bar{v}_k)}.$$
Notice that $c(v)=c(v_1)+c(\bar{v}_1)$ and $c(v_\ell)=c(v_{\ell+1})+c(\bar{v}_{\ell+1})$ for every $\ell\in[k-1]$. Furthermore $c(v_k)=1$ because $v_k$ is a leaf node.
We conclude that the probability of selecting $p$ is $\frac{1}{\sum_{v'\in\mathcal{N}(R)}c(v')}=\frac{1}{|P\cap R|}$.


\textit{Extension to sampling on weighted points.}
Given a set of weighted points, the range tree can be used to sample a point from $P\cap R$ with respect to their weights. Assume that each point $p\in P$ has a weight $w(p)$, which is a non-negative real number. Given a query hyper-rectangle $R$ the goal is to sample a point from $P\cap R$ with respect to their weight, i.e., a point $p\in P\cap R$ should be selected with probability $\frac{w(p)}{\sum_{p'\in P\cap R}w(p')}$. The construction is exactly the same as in the unweighted case. The only difference is that instead of storing the count $c(v)$ in each node $v$, we store $w(v)=\sum_{p'\in P\cap \square_v}w(p')$. The query time remains $O(\log^d n)$ and the correctness proof remains the same replacing $c(v)$ with $w(v)$, for each node $v$ of the range tree.
}

\paragraph{Range trees for S-entropy and R-entropy queries in $\O(m)$ query time.} The range tree can be used to design a near-linear space data structure for the range S-entropy and R-entropy query problem having $O(m\log^d n)$ query time. For every color $u_i\in U$ construct a range tree $\mathcal{T}_i$ on $P(u_i)$ for counting queries. Furthermore, construct a range tree $\mathcal{T}$ on $P$  for counting queries. Given a query rectangle $R$, for every color $u_i\in U$, we use $\mathcal{T}_i$ to get $|P(u_i)\cap R|$. We also use $\mathcal{T}$ to get $|P\cap R|$. These $m+1$ quantities are sufficient to compute $H(P\cap R)$ or $\ren_\alpha(P\cap R)$ in $O(m)$ additional time. The data structure uses $O(n\log^d n)$ space, but the query time is $O(m\log^d n)$. This data structure is sufficient if $m$ is small, for example $m=\polylog (n)$. However, this is not an efficient data structure because in the worst case $m=O(n)$. In this work, we focus on low space (ideally, near-linear space) data structures for the range S-entropy and R-entropy queries in strictly sublinear query time.

\paragraph{Expected Shannon entropy and monotonicity.}
Shannon (and \renyi) entropy is not monotone because if $P_1\subseteq P_2$, it does not always hold that $H(P_1)\leq H(P_2)$. Using the results in~\cite{li2004entropy}, we can show that $H(P_1)\geq \frac{|P_1|-1}{|P_1|}H(P_1\setminus\{p\})$, for a point $p\in P_1\subseteq P$. If we multiply with $|P_1|/n$ we have
$\frac{|P_1|}{n}H(P_1)\geq \frac{|P_1|-1}{n}H(P_1\setminus\{p\})$. Hence, we show that, for $P_1\subseteq P_2\subseteq P$,
$
\frac{|P_1|}{n}H(P_1)\leq \frac{|P_2|}{n}H(P_2)$.
The quantity $\frac{|P_1|}{|P|}H(P_1)$ is called expected Shannon entropy.
This monotonicity property helps us to design efficient partitioning algorithms with respect to expected entropy, for example, find a partitioning that minimizes the cumulative or maximum expected entropy.
\section{Lower Bounds}
\label{sec:lowerbound}
In this section, we show conditional lower bounds for range S-entropy and range R-entropy data structures in the real-RAM model. First, we show a connection to the matrix multiplication problem to study the tradeoff between the preprocessing and query time. Then, we show a connection to the set intersection problem to study the tradeoff between the query time and the space used.

\subsection{Preprocessing-query tradeoff}\label{subsec:query-tradeoff} We show a connection between the boolean matrix multiplication problem and range entropy queries.
We get our intuition from~\cite{Chan2014}, designing data structures for range mode queries, which are different from range S-entropy and R-entropy queries.
By making this connection, we show that it is unlikely to have a data structure for answering range entropy queries that has a near-linear preprocessing time and answers the queries in polylogarithmic time even for $d = 1$ ($1$-dimensional space). 

\begin{figure}[t]
\centering
\begin{tikzpicture}[scale=0.65, every node/.style={scale=0.6}]

\foreach \x in {0, 1, 2} {
    \foreach \y in {0, 1, 2} {
        \draw (\x, \y) rectangle (\x+1, \y+1);
    }
}
\node at (0.5, 2.5) {1};
\node at (1.5, 2.5) {0};
\node at (2.5, 2.5) {1};
\node at (0.5, 1.5) {0};
\node at (1.5, 1.5) {1};
\node at (2.5, 1.5) {0};
\node at (0.5, 0.5) {1};
\node at (1.5, 0.5) {1};
\node at (2.5, 0.5) {0};

\fill[red] (0.5, 3.2) circle (0.1); 
\fill[green] (1.5, 3.2) circle (0.1); 
\fill[blue] (2.5, 3.2) circle (0.1); 

\node at (1.5, -0.5) {$A$};

\foreach \x in {4, 5, 6} {
    \foreach \y in {0, 1, 2} {
        \draw (\x, \y) rectangle (\x+1, \y+1);
    }
}
\node at (4.5, 2.5) {0};
\node at (5.5, 2.5) {1};
\node at (6.5, 2.5) {1};
\node at (4.5, 1.5) {1};
\node at (5.5, 1.5) {1};
\node at (6.5, 1.5) {1};
\node at (4.5, 0.5) {0};
\node at (5.5, 0.5) {0};
\node at (6.5, 0.5) {0};

\fill[red] (7.2, 2.5) circle (0.1); 
\fill[green] (7.2, 1.5) circle (0.1); 
\fill[blue] (7.2, 0.5) circle (0.1); 

\node at (5.5, -0.5) {$B$};

\draw[->] (8, 1) -- (22.9, 1); 
\foreach \x [evaluate=\x as \pos using 8 + \x * 0.8] in {1, 2, ..., 18} {
    \draw (\pos, 0.9) -- (\pos, 1.1); 
    \node[below] at (\pos, 0.8) {\x}; 
}

\fill[green] (8.8, 1.1) circle (0.1);  
\fill[red] (9.6, 1.1) circle (0.1); 
\fill[blue] (10.4, 1.1) circle (0.1); 
\fill[red] (11.2, 1.1) circle (0.1); 
\fill[blue] (12.0, 1.1) circle (0.1); 
\fill[green] (12.8, 1.1) circle (0.1); 
\fill[blue] (13.6, 1.1) circle (0.1); 
\fill[red] (14.4, 1.1) circle (0.1); 
\fill[green] (15.2, 1.1) circle (0.1); 

\fill[green] (16.0, 1.1) circle (0.1); 
\fill[red] (16.8, 1.1) circle (0.1); 
\fill[blue] (17.6, 1.1) circle (0.1); 
\fill[red] (18.4, 1.1) circle (0.1); 
\fill[green] (19.2, 1.1) circle (0.1); 
\fill[blue] (20.0, 1.1) circle (0.1); 
\fill[red] (20.8, 1.1) circle (0.1); 
\fill[green] (21.6, 1.1) circle (0.1); 
\fill[blue] (22.4, 1.1) circle (0.1); 

\foreach \i [evaluate=\i as \x using 8 + (3 * \i + 0.5) * 0.8] in {0, 1, 2, 4, 5, 6} {
    \draw[black, dashed] (\x, 0.7) -- (\x, 1.3); 
}
\draw[red, thick, dashed] (15.6, 0.7) -- (15.6, 1.3);

\node at (15.6, -0.3) {$P$};

\draw[thick] (12.5, 1.1) -- (12.5, 1.6) -- (19.5, 1.6) -- (19.5, 1.1); 

\node at (16.0, 2) {$\rho_{2,2}$};

\end{tikzpicture}
\caption{An example of constructing the point set $P$ on the right based on two sample $3 \times 3$ matrices $A$ and $B$ on the left. The colors red, green, and blue represent colors $1$, $2$, and $3$, respectively. Points from $1$ to $9$, represent the points in $A_1A_2A_3$ corresponding to the rows of $A$, and the points from $10$ to $18$ represent the points in $B_1B_2B_3$ corresponding to the columns of $B$. Blocks are separated by vertical dashed lines. The interval $\rho_{2, 2}$ which is used to find the entry $c_{2,2}$ contains the points $6$ to $14$ as shown.}
    \label{fig:matrixaxis}
\end{figure}

We consider the boolean matrix multiplication problem.
Let $A$ and $B$ be two $\sqrt{n} \times \sqrt{n}$ boolean matrices and the goal is to compute the product $C = A \cdot B$. 
We show that the matrix $C$ can be computed using a range entropy query data structure over a set $P \subset \Re^1$ of $2n$ points using $\sqrt{n}$ colors. Observe that the entry $c_{i,j}\in C$ is $1$ if and only if there exists at least one index $k$ such that $a_{ik} = b_{kj} = 1$. Our goal is to first build $P$ and then find each entry $c_{i,j}$ using a single query to the data structure. 

For each $i \in [\sqrt{n}]$, we build an array of points $A_i$, containing exactly $\sqrt{n}$ points and color them based on the entries in the $i$'th row of the matrix $A$. We build each $A_i$, such that any point in $A_{i+1}$, has a larger coordinate than any point in $A_i$, for all $i \in [\sqrt{n} - 1]$. Moreover, we assume that the points in each $A_i$ are sorted based on their coordinates. For each $i$, let $Z_i = \{j | a_{i,j} = 0\}$, be the set of indices of $0$ values in the $i$'th row of $A$. Let $U = [\sqrt{n}]$ be our set of colors. We color the first $|Z_i|$ points in $A_i$, using the colors from $Z_i$ in an arbitrary order. We color the remaining $\sqrt{n} - |Z_i|$ points in $A_i$ using the colors from $U - Z_i$ in an arbitrary order. Note that during this coloring we use each color in $U$ exactly once. Intuitively, for each $A_i$, we color the first points using the indices of $0$ values of the $i$'th row and the remaining points using the indices of $1$ values. 

Similarly, for each $i \in [\sqrt{n}]$, we build an array of $\sqrt{n}$ points $B_i$ and color them based on the $i$'th column of $B$. We set the coordinates such that any point in $B_1$ has a larger coordinate than $A_{\sqrt{n}}$, and any point  in
$B_{i+1}$ has a larger coordinate than every point in $B_i$, for all $i \in [\sqrt{n} - 1]$. This time, we color the first points in each $B_i$ using the $1$ values and the remaining points based on the $0$ values from the $i$'th column of $B$. More formally, let $O_i = \{j | b_{ji} = 1\}$, be the set of indices of $1$ values in the $i$'th column of $B$. We color the first $\sqrt{n}$ points in $B_i$, using the colors from $O_i$ in an arbitrary order. We color the remaining $\sqrt{n} - |O_i|$ points in $B_i$ using the colors from $U - O_i$ in an arbitrary order. 

We refer to each of these constructed arrays $A_i$ and $B_i$ as \textit{blocks} and denote the $j$'th point in $A_i$ ($B_i$) by $A_{i}[j]$ ($B_{i}[j]$). We set the point set $P$ to be $A_1 A_2 \dots A_{\sqrt{n}} B_1 B_2 \dots B_{\sqrt{n}}$, the concatenation of the points in all the blocks. Note that by the construction of the blocks, the points in $P$ are sorted based on their coordinate. An example of this construction based on two sample matrices $A$ and $B$ is shown in Figure~\ref{fig:matrixaxis}.

We construct the range entropy data structure $\mathcal{D}$ over $P$. We first describe how we can find each entry $c_{i,j}$ using a single range S-entropy query and later show how we can do it by a single range R-entropy query. To compute $c_{i,j}$ we set the interval $\rho_{i,j} = [A_i[|Z_i| + 1], B_j[|O_j|]]$ and query the data structure to return $\mathcal{D}(\rho_{i,j})$.
Let $H_{i,j}$ denote the returned answer, which is the S-entropy of the points $P \cap \rho_{i,j}$. Observe that by this coloring, the entry $c_{i,j}$ is $1$ if and only if the last $|\sqrt{n} - Z_i|$ points of $A_i$ share a common color with the first $|O_j|$ points of $B_j$. An example is shown in Figure~\ref{fig:matrixaxis}.

Let $t$ denote the number of blocks that lie completely inside $\rho_{i,j}$, and let $P_1 = P \cap A_i$ and $P_2 = P \cap B_j$. We define the value $H'_{i,j}$ as follows:
\begin{align*}
    H'_{i,j} &= (|P_1| + |P_2|)\left(\frac{t+1}{t\sqrt{n} + |P_1| + |P_2|}\log \left( \frac{t\sqrt{n} + |P_1| + |P_2|}{{t+1}}\right) \right)
    \\ &+ (\sqrt{n} - |P_1| - |P_2|)\left(\frac{t}{t\sqrt{n} + |P_1| + |P_2|}\log\left(\frac{t\sqrt{n} + |P_1| + |P_2|}{{t}}\right)\right).
\end{align*}

It is straightforward to see that we can compute $H'_{i,j}$ in constant time since all the parameters are known. 

\begin{lem}\label{lem:s-matrix}
    In the preceding reduction, $c_{i,j} = 0$ if and only if $H_{i,j} = H'_{i,j}$.
\end{lem}
\begin{proof}
We first note that $H_{i,j}'$ is the Shannon entropy  of the points in $P\cap \rho_{i,j}$ assuming that $u(P_1)\cap u(P_2)=\emptyset$, or equivalently, $c_{i,j}=0$. Indeed, if $c_{i,j}=0$ then $u(P_1)\cap u(P_2)=\emptyset$, and there are $|P_1|+|P_2|$ colors with $t+1$ points and $\sqrt{n}-|P_1|-|P_2|$ colors with $t$ points in $P\cap \rho_{i,j}$, while $|P\cap \rho_{i,j}|=t\sqrt{n}+|P_1|+|P_2|$.
Next, we focus on the other direction assuming that $c_{i,j}=1$. In this case $u(P_1)\cap u(P_2)\neq \emptyset$.
Intuitively, this creates a distribution with lower uncertainty, so the entropy should be decreased. 
In the value $H_{i,j}'$, there are two colors, say $u_1\in u(P_1)$ and $u_2\in u(P_2)$, such that each of them contributed $\frac{t+1}{N}\log\frac{N}{t+1}$ in the Shannon entropy. Next, assume that $u_1$ has $t$ points, while $u_2$ has $t+2$ points. It is sufficient to show that $2\frac{t+1}{N}\log\frac{N}{t+1}>\frac{t}{N}\log\frac{N}{t}+\frac{t+2}{N}\log\frac{N}{t+2}\Leftrightarrow t\log (t) + (t+2)\log(t+2) -2(t+1)\log(t+1)>0$. The function $f(t)=t\log (t) + (t+2)\log(t+2) -2(t+1)\log(t+1)$ is decreasing for $t\geq 0$, $\lim_{t\rightarrow \infty}f(t)=0$ and $\lim_{t\rightarrow 0}f(t)=\infty$, so $f(t)>0$. The result follows.
\end{proof}

By the lemma above, we can report $c_{i,j}$ by comparing the answer received from $\mathcal{D}(\rho_{i,j}) = H_{i,j}$ and $H'_{i,j}$, and hence we can compute the matrix product $C = A\cdot B$, by making $n$ queries to $\mathcal{D}$. Furthermore, we can build the point set $P$ in $O(n)$ time. 

\paragraph{Extension to range R-entropy query.} We use the same reduction as for S-entropy queries. However, we set $\mathcal{D}$ to be a range R-entropy data structure and denote the order $\alpha$ R-entropy of the points in $\rho_{i,j} \cap P$ by $H^{\alpha}_{i,j}$. We define the value $H'^{\alpha}_{i,j}$ as follows: 
$$ H'^{\alpha}_{i,j} = \frac{1}{\alpha-1}\log\left(\frac{1}{(|P_1|+|P_2|)\left(\frac{t+1}{t\sqrt{n} + |P_1| + |P_2|}\right)^\alpha+(\sqrt{n} - |P_1| - |P_2|)\left(\frac{t}{t\sqrt{n} + |P_1| + |P_2|}\right)^\alpha}\right).$$



\begin{lem}\label{lem:r-matrix}
    In the preceding reduction, $c_{i,j} = 0$ if and only if $H^{\alpha}_{i,j} = H'^{\alpha}_{i,j}$.
\end{lem}
\begin{proof}
We first note that $H_{i,j}'^\alpha$ is the \renyi entropy  of the points in $P\cap \rho_{i,j}$ assuming that $u(P_1)\cap u(P_2)=\emptyset$, or equivalently, $c_{i,j}=0$. Indeed, if $c_{i,j}=0$ then $u(P_1)\cap u(P_2)=\emptyset$, and there are $|P_1|+|P_2|$ colors with $t+1$ points and $\sqrt{n}-|P_1|-|P_2|$ colors with $t$ points in $P\cap \rho_{i,j}$, while $|P\cap \rho_{i,j}|=t\sqrt{n}+|P_1|+|P_2|$.
Next, we focus on the other direction assuming that $c_{i,j}=1$. In this case $u(P_1)\cap u(P_2)\neq \emptyset$.
Intuitively, this creates a distribution with lower uncertainty, so the entropy should be decreased. 
In the value $H_{i,j}'^\alpha$, there are two colors, say $u_1\in u(P_1)$ and $u_2\in u(P_2)$, such that each of them contributed $\left(\frac{N}{t+1}\right)^\alpha$ in the $\log(\cdot)$ function of the \renyi entropy. Next, assume that $u_1$ has $t$ points, while $u_2$ has $t+2$ points. It is sufficient to show that 
$\frac{1}{2\left(\frac{t+1}{N}\right)^\alpha}>\frac{1}{\left(\frac{t}{N}\right)^\alpha + \left(\frac{t+2}{N}\right)^\alpha}$
or equivalently 
$t^\alpha + (t+2)^\alpha> 2(t+1)^\alpha$. Indeed, for $t\geq 0$, 
the function $f(t)=t^\alpha + (t+2)^\alpha- 2(t+1)^\alpha$ is 
i) increasing for $\alpha>2$ with $f(0)>0$ and $\lim_{t\rightarrow\infty} f(t)=\infty$, ii) $f(t)=2$ for $\alpha=2$, and iii) decreasing for $\alpha\in(1,2)$ with $\lim_{t\rightarrow 0}f(t)=\infty$ and $\lim_{t\rightarrow \infty}f(t)=0$.
The result follows.
\end{proof}

Thus, with the same argument as for S-entropy queries, we conclude with the following theorem. 

\new{
\begin{thm}\label{theorem:r-matrix}
Let $\matrixmul(\sqrt{n})$ be the running time of the optimum algorithm to multiple two $\sqrt{n}\times \sqrt{n}$ boolean matrices.
Any data structure for range R-entropy (resp. S-entropy) queries over $n$ points in $\Re^d$, for $d\geq 1$,
with $Q(n)$ query time must have $\Omega(\max\{\matrixmul(\sqrt{n})-n\cdot Q(n),1\})$ preprocessing time.



\end{thm}
}

\paragraph{Interpretation.}
There has been extensive work in the theory community studying lower bounds and designing algorithms for the problem of multiplying two boolean matrices. The results can be partitioned into two groups, combinatorial algorithms and algebraic algorithms.

For the problem of multiplying boolean matrices, there exists a well-known conjecture \cite{satta1994tree, lee2002fast} that no combinatorial algorithm\footnote{Combinatorial algorithms reduce redundancy in computations by exploiting the combinatorial properties of Boolean matrices. The formal definition of a combinatorial algorithm is an open problem~\cite{yu2018improved}.} with running time $O(n^{3-\eps})$ exists to multiply two $n\times n$ boolean matrices, for any positive value of $\eps<1$. A discussion about this pessimistic lower bound can be found in~\cite{yu2018improved, abboud2024new}. If this conditional lower bound does not hold, then we would have faster combinatorial algorithms for multiple fundamental discrete problems. Using this conditional lower bound and Theorem~\ref{theorem:r-matrix}, we get that any data structure for the range S-entropy or R-entropy query over $n$ points, with $O(n^{0.5-\eps})$ query time requires $\Omega(n^{1.5-\eps})$ preprocessing time, for any positive $\eps<1$.

On the other hand, there exist faster algebraic algorithms for multiplying two boolean matrices since they rely on the structure of the field, and in the ring structure of matrices over the field. Multiplying two $n\times n$ boolean matrices can be done in $O(n^{\omega})$ time for some value of $\omega\geq 2$. Currently, the best algebraic algorithm for this problem runs in $O(n^{\omega})$ time for $\omega=2.371552$~\cite{williams2024new}. Assuming that the optimum algorithm runs in $O(n^\omega)$ time for a value $\omega>2$, we can argue that any data structure for the range S-entropy or R-entropy query over $n$ points, with $O(n^{\omega/2-1-\eps})$ query time requires $\Omega(n^{\omega/2})$ preprocessing time, for any positive $\eps<\omega/2-1<1$.
Interestingly, the
only non-trivial (algebraic) lower bound for the matrix multiplication problem of two $n\times n$ boolean matrices
is $\Omega(n^2\log n)$. In this case, we can argue that any data structure for the range S-entropy or R-entropy query over $n$ points, with $O(\log^{1-\eps} n)$ query time requires $\Omega(n\log n)$ preprocessing time, for any positive $\eps<1$.

\subsection{Space-query tradeoff}\label{subsec:space-tradeoff}
Next, we show a reduction from the set intersection problem to range entropy problems. First, we show lower bounds for $d \geq 2$. At the end, we show that range entropy data structures with near-linear space and polylogarithmic query time are unlikely to exist even for $d=1$.

The set intersection problem is defined as follows. Given a family of sets $S_1, \ldots, S_g$, with $\sum_{i=1}^g|S_i|=n$, the goal is to construct a data structure such that given a query pair of indices $i, j$, it decides if $S_i\cap S_j=\emptyset$. It is widely believed that for any positive value $Q\in \Re$, any data structure for the set intersection problem with $O(Q)$ query time needs $\widetilde{\Omega}\left(\left(\frac{n}{Q}\right)^2\right)$ space~\cite{davoodi2012two, patrascu2010distance, rahul2012algorithms}. we call it the \emph{set intersection conjecture}.
Next, we show that any data structure for solving the range S-entropy query can be used to solve the set intersection problem. In the end, we extend the reduction to the range R-entropy query.

Let $S_1,\ldots, S_g$ be an instance of the set intersection problem as we defined above. We design an instance of the range entropy query constructing a set $P$ of $2n$ points in $\Re^2$ and $|U|=|\bigcup_i S_i|$.
Let $n_0=0$ and $n_i=n_{i-1}+|S_i|$ for $i=1,\ldots, g$.
Let $s_{i,k}$ be the value of the $k$-th item in $S_i$ (we consider any arbitrary order of the items in each $S_i$).
Let $S=\bigcup_i S_i$, and $q=|S|$. Let $\sigma_1,\ldots\sigma_q$ be an arbitrary ordering of $S$. We set $U=\{1,\ldots, q\}$.
Next, we create a geometric instance of $P$ in $\Re^2$: All points lie on two parallel lines $L=x+n$, and $L'=x-n$. For each $s_{i,k}$ we add in $P$ two points,
$p_{i,k}=(-(k+n_{i-1}), -(k+n_{i-1})+n)$ on $L$, and $p_{i,k}'=((k+n_{i-1}), k+n_{i-1}-n)$ on $L'$. If $s_{i,k}=\sigma_j$ for some $j\leq q$, we set the color/category of both points $p_{i,k}, p_{i,k}'$ to be $j$. Let $P_i$ be the set of points corresponding to $S_i$ that lie on $L$, and $P_i'$ the set of points corresponding to $S_i$ that lie on $L'$. We set $P=\bigcup_i (P_i\cup P_i')$. We note that for any pair $i, j$, points $P_i \cup P_j'$ have distinct categories if and only if $S_i\cap S_j=\emptyset$.
$P$ uses $O(n)$ space and can be constructed in $O(n)$ time.

Let $\mathcal{D}$ be a data structure for range entropy queries with space $S(n)$ and query time $Q(n)$ constructed on $n$ points. Given an instance of the set intersection problem, we construct $P$ as described above. Then we build $\mathcal{D}$ on $P$ and we construct a range tree $\mathcal{T}$ on $P$ for range counting queries.
Given a pair of indexes $i, j$ the question is if $S_i\cap S_j=\emptyset$. We answer this question using $\mathcal{D}$ and $\mathcal{T}$ on $P$. Geometrically, it is known that we can find a rectangle $\rho_{i,j}$ in $O(1)$ time such that $\rho_{i,j}\cap P=P_i\cup P_j'$ (see Figure~\ref{fig:lowerBound}). We run the range entropy query $\mathcal{D}(\rho_{i,j})$ and the range counting query $\mathcal{T}(\rho_{i,j})$. Let $H_{i,j}$ be the entropy of $P_i\cup P_j'$ and $n_{i,j}=|P_i\cup P_j'|$. If $H_{i,j}=\log n_{i,j}$ we return that $S_i\cap S_j=\emptyset$. Otherwise, we return $S_i\cap S_j\neq \emptyset$.

          \begin{figure}[H]
              \vspace{-1em}\includegraphics[width=0.4\linewidth]{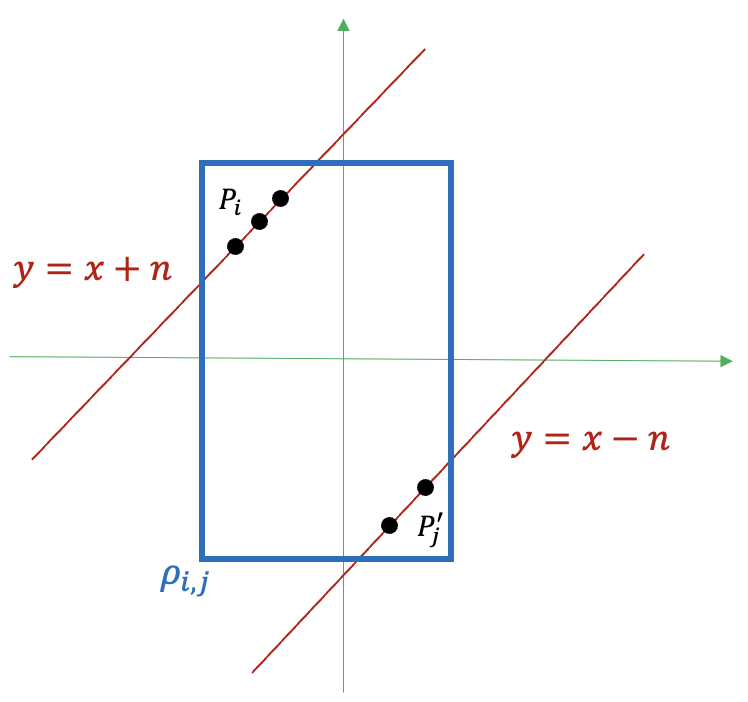}
              \caption{Lower bound construction.\label{fig:lowerBound}}
          \end{figure}

The data structure we construct for answering the set intersection problem has $O(S(2n)+n\log n)=\widetilde{O}(S(2n))$ space. The query time is $(Q(2n)+\log n)$ or just $O(Q(n))$ assuming that $Q(n)\geq \log n$.
\begin{lem}
\label{lem:lb}
In the preceding reduction, $S_i\cap S_j=\emptyset$ if and only if $H_{i,j}=\log n_{i,j}$.
\end{lem}
\begin{proof}
If $S_i\cap S_j=\emptyset$ then from the construction of $P$ we have that all colors in $P_i\cup P_j'$ are distinct, so $n_{i,j}=|u(P_i\cup P_j')|$. Hence, the entropy $H(P_i\cup P_j')$ takes the maximum possible value which is $H(P_i\cup P_j')=\sum_{v\in u(P_i\cup P_j')}\frac{1}{n_{i,j}}\log n_{i,j}=\log n_{i,j}$.

If $H_{i,j}\neq\log n_{i,j}$ we show that $S_i\cap S_j\neq\emptyset$. The maximum value that $H_{i,j}$ can take is $\log n_{i,j}$ so we have $H_{i,j}<\log n_{i,j}$. The entropy is a measure of uncertainty of a distribution. It is known that the discrete distribution with the maximum entropy is unique and it is the uniform distribution. Any other discrete distribution has entropy less than $\log n_{i,j}$. Hence the result follows.
\end{proof}

\paragraph{Extension to range R-entropy query.}
Following the same reduction, we can show that Lemma~\ref{lem:lb} also holds for the \renyi entropy of any parameter $\alpha>0$.
Let $\ren_{\alpha}(P_i\cup P_j')$ be the \renyi entropy 
(of any order $\alpha$) of $P_i\cup P_j'$ and $n_{i,j}=|P_i\cup P_j'|$.
\begin{lem}
    In the preceding reduction, for any parameter $\alpha>0$ such that $\alpha\neq 1$, $S_i\cap S_j=\emptyset$ if and only if $\ren_{\alpha}(P_i\cap P_j')=\log n_{i,j}$.
\end{lem}
\begin{proof}
    If $S_i\cap S_j=\emptyset$ then from the construction of $P$ we have that all colors in $P_i\cup P_j'$ are distinct, so $n_{i,j}=|u(P_i\cup P_j')|=|P_i|+|P_j'|$.
    It is known that the \renyi entropy of any order $\alpha>0$ is Schur concave so its optimum value is always achieved for the uniform distribution. Hence, $\ren_{\alpha}(P_i\cup P_j')=\log n_{i,j}$.

If $\ren_{\alpha}(P_i\cup P_j')\neq\log n_{i,j}$ then $S_i\cap S_j\neq\emptyset$.
Since $\ren_{\alpha}(P_i\cup P_j')\neq\log n_{i,j}$
it must be the case that $u(P_i\cup P_j')<n_{i,j}$ (the maximum value for the \renyi entropy is only achieved for the uniform distribution). Hence, there is at least a common color between the points in $P_i$ and $P_j'$, implying that $S_i\cap S_j\neq \emptyset$.
\end{proof}

We also conclude to the next theorem.
\begin{thm}
\label{thm:Renlb}
If there is a data structure for range R-entropy (resp. S-entropy) queries in dimension $d\geq 2$, with $S(n)$ space and $Q(n)$ query time, then for the set intersection problem there exists a data structure with $\O(S(2n))$ space and $\O(Q(2n))$ query time.
\end{thm}

\paragraph{Interpretation.}
Using the set intersection conjecture, we can also conclude that any data structure for the range S-entropy or R-entropy query over $n$ points with $Q(n)$ query time must have $\widetilde{\Omega}\left(\left(\frac{n}{Q(n)}\right)^2\right)$ space. For example, if the designed data structure for the range S-entropy (or R-entropy) query has $\polylog (n)$ query time, then the space should be $\widetilde{\Omega}(n^{2})$. Similarly, if the query time is $n^{0.25}$, then the space should be $\widetilde{\Omega}(n^{1.5})$. 

\new{
\begin{cor}\label{cor:prep-tradeoff}
 If the set intersection conjecture is true, then any data structure for range R-entropy (resp. S-entropy) queries over $n$ points in $\Re^d$, for $d\geq 2$, with $Q(n)$ query time must use $\widetilde{\Omega}\left(\left(\frac{n}{Q(n)}\right)^2\right)$ space.
\end{cor}
}

\new{
\paragraph{Space-query tradeoff for $d=1$.} Using the same ideas as in Subsection~\ref{subsec:query-tradeoff} and \cite{goldstein_set_intersection}, we show that we can obtain a weaker version of Corollary~\ref{cor:prep-tradeoff}, for the one-dimensional case, $d = 1$. While being a weaker lower bound than for the case $d \geq 2$, this still suggests that a data structure with near-linear space and polylogarithmic query time is unlikely to exist even for $d = 1$.
We use a similar reduction as in Subsection~\ref{subsec:query-tradeoff}, but instead of the matrix multiplication, we start from a set intersection instance. Given a family of sets $S_1, \ldots, S_g$, the goal is to construct a data structure such that given a query pair of indices $i, j$, it decides if $S_i\cap S_j=\emptyset$. Based on the given family of sets, we build a range entropy query data structure such that we can answer any set intersection query using a single query to the constructed data structure. Let $\mathcal{U} = \cup_{i \in [g]}S_i$ denote the universe of the sets and let $\upsilon = |\mathcal{U}|$. 
We follow the construction exactly like in Subsection~\ref{subsec:query-tradeoff}.
Let $U=\mathcal{U}$ be the set of colors in the instance of the range entropy query we construct. For each $i \in [g]$, we build an array of points $A_i$, containing exactly $\upsilon$ points, and color them based on the set $S_i$. We color the first $\upsilon - |S_i|$ points in $A_i$ using colors from $U - S_i$ and the last $|S_i|$ points using the colors from $S_i$ in an arbitrary order. Similarly, for each $i \in [g]$, we build the array of $\upsilon$ points $B_i$. We color the first $|S_i|$ points in $B_i$ using the colors from $S_i$ and the rest of $\upsilon - |S_i|$ points using the colors from $U - S_i$ in an arbitrary order. We define the point set $P$ to be $A_1A_2\cdots A_g B_1B_2 \cdots B_g$, the concatenation of the points in all the blocks similar to Subsection~\ref{subsec:query-tradeoff}. We then construct the range entropy data structure $\mathcal{D}$ over $P$. Given a set intersection query to decide whether $S_i \cap S_j = \emptyset$, we define $\rho_{i,j} = [A_i[\upsilon - |S_i| + 1], B_j[|S_j|]]$, and query the data structure to return $\mathcal{D}(\rho_{i,j})$. To answer the given set intersection query, we only need to decide whether there is a common color in the last $\upsilon - |S_i|$ points of $A_i$ and the first $|S_i|$ points of $B_i$. As shown in Subsection~\ref{subsec:query-tradeoff}, this can be decided using both S-entropy and R-entropy in $O(1)$ time, similar to Lemmas~\ref{lem:s-matrix} and \ref{lem:r-matrix}. Therefore, after constructing $\mathcal{D}$ as described, we are able to answer the set intersection queries by doing a single query to $\mathcal{D}$. We have $|P| = \sum_{i \in [g]}(|A_i| + |B_i|) = 2\cdot g\cdot \upsilon$. Goldstein et al. showed in Theorem~8 of \cite{goldstein_set_intersection} that this reduction is enough to obtain the following theorem. While they use the range mode queries to show their result, it is easy to verify that their proof also follows in our settings. 
\begin{thm}
\label{thm:lbd1}
    If the set intersection conjecture is true, any data structure for range R-entropy (resp. S-entropy) queries over $n$ points in $\Re^1$ with $Q(n)$ query time, must use $\widetilde{\Omega}(\frac{n^2}{(Q(n))^4})$ space.
\end{thm}
}
\section{Exact Data Structures}
\label{sec:exact}
In this section we describe data structures that return the entropy in a query range, exactly. First, we provide a data structure for $d=1$ and we extend it to any constant dimension $d$. Next, we provide a second data structure for any constant dimension $d$. The first data structure is better for $d=1$, while the second data structure is better for any constant $d>1$.
\new{We describe all data structures for the range R-entropy queries, however all all results can be extended straightforwardly to range S-entropy queries.}

\subsection{Efficient data structure for $d=1$}
\label{subsec:DS1}
Let $P$ be a set of $n$ points in $\Re^1$.
Since the range entropy query problem is not decomposable, the main idea is to precompute the entropy in some carefully chosen canonical subsets of $P$. When we get a query interval $R$, we find the maximal precomputed canonical subset in $R$, and then for each color among the colors of points in $R$ not included in the canonical subset, we update the overall entropy using Equations~\ref{eq:entropyupdate}, \ref{eq:entropyupdateinsert}, and~\ref{eq:entropyupdatedelete}. We also describe how we can precompute the entropy of all canonical subsets efficiently.

\paragraph{Data Structure.}
Let $t\in[0,1]$ be a parameter. Let $B_t=\{b_1, \ldots, b_k\}$ be $k=n^{1-t}$ points in $\Re^1$ such that $|P\cap [b_j,b_{j+1}]|=n^t$, for any $j<n^{1-t}$. For any pair $b_i, b_j\in B_t$ let $I_{i,j}=[b_i,b_j]$ be the interval with endpoints $b_i, b_j$. \new{Let $I=\{I_{i,j}\mid b_i, b_j\in B_t, b_i\leq b_j$\} be the set of all intervals defined by the points in $B$.}
For any pair $b_i, b_j$ we store the interval $I_{i,j}$ and we precompute $\hat{H}_{i,j}=\ren_\alpha(P\cap I_{i,j})$, and $n_{i,j}=|P\cap I_{i,j}|$.
Finally, for each color $u\in u(P)$ we construct a search binary tree $\mathcal{T}_u$ over $P(u)$.

We have $|B_t|=O(n^{1-t})$ so $|I|=O(n^{2(1-t)})$. 
Furthermore, all constructed search binary trees have $O(n)$ space in total.
Hence we need $O(n^{2(1-t)})$ space for our data structure.

\paragraph{Query procedure.}
Given a query interval $R$, we find the maximal interval $I_{i,j}\in I$ such that $I\subseteq R$ \new{using two predecessor queries}. Recall that we have precomputed the entropy $\hat{H}_{i,j}$.
Let $\hat{H}=\hat{H}_{i,j}$ be a variable that we will update throughout the algorithm storing the current entropy. Let also $N=n_{i,j}$ be the variable that stores the number of items we currently consider to compute $H$.
Let $P_R=P\cap (R\setminus I_{i,j})$ be the points in $P\cap R$ that are not included in the maximal interval $I_{i,j}$. See also Figure~\ref{fig:Query1d}.
\begin{figure}[H]
    \includegraphics[scale=0.3]{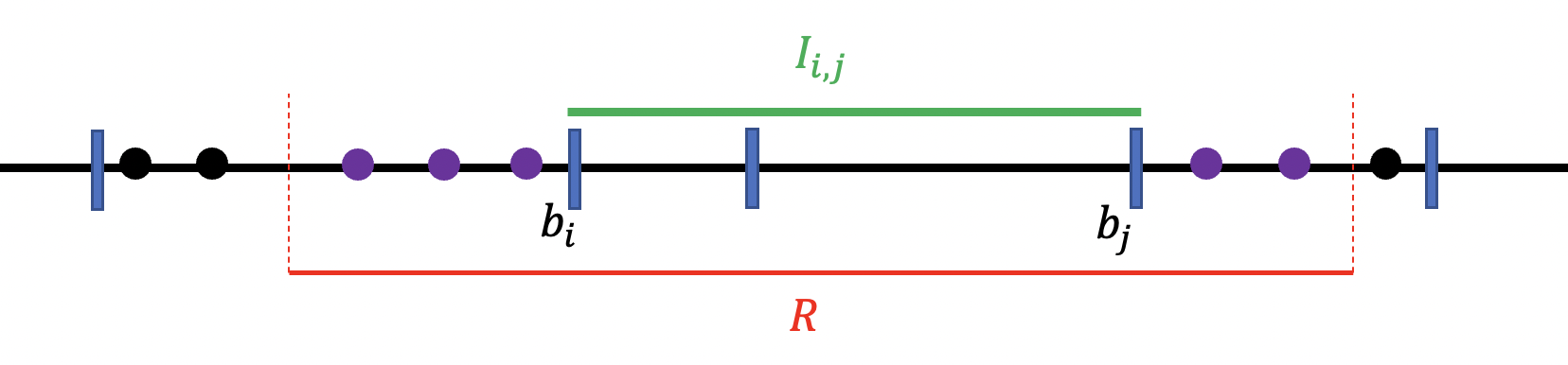}
   \caption{Instance of the query algorithm given query interval $R$. Purple points are points in $P_R$. \label{fig:Query1d}}
\end{figure}
We visit each point in $P_R$ and we identify $u(P_R)$.
For each $\mathbf{u}\in u(P_R)$, we run a query in $\mathcal{T}_\mathbf{u}$ with range $I_{i,j}$ finding the number of points in $P\cap I_{i,j}$ with color $\mathbf{u}$. Let $n_\mathbf{u}$ be this count.

If $n_\mathbf{u}=0$ then there is no point in $P\cap I_{i,j}$ with color $\mathbf{u}$ so we insert $|P_R(\mathbf{u})|$ items of color $\mathbf{u}$ in the current entropy using Equation~\ref{eq:Renentropyupdateinsert}. In that formula, $|P_1|=N$, $\ren_\alpha(P_1)=\hat{H}$ and $|P_2|=|u(P_R)|$. We update $N=N+|u(P_R)|$, and $\hat{H}$ with the updated entropy $\ren_\alpha(P_1\cup P_2)$.

If $n_\mathbf{u}>0$ then there is at least one point in $P\cap I_{i,j}$ with color $\mathbf{u}$. Hence, we update the entropy $\hat{H}$, by first removing the $n_\mathbf{u}$ points of color $\mathbf{u}$ in $P\cap I_{i,j}$ and then re-inserting $n_\mathbf{u}+|u(P_R)|$ points of color $\mathbf{u}$. We use Equation~\ref{eq:Renentropyupdatedelete} for removing the points with color $\mathbf{u}$ with $|P_1|=N$, $\ren_\alpha(P_1)=\hat{H}$, and $|P_3|=n_u$. We update $N=N-n_\mathbf{u}$ and $\hat{H}$ with the updated entropy $\ren_\alpha(P_1\setminus P_3)$. Then we use Equation~\ref{eq:Renentropyupdateinsert} for re-inserting the points with color $u$, with $|P_1|=N$, $\ren_\alpha(P_1)=\hat{H}$, and $|P_2|=n_\mathbf{u}+|u(P_R)|$. We update $N=N+n_\mathbf{u}+|u(P_R)|$ and $H$ with the updated entropy $\ren_\alpha(P_1\cup P_2)$.
After visiting all colors in $u(P_R)$, we return the updated entropy $\hat{H}$.
The correctness of the algorithm follows from Equations~\ref{eq:Renentropyupdateinsert}, \ref{eq:Renentropyupdatedelete}. For each color $u\in u(P_R)$ we update the entropy including all points of color $u$.

\new{For a query interval $R$ the predecessor queries take $O(\log n)$ time to find $I_{i,j}$}. The endpoints of $R$ intersect two intervals $[b_h, b_{h+1}]$ and $[b_v, b_{v+1}]$. Recall that by definition, such interval contains $O(n^{t})$ points from $P$. Hence, $|P_R|=O(n^{t})$ and $|u(P_R)|=O(n^t)$. For each $\mathbf{u} \in u(P_R)$, we spend $O(\log n)$ time to search $\mathcal{T}_\mathbf{u}$ and find $n_\mathbf{u}$. Then we update the entropy in $O(1)$ time. Overall, the query procedure takes $O(n^{t}\log n)$ time.

\paragraph{Fast Construction.}
In order to construct the data structure we need to compute $\hat{H}_{i,j}$ for every interval $I_{i,j}$. A straightforward algorithm is the following:
We first visit all intervals $I_{i,i+1}$ and compute the entropy by traversing all points in $P\cap I_{i,i+1}$.
Then we repeat the same for intervals $I_{i,i+2}$. More specifically, we first make a pass over $P$ and we compute $\hat{H}_{i,i+2}$ for each $i=\{1,3, 5, \ldots\}$. Then, we make another pass over $P$ and we compute, $\hat{H}_{i,i+2}$ for each $i=\{2,4, 6, \ldots\}$. We continue with the same way for intervals $I_{i,i+\ell}$. Overall the running time is upper bounded by $O\left(n+\sum_{\ell=2}^{n^{1-t}}\ell\cdot \frac{n^{1-t}}{\ell}n\right)=O(n^{3-2t})$.
We can improve the construction with the following trick. The high level idea of the algorithm remains the same. However, when we compute $\hat{H}_{i,i+\ell}$, notice that we have already computed $\hat{H}_{i,i+\ell-1}$. Hence, we can use $\hat{H}_{i,i+\ell-1}$ and only traverse the points in $P\cap I_{i+\ell-1,i+\ell}$ updating $\hat{H}_{i,i+\ell-1}$ as we did in the query procedure. Each interval $I_{i+\ell-1,i+\ell}$ contains $O(n^t)$ points so we need only $O(n^t\log n)$ time to find the new entropy. For each $\ell$, we need $O(\frac{n^{1-t}}{\ell}n^t)$ time to find all $\hat{H}_{i,i+\ell}$ for $i=\{1, 1+\ell, 1+2\ell,\ldots\}$. Hence, we need $O(\ell\frac{n^{1-t}}{\ell}n^t)$ time to compute all entropies $\hat{H}_{i,i+\ell}$.
Overall we can construct our data structure in $O\left(\sum_{\ell=1}^{n^{1-t}}\ell\cdot\frac{n^{1-t}}{\ell}n^t\right)=O(n^{2-t})$ time.

\new{\paragraph{Extension to Shannon Entropy.}
The data structure can be extended straightforwardly to the range S-entropy query. The only difference is that instead of computing $\ren_\alpha(P\cap I_{i,j})$, we pre-compute $H(P\cap I_{i,j})$ and we use the the Equations~\eqref{eq:entropyupdateinsert},~\eqref{eq:entropyupdatedelete} to update the Shannon entropy. We conclude with the next theorem.

\begin{thm}
\label{thm:exact1d}
Let $P$ be a set of $n$ points in $\Re^1$, where each point is associated with a color, and let $\alpha, t$ be two parameters such that $\alpha>1$ and $t\in[0,1]$. A data structure of $O(n^{2(1-t)})$ size can be constructed in $O(n^{2-t})$ time, such that given a query interval $R$, $H(P\cap R)$ and $\ren_\alpha(P\cap R)$ can be computed in $O(n^{t}\log n)$ time.
\end{thm}
}

\subsection{Efficient data structure for $d>1$}
\label{subsec:DSd}
While the previous data structure can be extended to higher dimensions, here we propose a more efficient data structure for $d>1$. In this data structure we split the points with respect to their colors. The data structure has some similarities with the data structure presented in~\cite{agarwal2018range, agarwal2016range} for the max query under uncertainty, however, the two problems are different and there are key differences on the way we construct the data structure and the way we compute the result of the query.

\begin{figure}[H]
        \includegraphics[width=0.5\linewidth]{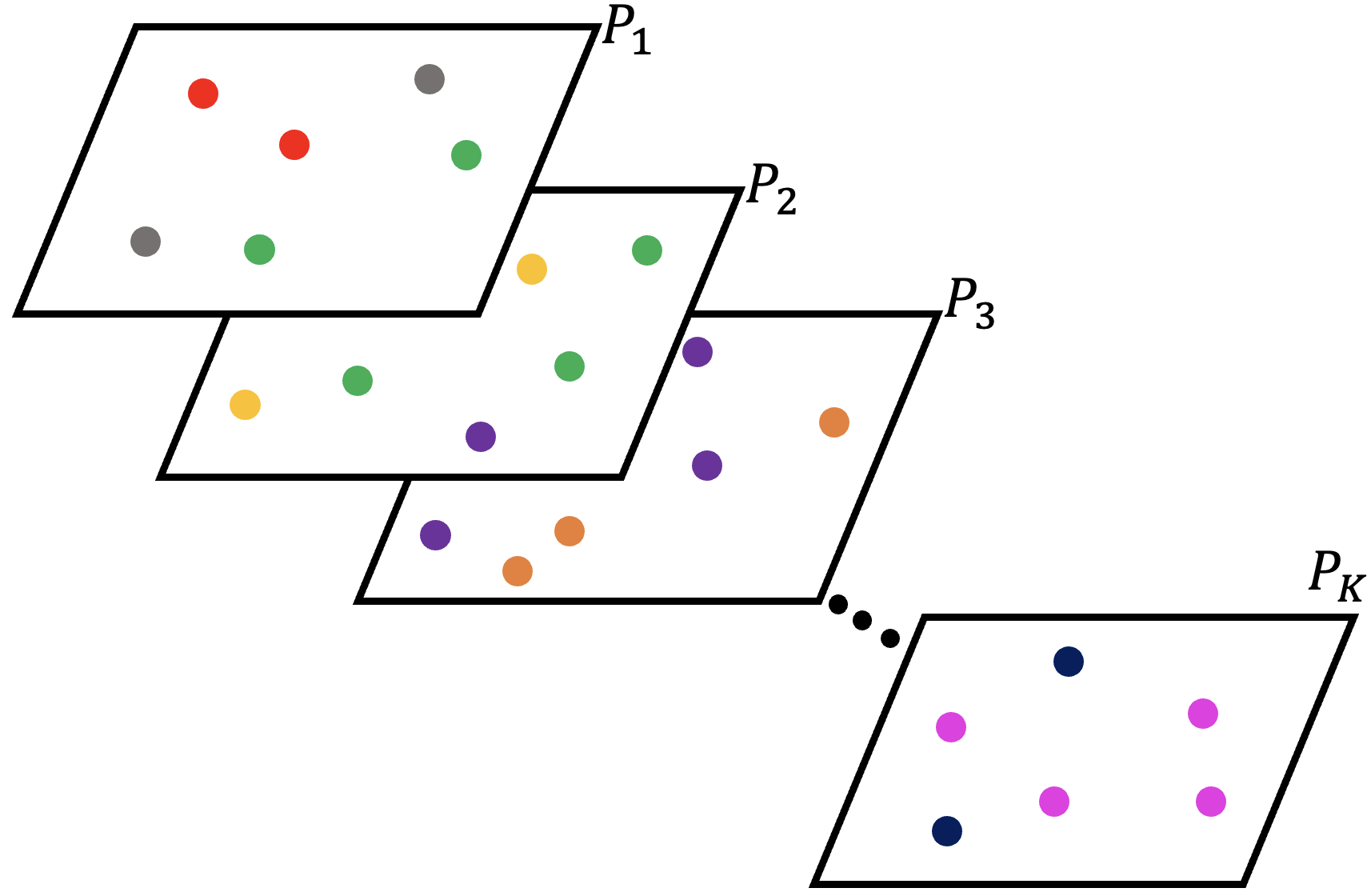}
              \caption{Partition $P$ into $K$ buckets in $\Re^2$. Two consecutive buckets have at most one color in common. \label{fig:Querydd}}
\end{figure}

\paragraph{Data Structure.}
We first consider an arbitrary permutation of the colors in $U$, i.e. $u_1, \ldots, u_m$. The order used to partition the items is induced from the permutation over the colors. Without loss of generality, we set $u_j=j$ for each $j\leq m$.
We split $P$ into $K=O(n^{1-t})$ buckets $P_1,\ldots, P_K$ such that i) each bucket contains $O(n^t)$ points, and ii) for every point $p\in P_i$ and $q\in P_{i+1}$, $u(p)\geq u(q)$.
We notice that for any pair of buckets $P_i$, $P_{i+1}$ it holds $|u(P_i)\cap u(P_{i+1})|\leq 1$, see Figure~\ref{fig:Querydd}. We slightly abuse the notation and we use $P_i$ to represent both the $i$-th bucket and the set of points in the $i$-th bucket.


For each bucket $P_i$, we take all combinatorially different (hyper)rectangles $R_i$ defined by the points $P_i$. For each such rectangle $r$, we precompute and store the entropy $\ren_\alpha(P_i\cap r)$ along with the number of points $n(P_i\cap r)=|P_i\cap r|$. In addition, we store $u^+(r)$, the color with the maximum value (with respect to the permutation of the colors)in $r\cap P_i$. Furthermore, we store $u^-(r)$, the color with the minimum value in $r\cap P_i$. Let $n^+(r)=|\{p\in r\cap P_i\mid u(p)=u^+(r)\}|$ and $n^-(r)=|\{p\in r\cap P_i\mid u(p)=u^-(r)\}|$.
Finally, for each bucket $P_i$ we construct a modified range tree $\mathcal{T}'_i$ over all $R_i$, such that given a query rectangle $R$ it returns the maximal rectangle $r\in R_i$ that lies completely inside $R$. We note that $r\cap P_i = R\cap P_i$. This can be done by representing the $d$-dimensional hyper-rectangles as $2d$-dimensional points merging the coordinates of two of their corners\new{, similarly to~\cite{esmailpour2025theoretical} (Section 4.2).}

Overall, we need $O(n\log^{2d-1} n)$ space for the modified range trees $\mathcal{T}'_i$, and $O(n^{1-t}\cdot n^{2dt})=O(n^{(2d-1)t+1})$ space to store all additional information (entropy, counts, max/min color) in each rectangle. This is because there are $O(n^{1-t})$ buckets, and in each bucket there are $O(n^{2dt})$ combinatorially different rectangles.
Overall, our data structure has $O\left(n\log^{2d-1}n + n^{(2d-1)t+1}\right)$ space.

\paragraph{Query Procedure.}
We are given a query (hyper)rectangle $R$.
We visit the buckets $P_1,\ldots P_K$ in order and compute the entropy for $R\cap(P_1\cup\ldots\cup P_i)$.
Let $\hat{H}$ be the overall entropy we have computed so far.
For each bucket $P_i$ we do the following: First we run a query using $\mathcal{T}'_i$ to find $r_i\in R_i$ that lies completely inside $R$. Then we update the entropy $\hat{H}$ considering the items in $P_i\cap r_i$. If $u^-(r_{i-1})=u^+(r_i)$ then we update the entropy $\hat{H}$ by removing $n^-(r_{i-1})$ points with color $u^-(r_{i-1})$ using Equation~\ref{eq:Renentropyupdatedelete}. Then we insert $n^-(r_{i-1})+n^+(r_i)$ points of color $u^+(r_i)$ in $\hat{H}$ using Equation~\ref{eq:Renentropyupdateinsert}. Finally, we remove $n^+(r_i)$ points of color $u^+(r_i)$ from the precomputed $\ren_\alpha(P_i\cap r_i)$ using Equation~\ref{eq:Renentropyupdatedelete} and we merge the updated $\hat{H}$ with $H(P_i\cap r_i)$ using Equation~\ref{eq:Renentropyupdate}. We note that in the last step we can merge the updated $\hat{H}$ with the updated $\ren_\alpha(P_i\cap r_i)$ because no color from the points used to compute the current $\hat{H}$ appears in the points used to compute the current $\ren_\alpha(P_i\cap r_i)$. On the other hand, if $u^-(r_{i-1})\neq u^+(r_i)$, then we merge the entropies $\hat{H}$ and $\ren_\alpha(P_i\cap r_i)$ using directly Equation~\ref{eq:Renentropyupdate}.

In each bucket $P_i$ we need $O(\log^{2d} n)$ to identify the maximal rectangle $r_i$ inside $R$. Then we need $O(1)$ time to update the current entropy $\hat{H}$. Overall, we need $O(n^{1-t}\log^{2d} n)$ time.

\paragraph{Fast Construction.}
All range trees can be computed in $O(n\log^{2d} n)$ time. Next, we focus on computing $\ren_\alpha(P_i\cap r)$ for all rectangles $r\in R_i$. We compute the other quantities $n(P_i\cap r)$, $u^-(r)$, and $u^+(r)$ with a similar way.
A straightforward way is to consider every possible rectangle $r$ and compute independently the entropy in linear time. There are $O(n^{2dt})$ rectangles so the running time is $O(n^{2dt+1})$. We propose a faster construction algorithm.

The main idea is to compute the entropy for rectangles in a specific order. In particular, we compute the entropy of rectangles that contain $c$ points after we compute the entropies for rectangles that contain $c-1$ points. Then we use Equations~\ref{eq:Renentropyupdateinsert}, \ref{eq:Renentropyupdatedelete} to update the entropy of the new rectangle without computing it from scratch.

More specifically,
let $L_d$ be the points in $P$ sorted in ascending order with respect to their $d$-th coordinate. For each color $u_k$ we construct a range tree $\mathcal{T}_k$ for range counting queries. Furthermore, we construct a range tree $\mathcal{T}$ for range counting queries (independent of color). Let $P_i$ be a bucket. Assume that we have already computed the entropy for every rectangle that contains $c-1$ points in $P_i$. We traverse all rectangles containing $c$ points: Let $p$ be any point in $P_i$. We assume that $p$ lies in the bottom hyperplane of the hyper-rectangle (with respect to $d$-th coordinate). Next we find the points that lie in the next $2d-2$ sides of the rectangle. In particular we try all possible sets of $2d-2$ points in $P_i$. We notice that each such set, along with the first point $p$, defines an open hyper-rectangle, i.e., a hyper-rectangle whose bottom hyperplane with respect to the $d$-th coordinate passes through point $p$ and there is no top hyperplane with respect to coordinate $d$. We find the top-hyperplane by running a binary search on $L_d$. For each point $q\in P_i$ we check in the binary search, let $r$ be the hyper-rectangle defined by the set of $2d$ points we have considered. Using $\mathcal{T}$, we run a range counting query on $r\cap P_i$. If $|r\cap P_i|<c$ then we continue the binary search on the larger values. If $|r\cap P_i|>c$, we continue the binary search on the smaller values. If $|r\cap P_i|=c$ then let $q\in P_i$ be the point on the top hyperplane we just checked in the binary search. We run another binary search on $L_d$ to find the hyper-rectangle $r'\subseteq r$ that contains $c-1$ points. Again, we use the range tree $\mathcal{T}$ to find the rectangle $r'$ as we run the binary search on $L_d$. We have, $\ren_\alpha(r\cap P_i)=\ren_\alpha\left((r'\cap P_i)\cup \{q\}\right)$. Let $u(q)=u_k$. Using $\mathcal{T}_k$ we count $n(r',u_k)$ the number of points in $r'$ with color $u_k$. Let $\hat{H}$ be the entropy of $\ren_\alpha(P_i\cap r')$ by removing $n(r',u_k)$ points of color $u_k$ from $P_i\cap r'$ as shown in Equation~\ref{eq:Renentropyupdatedelete}. Finally, we get the entropy $\ren_\alpha(P_i\cap r)$ by updating $\hat{H}$, inserting $n(r',u_k)+1$ points of color $u_k$, as shown in Equation~\ref{eq:Renentropyupdateinsert}.

The running time is bounded by $O(n^{(2d-1)t+1}\log^{d+1} n)$ time, because we have $O(n^{1-t})$ buckets, each rectangle in a bucket contains at most $O(n^{t})$ points so we have to check $O(n^{t})$ values of $c$, then we take $O(n^{t})$ possible points $p$, and all sets of size $2d-2$ are $O(n^{(2d-2)t})$. For each such rectangle we run two binary searches where each step takes $O(\log^{d} n)$ time to run the range counting query.

\new{\paragraph{Extension to Shannon Entropy.}
Similarly to Subsection~\ref{subsec:DS1}, the data structure can be extended straightforwardly to the range S-entropy query using the the Equations~\eqref{eq:entropyupdateinsert},~\eqref{eq:entropyupdatedelete},~\eqref{eq:entropyupdate} to update the Shannon entropy. We conclude with the next theorem.

\begin{thm}
\label{thm:d-dim-entopy}
Let $P$ be a set of $n$ points in $\Re^d$, where each point is associated with a color, and let $\alpha, t$ be two parameters such that $\alpha>1$ and $t\in [0,1]$. A data structure of $O(n\log^{2d-1}n + n^{(2d-1)t+1})$ size can be constructed in $O(n\log^{2d} n + n^{(2d-1)t+1}\log^{d+1} n)$ time, such that given a query hyper-rectangle $R$, $H(P\cap R)$ and $\ren_{\alpha}(P\cap R)$ can be computed in $O(n^{1-t}\log^{2d} n)$ time.
\end{thm}}

\section{Approximate Data Structures for S-Entropy Queries}
\label{sec:approx}
In this section we describe data structures that return the Shannon entropy in a query range, approximately. First, we present a data structure that returns an additive approximation of the Shannon entropy and next we present a data structure that returns a multiplicative approximation efficiently. Then, for $d=1$, we design a deterministic and more efficient data structure that returns an additive and multiplicative approximation of the Shannon entropy.

\subsection{Additive Approximation}
\label{subsec:ApproxAdd}
In this Subsection, we construct a data structure on $P$ such that given a query rectangle $R$ and a parameter $\Delta$, it returns a value $h$ such that $H(P\cap R)-\Delta\leq h\leq H(P\cap R)+\Delta$.
The intuition comes from the area of finding an additive approximation of the entropy of an unknown distribution in the dual access model~\cite{canonne2014testing}.

Let $D$ be a fixed distribution over a set of values (outcomes) $\out_1, \ldots, \out_N$. Each value $\out_i$ has a probability $D(\out_i)$ which is not known, such that $\sum_{i=1}^N D(\out_i)=1$. The authors in~\cite{canonne2014testing} show that if we ask $O\left(\frac{\log^2 \frac{N}{\Delta}\log N}{\Delta^2}\right)$ sample queries in the dual access model, then we can get a $\Delta$ additive-approximation of the entropy of $D$ with high probability in $O\left(\frac{\log^2 \frac{N}{\Delta}\log N}{\Delta^2}\mathcal{S}\right)$ time, where $\mathcal{S}$ is the running time to get a sample.
In the dual access model, we consider that we have a dual oracle for $D$ which is a pair of oracles $(\textsf{SAMP}_D, \textsf{EVAL}_D)$. When required, the sampling oracle $\textsf{SAMP}_D$ returns a value $\out_i$ with probability $D(\out_i)$, independently of all previous calls to any oracle. Furthermore, the evaluation oracle $\textsf{EVAL}_D$ takes as input a query element $\out_i$ and returns the probability weight $D(\out_i)$.

Next, we describe how the result above can be used in our setting. The goal in our setting is to find the entropy $H(P')$, where $P'=P\cap R$, for a query rectangle $R$. The colors in $u(P')$ define the distinct values in distribution $D$. By definition, the number of colors is bounded by $|P'|=O(n)$.
The probability weight is defined as $\frac{|P'(u_i)|}{|P'|}$.
We note that in~\cite{canonne2014testing} they assume that they know $N$, i.e., the number of values in distribution $D$. In our case, we cannot compute the number of colors $|u(P')|$ efficiently. Even though we can easily compute an $O(\log^d n)$ approximation of $|u(P')|$, it is sufficient to use the loose upper bound $|u(P')|\leq n$. This is because, without loss of generality, we can assume that there exist $n-|u(P')|$ values/colors with probability (arbitrarily close to) $0$. All the results still hold.
Next, we present our data structure to simulate the dual oracle.

\paragraph{Data structure.}
For each color $u_i\in U$ we construct a range tree $\mathcal{T}_i$ on $P(u_i)$ for range counting queries. We also construct another range tree $\mathcal{T}$ on $P$ for range counting queries, which is independent of the color.
Next, we construct a range tree $\mathcal{S}$ on $P$ for range sampling queries as described in Section~\ref{sec:prelim}.
We need $O(n\log^{d}n)$ time to construct all the range trees, while the overall space is $O(n\log^{d-1}n)$.

\paragraph{Query procedure.}
The query procedure involves the algorithm for estimating the entropy of an unknown distribution in the dual access model~\cite{canonne2014testing}.
Here, we only need to describe how to execute the oracles $\textsf{SAMP}_D$ and $\textsf{EVAL}_D$ in $P'=P\cap R$ using the data structure.
\begin{itemize}
    \item $\textsf{SAMP}_D$: Recall that $\textsf{SAMP}_D$ returns $\out_i$ with probability $D(\out_i)$. In our setting, values $\out_1,\ldots, \out_n$ correspond to colors. So, the goal is to return a color $u_i$ with probability proportional to the number of points with color $u_i$ in $P'$. Indeed, $\mathcal{S}$ returns a point $p$ uniformly at random in $P'$. Hence, the probability that a point with color $u_i$ is found is $\frac{|P'(u_i)|}{|P'|}$.
    \item $\textsf{EVAL}_D$: Recall that given a value $\out_i$, $\textsf{EVAL}_D$ returns the probability weight $D(\out_i)$. Equivalently, in our setting, given a color $u_i$, the goal is to return $\frac{|P'(u_i)|}{|P'|}$. Using $\mathcal{T}_i$ we run a counting query in the query rectangle $R$ and find $|P'(u_i)|$. Then using $\mathcal{T}$, we run a counting query in $R$ and we get $|P'|$. We divide the two quantities and return the result.
\end{itemize}
In each iteration, every oracle call $\textsf{SAMP}_D$ and $\textsf{EVAL}_D$ executes a constant number of range tree queries, so the running time is $O(\log^{d} n)$. The algorithm presented in~\cite{canonne2014testing} calls the oracles $O(\frac{\log^2\frac{n}{\Delta}\log n}{\Delta^2})$ times to guarantee the result with probability at least $1-1/n$, so the overall query time is $O\left(\frac{\log^{d+1}n \cdot \log^2\frac{n}{\Delta}}{\Delta^2}\right)$. We note that if $\Delta<\frac{1}{\sqrt{n}}$ then the query time is $\Omega(n\log n)$. However, it is trivial to compute the entropy in $P\cap R$ in $O(n\log n)$ time by traversing all points in $P\cap R$. Hence, the additive approximation is non-trivial when $\Delta\geq \frac{1}{\sqrt{n}}$. In this case, $\log^2 \frac{n}{\Delta^2}=O(\log^2 n)$.
We conclude that the query time is bounded by $O\left(\frac{\log^{d+3}n}{\Delta^2}\right)$.
We conclude with the next theorem.

\begin{thm}
\label{thm:add-approx}
Let $P$ be a set of $n$ points in $\Re^d$, where each point is associated with a color. A data structure of $O(n\log^{d-1}n)$ size can be constructed in $O(n\log^{d} n)$ time, such that given a query hyper-rectangle $R$ and a real parameter $\Delta$, a value $h$ can be computed in $O\left(\frac{\log^{d+3}n}{\Delta^2}\right)$ time, such that $H(P\cap R)-\Delta\leq h\leq H(P\cap R)+\Delta$, with high probability.
\end{thm}
This data structure can be made dynamic under arbitrary insertions and deletions of points using well known techniques~\cite{ bentley1980decomposable, erickson2011static, overmars1983design, overmars1981worst}. The update time is $O(\log^d n)$.

\subsection{Multiplicative Approximation}
\label{subsec:ApproxMult}
In this Subsection, we construct a data structure such that given a query rectangle $R$ and a parameter $\eps$, it returns a value $h$ such that $\frac{1}{1+\eps}H(P\cap R)\leq h\leq (1+\eps)H(P\cap R)$.
The intuition comes for the area of finding a multiplicative approximation of the entropy of an unknown distribution in the dual access model~\cite{guha2006streaming} and the streaming algorithms for finding a multiplicative approximation of the entropy~\cite{chakrabarti2007near}.
In particular, in this section we extend the streaming algorithm proposed in~\cite{chakrabarti2007near} to work in the query setting.

We use the notation from the previous Subsection where $D$ is an unknown distribution over a set of values $\out_1, \ldots, \out_N$.
It is known~\cite{guha2006streaming} that if we ask $O\left(\frac{\log N}{\eps^2\cdot H'}\right)$ queries in the dual access model, where $H'$ is a lower bound of the actual entropy of $D$, i.e., $H(D)\geq H'$, then we can get an $(1+\eps)$-multiplicative approximation of the entropy of $D$ with high probability, in $O\left(\frac{\log N}{\eps^2\cdot H'}\mathcal{S}\right)$ time, where $\mathcal{S}$ is the time to get a sample.
We consider that we have a dual oracle for $D$ which is a pair of oracles $(\textsf{SAMP}_D, \textsf{EVAL}_D)$, as we had in additive approximation. Similarly to the additive approximation, in our setting we do not know the number of colors in $P'=P\cap R$ or equivalently the number of values $N$ in distribution $D$. However, it is sufficient to use the upper bound $|u(P')|\leq n$ considering $n-|u(P')|$ colors with probability (arbitrarily close to) $0$.
If we use the same data structure constructed for the additive approximation, we could solve the multiplicative-approximation, as well. While this is partially true, there is a big difference between the two problems. What if the actual entropy is very small so $H'$ is also extremely small?
In this case, the factor $\frac{1}{H'}$ will be very large making the query procedure slow.

We overcome this technical difficulty by considering two cases. If $H'$ is large, say $H'\geq 0.9$, then we can compute a multiplicative approximation of the entropy efficiently applying~\cite{guha2006streaming}. On the other hand, if $H'$ is small, say $H'<0.9$, then we use the ideas from~\cite{chakrabarti2007near} to design an efficient data structure. In particular, we check if there exists a value $a_M$ with $D(a_M)>2/3$. If it does not exist then $H'$ is large so it is easy to handle. If $a_M$ exists, we write $H(D)$ as a function of $H(D\setminus\{a_M\})$ using Equation~\ref{eq:entropyupdatedelete}. In the end, if we get an additive approximation of $H(D\setminus\{a_M\})$ we argue that this is sufficient to get a multiplicative approximation of $H'$.

\paragraph{Data Structure.}
For each color $u_i$ we construct a range tree $\mathcal{T}_i$ over $P(u_i)$ as in the previous Subsection. Similarly, we construct a range tree $\mathcal{T}$ over $P$ for counting queries.
We also construct the range tree $\mathcal{S}$ for returning uniform samples in a query rectangle. In addition to $\mathcal{S}$, we also construct a variation of this range tree, denoted by $\bar{\mathcal{S}}$. Given a query rectangle $R$ and a color $u_i\in U$, $\bar{\mathcal{S}}$ returns a point from $\{p\in R\cap P\mid u(p)\neq u_i\}$ uniformly at random. In other words, $\bar{\mathcal{S}}$ is a data structure over $P$ that is used to return a point in a query rectangle uniformly at random excluding points of color $u_i$. While $\bar{\mathcal{S}}$ is an extension of $\mathcal{S}$, the low level details are more tedious and are shown in the next paragraphs.

\new{
We extend the range tree data structure for range sampling queries we showed in Section~\ref{sec:prelim}. Given a query rectangle $R$ and a color $u_j$, the goal is to return a uniform sample among the points in $(P\cap R)\setminus P(u_j)$. 
We construct a standard range tree on the points set $P$, as in Section~\ref{sec:prelim}. Using the same notation as in Section~\ref{sec:prelim}, for a $d$-level node $v$ of the range tree, we use the notation $P_v$ to denote the subset of points $P\cap \square_v$.
In each $d$-level node $v$ of the range tree, we store a hashmap $M_v$ having as keys the colors of the points stored in leaf nodes of the subtree rooted at $v$, and as values the number of leaf nodes in the subtree rooted at $v$ with color key. 
More formally, for each node $v$, we construct a hashmap $M_v$, such that for every color $u_i\in u(P_v)$, $M_v[u_i]=|P_v(u_i)|$.
For each node $v$ we also store the cardinality $c(v)=|P_v|=|P\cap \square_v|$. The modified range tree can be constructed in $O(n\log^d n)$ time and it has $O(n\log^d n)$ space because for every node $u$ the hashmap $M_v$ takes $O(|P\cap \square_v|)$ space. Given a query rectangle $R$ and a color $u_j\in U$, we get the set of canonical nodes $\mathcal{N}(R)$. For each node $v\in \mathcal{N}(R)$ we define the weight $w_v=\frac{c(v)-M_v[u_j]}{\sum_{v'\in \mathcal{N}(R)}(c(v')-M_{v'}[u_j])}$. We sample one node from $\mathcal{N}(R)$ with respect to the weights $\{w_v\mid v\in\mathcal{N}(R)\}$ using reservoir sampling.  
Let $v$ be the node that is sampled.
If $v$ is a leaf node then we return the point that is stored in node $v$. Otherwise, assume that $v$ has two children $x, y$. We move to the node $x$ with probability $\frac{c(x)-M_x[u_j]}{c(x)-M_x[u_j]+c(y)-M_y[u_j]}$ and to node $y$ with probability $\frac{c(y)-M_y[u_j]}{c(x)-M_x[u_j]+c(y)-M_y[u_j]}$. We recursively repeat this process until we reach a leaf node of the range tree.

\textit{Analysis.}
Similarly to the range tree for sampling without excluding any color, the query procedure takes $O(\log^d n)$ time.

Next, we show that the sampled point is chosen uniformly at random, i.e., with probability $\frac{1}{|(P\cap R)\setminus P(u_j)|}$.
Let $v\rightarrow v_1\rightarrow\ldots\rightarrow v_k$ be the path of nodes followed by the algorithm to sample a point $p$. Thus $p$ is stored in the leaf node $v_k$. Let $\bar{v}_1,\ldots, \bar{v}_k$ be the siblings of nodes $v_1, \ldots, v_k$, respectively. The probability that $p$ is selected is 
$$\frac{c(v)-M_v[u_j]}{\sum_{v'\in \mathcal{N}(R)}(c(v')\!\!-\!\!M_{v'}[u_j])}\cdot \frac{c(v_1)-M_{v_1}[u_j]}{c(v_1)\!\!-\!\!M_{v_1}[u_j]\!+\!c(\bar{v}_1)\!\!-\!\!M_{\bar{v}_1}[u_j]}\cdot \ldots\cdot \frac{c(v_k)-M_{v_k}[u_j]}{c(v_k)\!-\!M_{v_k}[u_j]\!+\!c(\bar{v}_k)\!-\!M_{\bar{v}_k}[u_j]}.$$
Notice that $c(v)-M_v[u_j]=c(v_1)-M_{v_1}[u_j]+c(\bar{v}_1)-M_{\bar{v}_1}[u_j]$ and $c(v_\ell)-M_{v_\ell}[u_j]=c(v_{\ell+1})-M_{v_{\ell+1}}[u_j]+c(\bar{v}_{\ell+1})-M_{\bar{v}_{\ell+1}}[u_j]$ for every $\ell\in[k-1]$. Furthermore $c(v_k)-M_{v_k}[u_j]=1$ because $v_k$ is a leaf node.
We conclude that the probability of selecting $p$ is $\frac{1}{\sum_{v'\in\mathcal{N}(R)}(c(v')-M_{v'}[u_j])}=\frac{1}{|(P\cap R)\setminus P(u_j)|}$.

Similarly to the range tree for sampling without excluding the points of any color, the data structure can be used to sample on weighted points.
Assume that each point $p\in P$ has a weight $w(p)$, which is a non-negative real number. Given a query hyper-rectangle $R$ the goal is to sample a point from $P\cap R$ with respect to their weight, i.e., a point $p\in P\cap R$ should be selected with probability $\frac{w(p)}{\sum_{p'\in P\cap R}w(p')}$. The construction is exactly the same as in the unweighted case. The only difference is that instead of storing the count $c(v)$ in each node $v$, we store $w(v)=\sum_{p'\in P\cap \square_v}w(p')$ and instead of setting $M_v[u_i]=|P_v(u_i)|$ we store $M_v[u_i]=\sum_{p'\in P_v(u_i)}w(p')$. The query time remains $O(\log^d n)$ and the correctness proof remains the same replacing $c(v)$ with $w(v)$, for each node $v$ of the range tree.
}


The complexity of the entire data structure is dominated by the complexity of $\bar{\mathcal{S}}$. Overall, it can be computed in $O(n\log^{d} n)$ time and it has $O(n\log^{d} n)$ space.

\paragraph{Query procedure.}
First, using $\mathcal{T}$ we get $N=|P\cap R|$. Using $\mathcal{S}$ we get $\frac{\log (2n)}{\log 3}$ independent random samples from $P\cap R$. Let $P_S$ be the set of returned samples. For each $p\in P_S$ with $u(p)=u_i$, we run a counting query in $\mathcal{T}_i$ to get $N_i=|P(u_i)\cap R|$. Finally, we check whether $\frac{N_i}{N}>2/3$.
If we do not find a point $p\in P_S$ (assuming $u(p)=u_i$) with $\frac{N_i}{N}> 2/3$ then we run the algorithm from~\cite{guha2006streaming}.
In particular, we set $H'=0.9$ and we
run $O\left(\frac{\log n}{\eps^2\cdot H'}\right)$ oracle queries $\textsf{SAMP}_D$ or $\textsf{EVAL}_D$, as described in~\cite{guha2006streaming}.
In the end we return the estimate $h$.
Next, we assume that the algorithm found a point with color $u_i$ satisfying $\frac{N_i}{N}>2/3$.
Using $\bar{\mathcal{S}}$ (instead of $\mathcal{S}$) we run the query procedure of the previous Subsection and we get an $\eps$-additive approximation of $H((P\setminus P(u_i))\cap R)$, i.e., the entropy of the points in $P\cap R$ excluding points of color $u_i$. Let $h'$ be the $\eps$-additive approximation we get. In the end, we return the estimate $h=\frac{N-N_i}{N}\cdot h'+\frac{N_i}{N}\log\frac{N}{N_i}+\frac{N-N_i}{N}\log\frac{N}{N-N_i}$.

\paragraph{Correctness.}
It is straightforward to see that if there exists a color $u_i$ containing more than $2/3$'s of all points in $P\cap R$ then $u_i\in u(P_S)$ with high probability.

\begin{lem}
\label{lem:ProbEntr}
Let $u_i$ be the color with $\frac{|P(u_i)\cap R|}{|P\cap R|}>2/3$, and let $B$ be the event that $u_i\in u(P_S)$. The following holds: $\Pr[B]\geq 1-1/(2n)$.
\end{lem}
\begin{proof}
Let $B_j$ be the event that the $j$-th point selected in $P_S$ does not have color $u_i$. We have $\Pr[B_j]\leq 1/3$.
Then we have $\Pr[\bigcap_j B_j]\leq\frac{1}{3^{|P_S|}}$, since the random variables $B_j$'s are independent.
We conclude that $\Pr[B]=1-\Pr[\bigcap_j B_j]\geq 1-\frac{1}{3^{|P_S|}}=1-\frac{1}{2n}$.
\end{proof}
Hence, with high probability, we make the correct decision.

\new{
The next Lemma holds by a simple convexity argument as shown in~\cite{chakrabarti2007near}.
\begin{lemC}[\hspace{-0.2mm}\cite{chakrabarti2007near}]
\label{lem:maxprobentr}
Let $D$ be a discrete distribution over $m$ values $\{\out_1,\ldots, \out_m\}$ and let $D(\out_i)>0$ for at least two indices $i$. If there is no index $j$ such that $D(\out_j)> 2/3$, then $H(D)>0.9$.
\end{lemC}
If for every color $u_i\in u(P)$ it holds $\frac{P(u_i)\cap R}{|P\cap R|}\leq \frac{2}{3}$, then by Lemma~\ref{lem:maxprobentr} it follows that $H(P\cap R)>0.9$.}

Hence, $O\left(\frac{\log n}{\eps^2}\right)$ oracle queries are sufficient to derive an $(1+\eps)$-multiplicative approximation of the correct entropy.

The interesting case is when we find a color $u_i$ such that $\frac{N_i}{N}>2/3$ and $\frac{N_i}{N}<1$ (if $\frac{N_i}{N}=1$ then $H(P\cap R)=0$). Using the results of the previous Subsection along with the new data structure $\bar{\mathcal{S}}$, we get $h'\in [H((P\setminus P(u_i))\cap R)-\eps, H((P\setminus P(u_i))\cap R)+\eps]$ with probability at least $1-1/(2n)$. We finally show that the estimate $h$ we return is a multiplicative approximation of $H(P\cap R)$.
From Equation~\ref{eq:entropyupdatedelete}, we have $H(P\cap R)=\frac{N-N_i}{N}H((P\setminus P(u_i))+\frac{N_i}{N}\log\frac{N}{N_i}+\frac{N-N_i}{N}\log\frac{N}{N-N_i}$. Since $h'\in [H((P\setminus P(u_i))\cap R)-\eps, H((P\setminus P(u_i))\cap R)+\eps]$, we get $h\in[H(P\cap R) - \eps\frac{N-N_i}{N_i}, H(P\cap R) + \eps\frac{N-N_i}{N_i}]$.
If we show that $\frac{N-N_i}{N_i}\leq H(P\cap R)$ then the result follows. By the definition of entropy we observe that $H(P\cap R)\geq\frac{N_i}{N}\log\frac{N}{N_i}+\frac{N-N_i}{N}\log\frac{N}{N-N_i}$.
\begin{lem}
\label{lem:ineqFinal}
If $1>\frac{N_i}{N}>2/3$, it holds that $\frac{N-N_i}{N_i}\leq \frac{N_i}{N}\log\frac{N}{N_i}+\frac{N-N_i}{N}\log\frac{N}{N-N_i}$.
\end{lem}
\begin{proof}
    Let $\alpha=\frac{N_i}{N}$. We define $f(\alpha)=\alpha\log\frac{1}{\alpha}+(1-\alpha)\log\frac{1}{1-\alpha} - \frac{1}{\alpha}+1$. We get the first and the second derivative and we have $f'(\alpha)=\frac{1}{\alpha^2}+\log\frac{1}{\alpha} -\log\frac{1}{1-\alpha}$, and $f''(\alpha)=\frac{\alpha^2-\alpha\cdot\ln 4+ \ln 4}{(\alpha-1)\alpha^3\ln 2}$.
    For $\frac{2}{3}<\alpha<1$, the denominator of $f''(\alpha)$ is always negative, while the nominator of $f''(\alpha)$ is positive. Hence $f''(\alpha)\leq 0$ and $f'(\alpha)$ is decreasing. We observe that $f'(0.75)>0$ while $f'(0.77)<0$, hence there is a unique root of $f'$ which is $\beta\in(0.75,0.77)$. Hence for $\alpha\leq \beta$ $f'(\alpha)\geq 0$ so $f(\alpha)$ is increasing, while for $\alpha> \beta$ we have $f'(\alpha)\leq 0$ so $f(\alpha)$ is decreasing.
    We observe that
    $f(0.5)=0$ and $\lim_{\alpha\rightarrow 1} f(\alpha)=0$. Notice that $0.5<\frac{2}{3}<\beta<1$, so $f(\alpha)\geq 0$ for $\alpha\in[0.5,1)$. Recall that $2/3<\alpha< 1$ so $f(\alpha)\geq 0$. The result follows.
\end{proof}
Using Lemma~\ref{lem:ineqFinal}, we conclude that $h\in[(1-\eps)H(P\cap R), (1+\eps)H(P\cap R)]$.

\paragraph{Analysis.}
We first run a counting query on $\mathcal{T}$ in $O(\log^d n)$ time. Then the set $P_S$ is constructed in $O(\log^{d+1}n)$ time, running $O(\log n)$ queries in $\mathcal{S}$. In the first case of the query procedure (no point $p$ with $\frac{N_i}{N}>2/3$) we run $O(\frac{\log n}{\eps^2})$ oracle queries so in total it runs in $O(\frac{\log^{d+1}}{\eps^2})$ time.
In the second case of the query procedure (point $p$ with $\frac{N_i}{N}>2/3$) we run the query procedure of the previous Subsection using $\bar{\mathcal{S}}$ instead of $\mathcal{S}$, so it takes $O(\frac{\log^{d+3}}{\eps^2})$ time. Overall, the query procedure takes $O(\frac{\log^{d+3}}{\eps^2})$ time.

\begin{thm}
\label{thm:mult-approx}
Let $P$ be a set of $n$ points in $\Re^d$, where each point is associated with a color. A data structure of $O(n\log^{d}n)$ size can be constructed in $O(n\log^{d} n)$ time, such that given a query hyper-rectangle $R$ and a parameter $\eps\in(0,1)$, a value $h$ can be computed in $O\!\!\left(\frac{\log^{d+3} n}{\eps^2}\!\!\right)$ time, such that $\frac{1}{1+\eps}H(P\cap R)\!\leq\!\! h\!\!\leq\! (1+\eps)H(P\!\cap \!R)$, with high probability.
\end{thm}
This structure can be made dynamic under arbitrary insertions and deletions of points using well known techniques~\cite{ bentley1980decomposable, erickson2011static, overmars1983design, overmars1981worst}. The update time is $O(\log^d n)$.

\subsection{Efficient additive and multiplicative approximation}
\label{subsec:ApproxAddMult1}
Next, for $d=1$, we propose a deterministic, faster approximate data structure with query time $O(\polylog n)$ that returns an additive and multiplicative approximation of the entropy $H(P\cap R)$, given a query rectangle $R$.

Instead of using the machinery for entropy estimation on unknown distributions, we get the intuition from data structures that count the number of colors in a query region $R$.
In~\cite{gupta1995further}, the authors presented a data structure to count/report colors in a query interval for $d=1$.
In particular, they map the range color counting/reporting problem for $d=1$ to the standard range counting/reporting problem in $\Re^2$. Let $P$ be the set of $n$ colored points in $\Re^1$. Let $\bar{P}=\emptyset$ be the corresponding points in $\Re^2$ they construct. For every color $u_i\in U$, without loss of generality, let $P(u_i)=\{p_1, p_2,\ldots, p_k\}$ such that if $j<\ell$ then the $x$-coordinate of point $p_j$ is smaller than the $x$-coordinate of point $p_\ell$. For each point $p_j\in P(u_i)$, they construct the $2$-d point $\bar{p}_j=(p_j, p_{j-1})$ and they add it in $\bar{P}$. If $p_j=p_1$, then $\bar{p}_1=(p_1, -\infty)$. Given a query interval $R=[l,r]$ in $1$-d, they map it to the query rectangle $\bar{R}=[l,r]\times (-\infty, l)$. It is straightforward to see that a point of color $u_i$ exists in $R$ if and only if $\bar{R}$ contains exactly one transformed point of color $u_i$. Hence, using a range tree $\bar{\mathcal{T}}$ on $\bar{P}$ they can count (or report) the number of colors in $P\cap R$ efficiently.
While this is more than enough to count or report the colors in $P\cap R$, for the entropy we also need to know (in fact precompute) the number of points of each color $u_i$ in $P'$, along with the actual entropy in each canonical subset. Notice that a canonical subset/node in $\bar{\mathcal{T}}$ might belong to many different query rectangles $\bar{R}$ that correspond to different query intervals $R$. Even though a point of color $u_i$ appears only once in $\bar{R}\cap \bar{P}$, there can be multiple points with color $u_i$ in $R\cap P$. Hence, there is no way to know in the preprocessing phase the exact number of points of each color presented in a canonical node of $\bar{\mathcal{T}}$.
We overcome this technical difficulty by pre-computing for each canonical node $v$ in $\bar{\mathcal{T}}$, monotone pairs with approximate values of (interval, number of points), and (interval, entropy) over a sufficiently large number of intervals.
Another issue is that entropy is not monotone, so we split it into two monotone functions and we handle each of them separately until we merge them in the end to get the final estimation.

Before we start describing the data structure we prove some useful properties that we need later.
\begin{lem}
\label{lem:minEntropy}
Assume that we have a set $P'\subseteq P$ with $N=|P'|$ and $|u(P')|>2$ colors. Then the minimum entropy is encountered when we have $|u(P')|-1$ colors having exactly one point, and one color having $|P'|-|u(P')|+1$ points.
\end{lem}
\begin{proof}
Consider any other arbitrary instance. Let $u_i$ be the color with the maximum number of points in $P'$. We consider any other color $u_j\neq u_i$ having at least 2 points, so $|P'(u_i)|\geq |P'(u_j)|\geq 2$. We assume that we move one point from color $u_j$ to color $u_i$ and we argue that the new instance has lower entropy. If this is true, we can iteratively apply it, and whatever the initial instance is, we can create an instance as described in the lemma with lower entropy. Hence, the minimum entropy is encountered when we have $|u(P')|-1$ colors having exactly one point, and one color having all the rest $|P'|-u(P')+1$ points.

Initially, we have $$H(P')=\sum_{\ell\in u(P')}\frac{N_\ell}{N}\log\frac{N}{N_\ell}=\sum_{\ell\in u(P')}\frac{N_\ell}{N}(\log N - \log N_\ell)=\log N -\frac{1}{N}\sum_{\ell\in u(P')}N_\ell\log N_\ell.$$
The new instance has entropy $$H'=H(P')-\frac{1}{N}\left(-N_i\log N_i - N_j\log N_j + (N_i+1)\log (N_i+1) + (N_j-1)\log (N_j-1)\right).$$
Next, we show that $$H'\leq H(P')\Leftrightarrow -N_i\log N_i - N_j\log N_j + (N_i+1)\log (N_i+1) + (N_j-1)\log (N_j-1)\geq 0.$$
We define the function $$f(x)=(x+1)\log (x+1) - x\log x +(N_j-1)\log (N_j-1) - N_j\log N_j,$$
for $x\geq N_j\geq 2$.
We have $f'(x)=\log(x+1)-\log(x)\geq 0$ for $x>0$, so function $f$ is monotonically increasing for $x\geq 2$.
Since $x\geq N_j$, we have $f(x)\geq f(N_j)\geq 0$.
Hence, we proved that the new instance has lower entropy. In particular, if $N_i=N_j$ then the new instance has no higher entropy, and if $N_i>N_j$ then the new instance has strictly lower entropy.
\end{proof}

For a set of colored points $P'\subseteq P$, with $N=|P'|$, let $F(P')=N\cdot H(P')=\sum_{u_i\in u(P')}N_i\cdot \log\frac{N}{N_i}$, where $N_i$ is the number of points in $P'$ with color $u_i$.

\begin{lem}
\label{lem:monotone}
The function $F(\cdot)$ is monotonically increasing. Furthermore, $F(P')=O(n\log n)$, and the smallest non-zero value that $F(\cdot)$ can take is at least $2$.
\end{lem}
\begin{proof}
Let $p\in P$ be a point such that $p\notin P'$. We show that $F(P'\cup\{p\})\geq F(P')$. If $u(p)\notin u(P')$ it is clear that $F(P'\cup\{p\})\geq F(P')$ because all nominators in the log factors are increasing and a new positive term is added to the sum. Next, we focus on the more interesting case where $u(p)\in u(P')$. Without loss of generality assume that $u(P')=\{u_1,\ldots, u_k\}$ and $u(p)=u_k$. We have $F(P'\cup\{p\})=\sum_{i=1}^{k-1}N_i\log\frac{N+1}{N_i} + (N_{k}+1)\log\frac{N+1}{N_k+1}$. For $i<k$, each term $N_i\log\frac{N+1}{N_i}$ in $F(P'\cup\{p\})$ is larger than the corresponding term $N_i\log\frac{N}{N_i}$ in $F(P')$ (1). Let $g(x)=x\log\frac{c+x}{x}$, for any real number $c>2$. We have $g'(x)=\frac{(c+x)\ln\frac{c+x}{x} - c}{(c+x)\ln(2)}$. Using the well known inequality $\ln a\geq 1-\frac{1}{a}$, we note that $(c+x)\ln(1+\frac{c}{x})\geq (c+x)\frac{cx}{x(c+x)}=c$ so $g'(x)\geq 0$ and $g(x)$ is monotonically increasing. Hence we have $(N_k+1)\log\frac{N+1}{N_k+1}\geq N_k\log\frac{N}{N_k}$ (2). From (1), (2), we conclude that $F(P'\cup\{p\})\geq F(P')$.

The inequalities in the end follow straightforwardly from the monotonicity of $F$ and Lemma~\ref{lem:minEntropy} (we actually show a more general result in Lemma~\ref{lem:minEntropy}).
\end{proof}

\paragraph{Data structure.}
We apply the same mapping from $P$ to  $\bar{P}$ as described above~\cite{gupta1995further} and construct a range tree $\bar{\mathcal{T}}$ on $\bar{P}$. Then we visit each canonical node $v$ of $\bar{\mathcal{T}}$. If node $v$ contains two points with the same color then we can skip it because this node will not be returned as a canonical node for any query $\bar{R}$. Let $v$ be a node such that $\bar{P}_v$ does not contain two points with the same color. Let also $x_v$ be the smallest $x$-coordinate of a point in $\bar{P}_v$. Finally, let $U_v=u(\bar{P}_v)$, and $P(U_v)=\{p\in P\mid u(p)\in U_v\}$. Notice that $P(U_v)$ is a subset of $P$ and not of $\bar{P}$.
We initialize an empty array $S_v$ of size $O(\frac{\log n}{\eps})$. Each element $S_v[i]$ stores the maximum $x$ coordinate such that $(1+\eps)^i\geq |P(U_v)\cap [x_v,x]|$.
Furthermore, we initialize an empty array $H_v$ of size $O(\frac{\log n}{\eps})$. Each element $H_v[i]$ stores the maximum $x$ coordinate such that $(1+\eps)^i\geq F(P(U_v)\cap [x_v,x])$. We notice that both functions $F(\cdot)$, and cardinality of points are monotonically increasing.
For every node of $\bar{\mathcal{T}}$ we use $O(\frac{\log n}{\eps})$ space (from Lemma~\ref{lem:monotone} there are $O(\frac{\log n}{\eps})$ possible exponents $i$ in the discrete values $(1+\eps)^i$), so in total, the space of our data structure is $O(\frac{n}{\eps}\log^2 n)$. Next, we show that the data structure can be constructed in $O(\frac{n}{\eps}\log^5 n)$ time.
\begin{lem}
\label{lem:preproc}
The data structure $\bar{\mathcal{T}}$ can be constructed in $O\left(\frac{n}{\eps}\log^5 n\right)$ time.
\end{lem}
\begin{proof}
The structure of $\bar{\mathcal{T}}$ can be constructed in $O(n\log^2 n)$ time. For each color $\mathbf{u}\in u(P)$, we construct a $1$d binary search tree $T_\mathbf{u}$. In total, it takes $O(n\log n)$ time. These auxiliary trees are useful for the construction of our main data structure.
A $2$d range tree consists of one search binary tree with respect to $x$-coordinate and for each node in this tree there is a pointer to another tree based on the $y$ coordinates. Hence, it is a $2$-level structure. Recall that we need to compute the values in tables $S_v$, $H_v$ for each node $v$ in the $2$-level trees. For each tree in the second level we do the following. We visit the nodes level by level. Assume that we have already computed $S_v[i]$ and $H_v[i]$. In order to compute the next value in $H_v$ (or $S_v$), we run a binary search on the $x$-coordinates of $P$ that are larger than $H_v[i]$ (or $S_v[i]$). Let $x'$ be the $x$-coordinate value we check. We visit all colors $u$ stored in the leaf nodes of the subtree with root $v$ and we run another binary search on $T_u$ to get the total number of points of color $u$ in the range $[x_u,x']$. In that way we check whether the interval $[x_u,x']$ satisfies the definition of $H_v[i+1]$ (or $S_v[i+1]$). Based on this decision we continue the binary search on the $x$-coordinates of $P$.
Using the data structures $T_\mathbf{u}$ to run counting queries when needed, in each level we spend time $O(\frac{\log n}{\eps}(\sum_{z\in \mathcal{L}} \log n_z)\log n)=O(\frac{n\log^3 n}{\eps})$, where $\mathcal{L}$ is the set of leaf nodes of the current $2$-level tree and $n_z$ is the number of points with color equal to the color of point stored in $z$.
Notice that we run this algorithm only for the nodes of the tree that do not contain points with the same colors. The tree has $O(\log n)$ levels so for each $2$-level tree we spend $O(\frac{n\log^4 n}{\eps})$ time.
We finally notice that the $1$-level tree in $\bar{\mathcal{T}}$ has $O(\log n)$ levels and two nodes of the same level do not ``contain'' any point in common. Hence, the overall running time to compute all values $S_v[i], H_v[i]$ is $O(\frac{n\log^5 n}{\eps})$.  
\end{proof}

\paragraph{Query procedure.}
Given a query interval $R=[a,b]$, we run a query in $\bar{\mathcal{T}}$ using the query range $\bar{R}$.
Let $V=\{v_1, \ldots, v_k\}$ be the set of $k=O(\log^2 n)$ returned canonical nodes.
For each node $v\in V$ we run a binary search in array $S_v$ and a binary search in $H_v$ with key $b$.
Let $\ell_{v}^S$ be the minimum index such that $b\leq S_v[\ell_v^S]$ and $\ell_{v}^H$ be the minimum index such that $b\leq H_v[\ell_v^H]$.
From their definitions, it holds that $|P(U_v)\cap R|\leq (1+\eps)^{\ell_v^S}\leq (1+\eps)|P(U_v)\cap R|$, and 
$F(P(U_v)\cap R)\leq (1+\eps)^{\ell_v^H}\leq (1+\eps)F(P(U_v)\cap R)$.
Hence, we can approximate the entropy of $P(U_v)\cap R$, defining $\mathcal{H}_v=\frac{(1+\eps)^{\ell_v^H}}{(1+\eps)^{\ell_v^S-1}}$.
We find the overall entropy by merging together pairs of canonical nodes. Notice that we can do it easily using Equation~\ref{eq:entropyupdate} because all colors are different between any pair of nodes in $V$.
For example, we apply Equation~\ref{eq:entropyupdate} for two nodes $v, w\in V$
as follows:
\newcommand*{\Scale}[2][4]{\scalebox{#1}{$#2$}}%
\[\Scale[1.2]{ \frac{(1+\eps)^{\ell_v^S}\mathcal{H}_v+(1+\eps)^{\ell_w^S}\mathcal{H}_w+(1+\eps)^{\ell_v^S}\log\left(\frac{(1+\eps)^{\ell_v^S}+(1+\eps)^{\ell_w^S}}{(1+\eps)^{\ell_v^S-1}}\right) + (1+\eps)^{\ell_w^S}\log\left(\frac{(1+\eps)^{\ell_v^S}+(1+\eps)^{\ell_w^S}}{(1+\eps)^{\ell_w^S-1}}\right)}{(1+\eps)^{\ell_v^S-1}+(1+\eps)^{\ell_w^S-1}}.}\]
In the end we compute the overall entropy $\mathcal{H}$.

\paragraph{Correctness and analysis.}
The next Lemma shows that $\mathcal{H}_v$ is a good approximation of $H(P(U_v)\cap R)$.
\begin{lem}
\label{lem:helper}
It holds that $H(P(U_v)\cap R)\leq \mathcal{H}_v\leq (1+\eps)^2H(P(U_v)\cap R)$.
\end{lem}
\begin{proof}
We have $\mathcal{H}_v=\frac{(1+\eps)^{\ell_v^H}}{(1+\eps)^{\ell_v^S-1}}$. From their definitions, we have that $|P(U_v)\cap R|\leq (1+\eps)^{\ell_v^S}\leq (1+\eps)|P(U_v)\cap R|$, and 
$F(P(U_v)\cap R)\leq (1+\eps)^{\ell_v^H}\leq (1+\eps)F(P(U_v)\cap R)$. It also holds that $(1+\eps)^{\ell_v^S-1}\leq |P(U_v)\cap R|$ and $(1+\eps)^{\ell_v^S-1}\geq \frac{|P(U_v)\cap R|}{(1+\eps)}$.
Hence $\mathcal{H}_v\leq \frac{(1+\eps)F(P(U_v)\cap R)}{|P(U_v)\cap R|/(1+\eps)}\leq (1+\eps)^2H(P(U_v)\cap R)$. Furthermore, $\mathcal{H}_v\geq \frac{F(P(U_v)\cap R)}{|P(U_v)\cap R|}=H(P(U_v)\cap R)$.
\end{proof}

The next Lemma shows the correctness of our procedure.
\begin{lem}
\label{lem:error1d}
If we set $\eps\leftarrow \frac{\eps}{4\cdot c\cdot\log\log n}$, it holds that $H(P\cap R)\leq \mathcal{H}\leq (1+\eps)H(P\cap R)+\eps$, for a constant $c>0$.
\end{lem}
\begin{proof}
We assume that we take the union of two nodes $v, w\in V$ using Equation~\ref{eq:entropyupdate}. We can use this equation because nodes $v, w$ do not contain points with similar colors. Let $H_1=H(P(U_v)\cap R)$, $H_2=H(P(U_w)\cap R)$, $N_1=|P(U_v)\cap R|$, and $N_2=|P(U_2)\cap R|$.
We have
\[\Scale[1.1]{ \mathcal{H}_{v,w}=\frac{(1+\eps)^{\ell_v^S}\mathcal{H}_v+(1+\eps)^{\ell_w^S}\mathcal{H}_w+(1+\eps)^{\ell_v^S}\log\left(\frac{(1+\eps)^{\ell_v^S}+(1+\eps)^{\ell_w^S}}{(1+\eps)^{\ell_v^S-1}}\right) + (1+\eps)^{\ell_w^S}\log\left(\frac{(1+\eps)^{\ell_v^S}+(1+\eps)^{\ell_w^S}}{(1+\eps)^{\ell_w^S-1}}\right)}{(1+\eps)^{\ell_v^S-1}+(1+\eps)^{\ell_w^S-1}}.}\]
Using Lemma~\ref{lem:helper}, we get
\[ \Scale[1.1]{\mathcal{H}_{v,w}\leq \frac{(1+\eps)^4N_1H_1+(1+\eps)^4N_2H_2+(1+\eps)^2N_1\log\left((1+\eps)^2\frac{N_1+N_2}{N_1}\right)+(1+\eps)^2N_2\log\left((1+\eps)^2\frac{N_1+N_2}{N_2}\right)}{N_1+N_2}}\]
and we conclude that
$$\mathcal{H}_{v,w}\leq (1+\eps)^4H((P(U_v)\cup P(U_w))\cap R)+(1+\eps)^2\log(1+\eps)^2.$$

Similarly if we have computed $\mathcal{H}_{x,y}$ for two other nodes $x,y\in V$, then
$$\mathcal{H}_{x,y}\leq (1+\eps)^4H((P(U_x)\cup P(U_y))\cap R)+(1+\eps)^2\log(1+\eps)^2.$$

If we compute their union, we get
$$\mathcal{H}_{v,w,x,y}\leq (1+\eps)^6H((P(U_v)\cup P(U_w)\cup P(U_x)\cup P(U_y))\cap R)+[(1+\eps)^4+(1+\eps)^2]\log(1+\eps)^2.$$

At the end of this process, we have $$\mathcal{H}\geq H(P\cap R)$$ because all intermediate estimations of entropy are larger than the actual entropy. For a constant $c$, it also holds that
$$\mathcal{H}\leq (1+\eps)^{c\log(\log n)}H(P\cap R)+\sum_{j=1}^{c\log(\log n)/2}(1+\eps)^{2j}\log(1+\eps)^2.$$
This quantity can be bounded by
$$\mathcal{H}\leq (1+\eps)^{c\log(\log n)}H(P\cap R)+c\log(\log n)(1+\eps)^{c\log(\log n)}\log(1+\eps).$$
We have the factor $\log(\log n)$ because $|V|=O(\log^2 n)$ so the number of levels of recurrence is $O(\log(\log n))$.

Next, we show that if we set $\eps\leftarrow\frac{\eps}{4\cdot c\log(\log n)}$, then $\mathcal{H}\leq (1+\eps)H(P\cap R)+\eps$.

We have $$\left(1+\frac{\eps/4}{c\log(\log n)}\right)^{c\log(\log n)}\leq e^{\eps/4}\leq 1+\eps.$$
The first inequality holds because of the well known inequality $(1+x/n)^n\leq e^x$. The second inequality is always true for $\eps\in (0,1)$.
Then we have
$$(1+\eps)c\log(\log n)\log\left(1+\frac{\eps}{4\cdot c\log(\log n)}\right)\leq 2c\log(\log n)\log\left(1+\frac{\eps}{4\cdot c\log(\log n)}\right).$$
Next, we show that this quantity is at most $\eps$.
Let $L=c\log( \log n)$ and
let $$f(x)=x-2L\log\left(1+\frac{x}{4L}\right)$$
be a real function for $x\in[0,1]$.
We have
$$f'(x)=1-\frac{2L}{L\ln(16)+x\ln(2)}.$$
We observe that $\ln(16)\approx 2.77$ and $x\ln(2)\geq 0$ so $f'(x)\geq 0$ and $f$ is monotonically increasing. So $f(x)\geq f(0)=0$. Hence, for any $\eps\in[0,1]$ we have
$$\eps-2L\log\left(1+\frac{\eps}{4L}\right)\geq 0.$$

We conclude with
\begin{align*} \mathcal{H}\leq (1+\eps)H(P\cap R)+\eps.\tag*{\qed}\end{align*}
\renewcommand{\qed}{}    
\end{proof}

We need $O(\log^2 n)$ time to get $V$ from $\bar{\mathcal{T}}$. Then, we run binary search for each node $v\in V$ so we spend $O(\log^2 n \log\frac{\log n\log \log n}{\eps})=O(\log^2 n \log\frac{\log n}{\eps})$ time. We merge and update the overall entropy in time $O(|V|)$, so in total the query time is $O(\log^2 n \log\frac{\log n}{\eps})$.

\begin{thm}
\label{thm:mult-add-approx}
Let $P$ be a set of $n$ points in $\Re^1$, where each point is associated with a color, and let $\eps\in(0,1)$ be a parameter. A data structure of $O(\frac{n}{\eps}\log^2 n)$ size can be constructed in $O(\frac{n}{\eps}\log^5 n)$ time, such that given a query \new{interval} $R$, a value $h$ can be computed in $O\left(\log^2 n \log\frac{\log n}{\eps}\right)$ time, such that $H(P\cap R)\leq h\leq (1+\eps)H(P\cap R)+\eps$.
\end{thm}


\section{Approximate Data Structures for R-Entropy Queries}
In this section we describe data structures that return the \renyi entropy in a query range, approximately. First, we present a (randomized) data structure that returns an additive approximation of the \renyi entropy. Then, for $d=1$, we design a deterministic and faster data structure that returns an additive approximation. Finally, we present a data structure that returns a multiplicative approximation of the \renyi entropy.

\subsection{Additive Approximation for R-Entropy Queries}
\label{subsec:RenApprox1}
In this Subsection, we construct a data structure on $P$ such that given a query rectangle $R$, a parameter $\alpha>1$, and a parameter $\Delta>0$, it returns a value $h$ such that $\ren_{\alpha}(P\cap R)-\Delta\leq h\leq \ren_{\alpha}(P\cap R)+\Delta$. We will use ideas from the area of finding an additive approximation of the \renyi entropy of an unknown distribution in the samples-only model (access only to random samples; only $\textsf{SAMP}_D$ oracles) or the dual access model (access to random samples, and probability mass of a value; access to both $\textsf{SAMP}_D, \textsf{EVAL}_D$ oracles).

We use the notation from the previous Subsections where $D$ is an unknown distribution over a set of values $\out_1,\ldots, \out_N$. It is known~\cite{acharya2016estimating, obremski2017renyi} that if we get $O(\frac{n^{1-1/\alpha}}{\Delta^2}\log N)$ samples from the unknown distribution $D$, then we can get a $\Delta$ additive approximation of the \renyi entropy of order $\alpha$ of $D$ with high probability in $O(\frac{N^{1-1/\alpha}}{\Delta^2}\log N)$ time, for integer values of $\alpha>1$.
Using, ideas from~\cite{alon1996space, thorup2004tabulation, estfrmomstr}, we can extend this result to any real value of $\alpha>1$ and $\Delta\in(0,1)$, getting $O\left(\max\left\{1, \frac{1}{(\alpha-1)^2}\right\}\cdot\frac{\alpha\cdot N^{1-1/\alpha}}{\Delta^2}\log N\right)$ samples. In particular, if $\alpha\in (1,2]$ then we get 
$O\left(\frac{1}{(\alpha-1)^2}\frac{\alpha\cdot N^{1-1/\alpha}}{\Delta^2}\log N\right)$ samples, while if $\alpha>2$, we get
$O\left(\frac{\alpha\cdot N^{1-1/\alpha}}{\Delta^2}\log N\right)$ samples.

Even though the number of samples is sublinear on $N$, it is not $O(\polylog(N))$. 
A natural question to ask is whether this complexity can be improved in the dual access model, i.e., whether less queries can be performed in the dual access model to get an additive approximation, as we had in the Shannon entropy. Interestingly, in~\cite{caferov2015optimal} the authors studied the additive approximation of the \renyi entropy of an unknown distribution $D$ in the dual access model. They first prove a lower bound, showing that $\Omega(\frac{N^{1-1/\alpha}}{2^\Delta})$ queries in the dual access model are necessary to get an additive approximation. Hence, unlike in the Shannon entropy, the dual access model does not help to perform $\polylog(N)$ queries. Furthermore, in~\cite{caferov2015optimal} they give an algorithm that returns a $\Delta$ additive approximation of the \renyi entropy of order $\alpha$ that performs $O(\frac{N^{1-1/\alpha}}{(1-2^{(1-\alpha)\Delta})^2}\log N)$ queries in the dual access model in $O(\frac{N^{1-1/\alpha}}{(1-2^{(1-\alpha)\Delta})^2}\log N)$ time, for any real value of $\alpha>1$.
We note that the number of queries in the dual access model does not dominate the number of samples in the samples-only model and vice versa. For example, if $\alpha=2$ and $\Delta=0.01$ then $\frac{N^{1-1/\alpha}}{(1-2^{(1-\alpha)\Delta})^2}\log N > \frac{\alpha\cdot N^{1-1/\alpha}}{\Delta^2}\log N$, while if $\alpha=3$ and $\Delta=0.01$, then $\frac{N^{1-1/\alpha}}{(1-2^{(1-\alpha)\Delta})^2}\log N < \frac{\alpha\cdot N^{1-1/\alpha}}{\Delta^2}\log N$. Hence, we will design a data structure that gets the best of the two.

The idea in all the estimation algorithms above is the same: They first use ideas from the AMS sketch~\cite{alon1996space, thorup2004tabulation} to get a multiplicative approximation of the $\alpha$-th frequency moment $\sum_{i=1}^N (D(\out_i))^\alpha$, which leads to an additive approximation for the \renyi entropy.

Next, we show the data structure we use to get an additive approximation for R-entropy queries in our setting. As we had in the Shannon entropy, in our setting we do not know the number of colors in $P'=P\cap R$, which is equivalent to the number of values $N$ in distribution $D$. However, it is sufficient to use the upper bound $|u(P')|\leq n$.

\paragraph{Data structure.} 
For each color $u_i\in U$ we construct a range tree $\mathcal{T}_i$ over $P(u_i)$ for range counting queries. Similarly, we construct a range tree $\mathcal{T}$ over $P$ for counting queries. These range trees will be used for the $\textsf{EVAL}_D$ oracle in the dual access model.
We also construct the range tree $\mathcal{S}$ for returning uniform samples in a query rectangle. This tree will be used for the $\textsf{SAMP}_D$ oracles in the dual access model or the samples-only model.
Overall, the data structure has $O(n\log^{d-1} n)$ size and can be constructed in $O(n\log^d n)$ time.

\paragraph{Query procedure.}
We are given a hyper-rectangle $R$ and a parameter $\alpha$.
If $\frac{1}{(1-2^{(1-\alpha)\Delta})^2} \geq  \max\{1,\frac{1}{(\alpha-1)^2}\}\cdot\frac{\alpha}{\Delta^2}$, then we use $\mathcal{S}$ to get $O(\max\{1,\frac{1}{(\alpha-1)^2}\}\cdot\frac{\alpha\cdot N^{1-1/\alpha}}{\Delta^2}\log n)$ random samples from $P\cap R$. Then the algorithm from~\cite{acharya2016estimating} is executed.
If $\frac{1}{(1-2^{(1-\alpha)\Delta})^2} < \max\{1,\frac{1}{(\alpha-1)^2}\}\cdot\frac{\alpha}{\Delta^2}$, then we mimic the algorithm from~\cite{caferov2015optimal} on $P\cap R$ in the dual access model. When a random sample is required (oracle $\textsf{SAMP}_D$) we use $\mathcal{S}$. When the probability of a color $u_i$ is required (oracle $\textsf{EVAL}_D$) in $P\cap R$, we use $\mathcal{T}_i$ to get $|u(P\cap R)|$ and $\mathcal{T}$ to get $|P\cap R|$ and we set the probability of color $u_i$ to be $\frac{|u(P\cap R)|}{|P\cap R|}$.

\paragraph{Correctness.} The correctness follows from~\cite{acharya2016estimating, caferov2015optimal} estimating the \renyi entropy in the samples-only model and the dual access model.

\paragraph{Analysis.} 
In the first case, the query procedure runs $O(\max\{1,\frac{1}{(\alpha-1)^2}\}\cdot\frac{\alpha\cdot n^{1-1/\alpha}}{\Delta^2}\log n)$ queries to $\mathcal{S}$, where each query takes $O(\log^d n)$ time. In the second case, the query procedure runs $O(\frac{n^{1-1/\alpha}}{(1-2^{(1-\alpha)\Delta})^2}\log n)$ queries in $\mathcal{T}_i$, where each query takes $O(\log^{d} n)$ time.
Overall the query time is $$O\left(\min\left\{\max\left\{1,\frac{1}{(\alpha-1)^2}\right\}\cdot\frac{\alpha}{\Delta^2}, \frac{1}{(1-2^{(1-\alpha)\Delta})^2}\right\}\cdot n^{1-1/\alpha}\cdot \log^{d+1} n\right).$$ If $\alpha$ is an integer number then $O(\frac{n^{1-1/\alpha}}{\Delta^2}\log n)$ samples are only required in the first case, so the overall query time can be improved to $O\left(\min\left\{\frac{1}{\Delta^2}, \frac{1}{(1-2^{(1-\alpha)\Delta})^2}\right\}\cdot n^{1-1/\alpha}\log^{d+1} n\right)$.
We conclude with the next theorem.

\begin{thm}
\label{thm:Renadd-approx}
Let $P$ be a set of $n$ points in $\Re^d$, where each point is associated with a color. A data structure of $O(n\log^{d-1}n)$ size can be constructed in $O(n\log^{d} n)$ time, such that given a query hyper-rectangle $R$, a real parameter $\alpha>1$ and a real parameter $\Delta$, a value $h$ can be computed such that $\ren_{\alpha}(P\cap R)-\Delta\leq h\leq \ren_{\alpha}(P\cap R)+\Delta$, with high probability. The query time is  $O\left(\min\left\{\frac{1}{(\alpha-1)^2}\cdot\frac{\alpha}{\Delta^2}, \frac{1}{(1-2^{(1-\alpha)\Delta})^2}\right\}\cdot n^{1-1/\alpha}\log^{d+1} n\right)$ if $\alpha\in (1,2]$, and $O\left(\min\left\{\frac{\alpha}{\Delta^2}, \frac{1}{(1-2^{(1-\alpha)\Delta})^2}\right\}\cdot n^{1-1/\alpha}\log^{d+1} n\right)$ if $\alpha>2$.
\end{thm}
This data structure can be made dynamic under arbitrary insertions and deletions of points using well known techniques~\cite{ bentley1980decomposable, erickson2011static, overmars1983design, overmars1981worst}. The update time is $O(\log^d n)$.


\subsection{Faster Additive Approximation for $d=1$}
\label{subsec:RenApprox2}
Next, for $d=1$, we propose a deterministic, faster approximate data structure with query time $O(\polylog n)$ that returns an additive approximation of the \renyi entropy $\ren_{\alpha}(P\cap R)$, given a query rectangle $R$. The additive approximation term will be $\eps\cdot\frac{\alpha+1}{\alpha-1}$.

Instead of using the machinery for entropy estimation on unknown distributions, we get the intuition from data structures that count the number of colors in a query region $R$, as we did for the Shannon entropy.
Again, we consider the mapping $\bar{P}\subset \Re^2$ of $P\subset \Re$ as shown in~\cite{gupta1995further} and described in Subsection~\ref{subsec:ApproxAddMult1}. Recall that having a range tree $\bar{\mathcal{T}}$ on $\bar{P}$ allows us to count or report the number of colors in $P\cap R$ efficiently.
While this is more than enough to count or report the colors in $P\cap R$, for the \renyi entropy we also need to know (in fact precompute) the number of points of each color $u_i$ in $P'=P\cap R$, along with the actual \renyi entropy in each canonical subset. Notice that a canonical subset/node in $\bar{\mathcal{T}}$ might belong to many different query rectangles $\bar{R}$ that correspond to different query intervals $R$. Even though a point of color $u_i$ appears only once in $\bar{R}\cap \bar{P}$, there can be multiple points with color $u_i$ in $R\cap P$. Hence, there is no way to know in the preprocessing phase the exact number of points of each color presented in a canonical node of $\bar{\mathcal{T}}$. Furthermore, the \renyi entropy is not monotone.
We overcome the technical difficulties by pre-computing for each canonical node $v$ in $\bar{\mathcal{T}}$, monotone pairs with approximate values of (interval, number of points), and (interval, sum of number of points of each color to the power of $\alpha$) over a sufficiently large number of intervals.

Before we start describing the data structure we prove some useful properties that we need later.

For a set of colored points $P'\subseteq P$, with $N=|P'|$, let $G(P')=\sum_{u_i\in u(P')}N_i^\alpha$, where $N_i$ is the number of points in $P'$ with color $u_i$.
\begin{lem}
\label{lem:RenGfunc}
    The function $G(\cdot)$ is monotonically increasing. Furthermore, $G(P')=O(n^{\alpha+1})$, and the smallest value that $G(\cdot)$ can take if $u(P')>1$ is at least $2$.
\end{lem}
\begin{proof}
    Let $p\in P$ be a point such that $p\notin P'$. We show that $G(P'\cup\{p\})\geq G(P')$.
    If $u(p)\notin u(P')$ then $G(P'\cup \{p\})=G(P')+1^\alpha>G(P')$. If $u(p)\in u(P')$, let $u(p)=u_j$.
    Then $G(P'\cup \{p\})=\sum_{u_i\in u(P')\setminus u_j}N_i^\alpha + (N_j+1)^\alpha>\sum_{u_i\in u(P')\setminus u_j}N_i^\alpha + N_j^\alpha=G(P')$.
    
    The inequalities in the end follow straightforwardly from the monotonicity of $G$.
\end{proof}

\paragraph{Data structure.}
We apply the same mapping from $P$ to  $\bar{P}$ as described above~\cite{gupta1995further} and construct a range tree $\bar{\mathcal{T}}$ on $\bar{P}$. Then we visit each canonical node $v$ of $\bar{\mathcal{T}}$. If node $v$ contains two points with the same color then we can skip it because this node will not be returned as a canonical node for any query $\bar{R}$. Let $v$ be a node such that $\bar{P}_v$ does not contain two points with the same color. Let also $x_v$ be the smallest $x$-coordinate of a point in $\bar{P}_v$. Finally, let $U_v=u(\bar{P}_v)$, and $P(U_v)=\{p\in P\mid u(p)\in U_v\}$. Notice that $P(U_v)$ is a subset of $P$ and not of $\bar{P}$.
We initialize an empty array $S_v$ of size $O(\frac{\log n}{\eps})$. Each element $S_v[i]$ stores the maximum $x$ coordinate such that $(1+\eps)^i\geq |P(U_v)\cap [x_v,x]|$.
Furthermore, we initialize an empty array $H_v$ of size $O(\frac{\alpha\log n}{\eps})$. Each element $H_v[i]$ stores the maximum $x$ coordinate such that $(1+\eps)^i\geq G(P(U_v)\cap [x_v,x])$. We notice that both functions $G(\cdot)$, and cardinality of points are monotonically increasing.
For every node of $\bar{\mathcal{T}}$ we use $O(\frac{\alpha\log n}{\eps})$ space (from Lemma~\ref{lem:RenGfunc} there are $O(\frac{\alpha \cdot \log n}{\eps})$ possible exponents $i$ in the discrete values $(1+\eps)^i$), so in total, the space of our data structure is $O(\frac{\alpha\cdot n}{\eps}\log^2 n)$. 
Using the proof of Lemma~\ref{lem:preproc}, the data structure can be constructed in $O(\frac{\alpha \cdot n}{\eps}\log^5 n)$ time.

\paragraph{Query procedure.}
Given a query interval $R=[a,b]$, we run a query in $\bar{\mathcal{T}}$ using the query range $\bar{R}$.
Let $V=\{v_1, \ldots, v_k\}$ be the set of $k=O(\log^2 n)$ returned canonical nodes.
For each node $v\in V$ we run a binary search in array $S_v$ and a binary search in $H_v$ with key $b$.
Let $\ell_{v}^S$ be the minimum index such that $b\leq S_v[\ell_v^S]$ and $\ell_{v}^H$ be the minimum index such that $b\leq H_v[\ell_v^H]$.
From their definitions, it holds that $|P(U_v)\cap R|\leq (1+\eps)^{\ell_v^S}\leq (1+\eps)|P(U_v)\cap R|$, and 
$G(P(U_v)\cap R)\leq (1+\eps)^{\ell_v^H}\leq (1+\eps)G(P(U_v)\cap R)$.
We return $\mathcal{H}=\frac{1}{\alpha-1}\log\left(\frac{\left(\sum_{v_i\in V}(1+\eps)^{\ell_{v_i}^S}\right)^{\alpha}}{\sum_{v_i\in V}(1+\eps)^{\ell_{v_i}^H-1}}\right).$

\paragraph{Correctness and analysis.}
\begin{lem}
\label{lem:Renprop-add}
    It holds that $$|P\cap R|\leq \sum_{v_i\in V}(1+\eps)^{\ell_{v_i}^S}\leq (1+\eps)|P\cap R|$$ and $$\sum_{u_i\in u(P\cap R)}|P(u_i)\cap R|^\alpha \geq \sum_{v_i\in V}(1+\eps)^{\ell_{v_i}^H-1}\geq \frac{1}{1+\eps}\sum_{u_i\in u(P\cap R)}|P(u_i)\cap R|^\alpha.$$
\end{lem}
\begin{proof}
    We first focus on the first inequality. Let $v_i\in V$. By definition, we had that $|P(U_{v_i})\cap R|\leq (1+\eps)^{\ell_{v_i}^S}\leq (1+\eps)|P(U_{v_i})\cap R|$. We take the sum over the canonical nodes in $V$ and we get $\sum_{v_i\in V}|P(U_{v_i})\cap R|\leq \sum_{v_i\in V}(1+\eps)^{\ell_{v_i}^S}\leq (1+\eps)\sum_{v_i\in V}|P(U_{v_i})\cap R|$. We note that $\sum_{v_i\in V}|P(U_{v_i})\cap R|=|P\cap R|$ because no color is shared between two different nodes in $V$. Hence, the first inequality follows.

    Next, we focus on the second inequality. Let $v_i\in V$. By definition, we had that $G(P(U_{v_i})\cap R)\leq (1+\eps)^{\ell_{v_i}^H}\leq (1+\eps)G(P(U_{v_i})\cap R)$. Hence, it also follows that $G(P(U_{v_i})\cap R)\geq (1+\eps)^{\ell_{v_i}^H-1}\geq \frac{1}{1+\eps}G(P(U_{v_i})\cap R)$. We take the sum over the canonical nodes in $V$ and we get $\sum_{v_i\in V}G(P(U_{v_i})\cap R)\geq \sum_{v_i\in V}(1+\eps)^{\ell_{v_i}^H-1}\geq \frac{1}{1+\eps}\sum_{v_i\in V}G(P(U_{v_i})\cap R)$. Recall that $G(P(U_{v_i})\cap R)=\sum_{u_j\in U_{v_i}}|P(u_j)\cap R|^\alpha$, so $\sum_{v_i\in V}G(P(U_{v_i})\cap R)=\sum_{u_i\in u(P\cap R)}|P(u_i)\cap R|^\alpha$, since no color is shared between two different nodes in $V$.
    The second inequality follows.
\end{proof}
The next Lemma shows the correctness of our procedure.
\begin{lem}
   If we set $\eps\leftarrow \eps/2$, it holds that $\ren_{\alpha}(P\cap R)\leq \mathcal{H}\leq \ren_{\alpha}(P\cap R)+\eps\cdot\frac{\alpha+1}{\alpha-1}$. 
\end{lem}
\begin{proof}
From Lemma~\ref{lem:Renprop-add}, we have
\begin{align*}
\mathcal{H}&=\frac{1}{\alpha-1}\log\left(\frac{\left(\sum_{v_i\in V}(1+\eps/2)^{\ell_{v_i}^S}\right)^{\alpha}}{\sum_{v_i\in V}(1+\eps/2)^{\ell_{v_i}^H-1}}\right)\geq \frac{1}{\alpha-1}\log\left(\frac{|P\cap R|^\alpha}{\sum_{u_i\in u(P\cap R)}|P(u_i)\cap R|^\alpha}\right)\\
&=\frac{1}{\alpha-1}\log\left(\frac{1}{\sum_{u_i\in u(P\cap R)}\frac{|P(u_i)\cap R|^\alpha}{|P\cap R|^\alpha}}\right)=\ren_{\alpha}(P\cap R).
\end{align*}
From Lemma~\ref{lem:Renprop-add}, we also have,
\begin{align*}
    \mathcal{H}&\leq \frac{1}{\alpha-1}\log\left(\frac{(1+\eps/2)^\alpha|P\cap R|^\alpha}{\frac{1}{1+\eps/2}\sum_{u_i\in u(P\cap R)}|P(u_i)\cap R|^\alpha}\right)=\ren_{\alpha}(P\cap R)+\frac{\alpha+1}{\alpha-1}\log(1+\eps/2)\\
    &\leq \ren_{\alpha}(P\cap R)+\eps\cdot \frac{\alpha+1}{\alpha-1}.
\end{align*}
The last inequality holds because $\log(1+\eps/2)\leq \eps$ for $\eps\geq 0$.
\end{proof}
We need $O(\log^2 n)$ time to get $V$ from $\bar{\mathcal{T}}$. Then, we run binary search for each node $v\in V$ so we spend $O(\log^2 n \log\frac{\alpha\cdot\log n}{\eps})$ time. We merge and update the overall entropy in time $O(|V|)$, so in total the query time is $O(\log^2 n \log\frac{\alpha\cdot\log n}{\eps})$.

\begin{thm}
\label{thm:Renmult-add-approx}
Let $P$ be a set of $n$ points in $\Re^1$, where each point is associated with a color, let $\alpha>1$ be a parameter and let $\eps\in(0,1)$. A data structure of $O(\frac{\alpha\cdot n}{\eps}\log^2 n)$ size can be constructed in $O(\frac{\alpha\cdot n}{\eps}\log^5 n)$ time, such that given a query \new{interval} $R$, a value $h$ can be computed in $O\left(\log^2 n \log\frac{\alpha\cdot \log n}{\eps}\right)$ time, such that $\ren_{\alpha}(P\cap R)\leq h\leq \ren_{\alpha}(P\cap R)+\eps\cdot \frac{\alpha+1}{\alpha-1}$.
\end{thm}

\subsection{Multiplicative Approximation}
\label{subsec:Rmult}
While the problem of estimating the \renyi entropy has been studied in the samples-only model and the dual access model, to the best of our knowledge there is no known multiplicative approximation for every $\alpha>1$. Interestingly, by taking advantage of the properties of the geometric space, we are able to return a multiplicative $(1+\eps)$-approximation of the \renyi entropy in the query setting for any $\alpha>1$.
Our high level idea is the following. Harvey et al.~\cite{harvey2008sketching} show a multiplicative approximation of the \renyi entropy in the streaming setting for $\alpha\in (1,2]$. While in the streaming setting their algorithm does not work for $\alpha>2$, (they only give a lower bound on the number of samples they get when $\alpha>2$), we show that in our query setting, we can extend it to every $\alpha>1$.
First, similarly to the multiplicative approximation for the Shannon entropy, we decide if the \renyi entropy $\ren_{\alpha}(P\cap R)$ is sufficiently large by checking whether there exists a color $u_i\in u(P\cap R)$ that contains more than $2/3$ of the points in $P\cap R$. If no such color exists then $\ren_\alpha(P\cap R)$ is sufficiently large and an additive approximation using Theorem~\ref{thm:Renmult-add-approx} can be used to derive a multiplicative approximation. On the other hand, if such a color $u_i$ exists, we use a technical lemma from~\cite{harvey2008sketching} that shows that a multiplicative approximation (by a sufficiently small approximation factor) of $t-1$ suffices to get a multiplicative approximation of $\log(t)$. Notice that in our case the value $t$ is the inverse of the $\alpha$-th moment of the distribution in $P\cap R$. In order to compute a multiplicative approximation of $t-1$, using the results in~\cite{harvey2008sketching}, it suffices to compute a a multiplicative approximation of $\gamma_1=1-\left(\frac{|P(u_i)\cap R|}{|P\cap R|}\right)^\alpha$ and a multiplicative approximation of $\gamma_2=\sum_{u_j\in u(P\cap R)\setminus\{u_i\}}\left(\frac{|P(u_j)\cap R|}{|P\cap R|}\right)^\alpha$. We approximate $\gamma_2$ using a data structure for estimating the $\alpha$-th frequency moment in the query setting as shown in the next paragraph. Interestingly, in our setting, after we have identified the color $u_i$ the value $\gamma_1$ can be computed exactly using two range trees. In contrast, in~\cite{harvey2008sketching} they get a multiplicative approximation of $\gamma_1$ in the streaming setting only for $\alpha\in (1,2]$.

\paragraph{Data structure for $\alpha$-th frequency moment.}
Before we start describing our data structure for returning a multiplicative approximation in the query setting, we show an efficient way to compute the $\alpha$-th frequency moment in the query setting. This is an important tool that we are going to use in the design of our data structure, later. 
Using the results from~\cite{alon1996space, thorup2004tabulation}, as described in~\cite{estfrmomstr}, we can get a multiplicative approximation of the $\alpha$-th frequency moment in the samples-only model. Hence, using the range tree for range sampling in our model we can directly get the following useful result. 
\begin{lem}
    \label{lem:moments}
    Given a set of $n$ weighted points $P\subset \Re^d$, there exists a data structure of $O(n\log^{d-1} n)$ space that is constructed in $O(n\log^d n)$ time, such that given a query rectangle $R$, a parameter $\alpha>1$ and a parameter $\eps\in(0,1)$, it returns a value $h$ in $O(\frac{\alpha \cdot n^{1-1/\alpha}}{\eps^2}\log^{d+1}n)$ time, such that 
    $\sum_{u_j\in u(P\cap R)}\left(\frac{|P(u_j)\cap R|}{|P\cap R|}\right)^\alpha\leq h \leq (1+\eps)\sum_{u_j\in u(P\cap R)}\left(\frac{|P(u_j)\cap R|}{|P\cap R|}\right)^\alpha$, with high probability. 
\end{lem}
Using the modified range tree $\bar{\mathcal{S}}$ to perform sampling excluding the points of a color, we can also get the next result.
\begin{lem}
    \label{lem:moments2}
    Given a set of $n$ weighted points $P\subset \Re^d$, there exists a data structure of $O(n\log^{d} n)$ space that is constructed in $O(n\log^d n)$ time, such that given a query rectangle $R$, a color $u_i\in U$, a parameter $\alpha>1$ and a parameter $\eps\in(0,1)$, it returns a value $h$ in $O(\frac{\alpha \cdot n^{1-1/\alpha}}{\eps^2}\log^{d+1}n)$ time, such that 
    $\sum_{u_j\in u(P\cap R)\setminus\{u_i\}}\left(\frac{|P(u_j)\cap R|}{|P\cap R|}\right)^\alpha\leq h \leq (1+\eps)\sum_{u_j\in u(P\cap R)\setminus\{u_i\}}\left(\frac{|P(u_j)\cap R|}{|P\cap R|}\right)^\alpha$, with high probability. 
\end{lem}

\paragraph{Data structure.}
For each color $u_i\in U$ we construct a range tree $\mathcal{T}_i$ over $P(u_i)$ for counting queries as in Subsection~\ref{subsec:ApproxMult}. Similarly, we construct a range tree $\mathcal{T}$ over $P$ for counting queries.
We also construct the range tree $\mathcal{S}$ for returning uniform samples in a query rectangle. We also construct the variation of the range tree, denoted by $\bar{\mathcal{S}}$, that returns a sample uniformly at random, excluding the points of a color $u_j\in U$, as described in Subsection~\ref{subsec:ApproxMult}.
Finally, we construct the data structure from Lemma~\ref{lem:moments2}, for approximating the $\alpha$-th frequency moment.

Overall, the proposed data structure can be computed in $O(n\log^d n)$ time and it uses $O(n\log^d n)$ space.

\paragraph{Query procedure.}
We first explore whether there exists a color $u_i\in U$ such that $|P(u_i)\cap R|\geq \frac{2}{3}|P\cap R|$, as we did in Subsection~\ref{subsec:ApproxMult}.
Using $\mathcal{T}$ we compute $N=|P\cap R|$. Using $\mathcal{S}$ we get $\frac{\log (2n)}{\log 3}$ independent random samples from $P\cap R$. Let $P_S$ be the set of returned samples. For each $p\in P_S$ with $u(p)=u_i$, we run a counting query in $\mathcal{T}_i$ to get $N_i=|P(u_i)\cap R|$. Finally, we check whether $\frac{N_i}{N}>2/3$.

If we do not find a point $p\in P_S$ (assuming $u(p)=u_i$) with $\frac{N_i}{N}> 2/3$ then we run the additive approximation query from Theorem~\ref{thm:Renadd-approx}, for $\Delta=\log(\frac{3}{2})\eps$ to get the additive estimator $h_{\textsf{add}}$. We return $h=h_{\textsf{add}}$.

Next, we assume that the algorithm found a point with color $u_i$ satisfying $\frac{N_i}{N}>2/3$.
We set $h_1=1-\left(\frac{N_i}{N}\right)^\alpha$.
Let $\eps_0 = \eps/C_1$, for a constant $C_1$ as shown in Lemma 5.7 of~\cite{harvey2008sketching}, and let $\eps_1=\eps_0/3$. 
Then, for simplicity, we distinguish between $\alpha\leq 2$ and $\alpha>2$. For $\alpha\in(1,2]$ (resp. $\alpha>2$), we set $\eps_2=(\alpha-1)\eps_1/C_2$ (resp. $\eps_2=\eps_1/C_2$), where $C_2$ is a sufficiently large constant as shown in~\cite{harvey2008sketching}, and we use the data structure from Lemma~\ref{lem:moments2} to compute an $(1+\eps_2)$ multiplicative approximation of the $\alpha$-th frequency moment in $P\cap R$ excluding the points with color $u_i$. Let $h'$ be this estimator.
We set $h_2=h'\cdot\frac{(N-N_i)^\alpha}{N^\alpha}$, and $\bar{h}=h_1-h_2$.
We also use the data structure from Lemma~\ref{lem:moments} to compute an $(1+\eps_1)$ multiplicative approximation of the $\alpha$-th frequency moment in $P\cap R$ (without excluding any color). Let $\hat{h}$ be this estimator. We return $h=\frac{1}{\alpha-1}\log\left(\frac{\bar{h}}{\hat{h}}+1\right)$.

\paragraph{Correctness.} We show the correctness by proving the following lemma. 
\begin{lem}
    It holds that $\frac{1}{1+\eps}\ren_{\alpha}(P\cap R)\leq h \leq (1+\eps)\ren_{\alpha}(P\cap R)$, with high probability.
\end{lem}
\begin{proof}
Using the proof of Lemma~\ref{lem:ProbEntr} we correctly decide whether there exists a color $u_i\in U$ such that $\frac{N_i}{N}>2/3$, with high probability. 

If there is no color $u_i$ with $\frac{N_i}{N}>2/3$, then $\ren_{\alpha}(P\cap R)\geq \log\frac{1}{\max_{u_j\in u(P\cap R)}N_j/N}\geq \log(3/2)$. Therefore, the additive $\log(3/2)\cdot \eps$ approximation $h_{\textsf{add}}$ returns a multiplicative $(1+\eps)$ approximation.

Next, we assume that there exists a color $u_i$ satisfying $\frac{N_i}{N}>2/3$.
In~\cite{harvey2008sketching}, (Lemma 5.7) the authors show that for any real number $t>4/9$, it suffices to have  multiplicative $(1+\eps_0)$-approximation to $t-1$, to compute a multiplicative $(1+\eps)$ approximation to $\log(t)$.
In our proof we set $t=\frac{1}{\sum_{u_j\in u(P\cap R)}\left(\frac{N_j}{N}\right)^\alpha}>1$. 
If we show that $\frac{1}{1+\eps_0}(t-1)\leq \frac{\bar{h}}{\hat{h}}\leq (1+\eps_0)(t-1)$, then the result follows.

We note that,
$$t-1=\frac{1}{\sum_{u_j\in u(P\cap R)}\left(\frac{N_j}{N}\right)^\alpha}-1=\frac{1-\sum_{u_j\in u(P\cap R)}\left(\frac{N_j}{N}\right)^\alpha}{\sum_{u_j\in u(P\cap R)}\left(\frac{N_j}{N}\right)^\alpha}.$$
From Lemma~\ref{lem:moments} and definition of $\hat{h}$, we have $\sum_{u_j\in u(P\cap R)}\!\!\left(\frac{N_j}{N}\right)^\alpha\!\!\!\leq \!\hat{h}\!\leq \!\!(1+\eps_1)\!\sum_{u_j\in u(P\cap R)}\!\!\left(\frac{N_j}{N}\right)^\alpha$. Hence, we have a good estimation of the denominator. Next we focus on the nominator $1-\sum_{u_j\in u(P\cap R)}\left(\frac{N_j}{N}\right)^\alpha$.
We consider two cases, $\alpha\in(1,2]$ and $\alpha>2$.
We can re-write it as $1-\left(\frac{N_i}{N}\right)^\alpha - \sum_{u_j\in u(P\cap R)\setminus\{u_i\}}\left(\frac{N_j}{N}\right)^\alpha$.

In~\cite{harvey2008sketching}, they consider the case where $\alpha\in(1,2]$. They show that if we compute a $(1+\eps_2)$ multiplicative approximation of $1-\left(\frac{N_i}{N}\right)^\alpha$, denoted by $\beta_1$, and a $(1+\eps_2)$ multiplicative approximation of $\sum_{u_j\in u(P\cap R)\setminus\{u_i\}}\left(\frac{N_j}{N}\right)^\alpha$, denoted by $\beta_2$, then $\beta_1-\beta_2$ is a $(1+\eps_1)$ multiplicative approximation of $1-\sum_{u_j\in u(P\cap R)}\left(\frac{N_j}{N}\right)^\alpha$.
Recall that $h_1=1-\left(\frac{N_i}{N}\right)^\alpha$ so this is an exact estimator of $1-\left(\frac{N_i}{N}\right)^\alpha$. We show that $h_2$ is a $(1+\eps_2)$ multiplicative approximation of $\sum_{u_j\in u(P\cap R)\setminus\{u_i\}}\left(\frac{N_j}{N}\right)^\alpha$. By Lemma~\ref{lem:moments2}, and by the definition of $h'$ we have that $\sum_{u_j\in u(P\cap R)\setminus\{u_1\}}\left(\frac{N_j}{N-N_i}\right)^\alpha\leq h'\leq (1+\eps_2)\sum_{u_j\in u(P\cap R)\setminus\{u_1\}}\left(\frac{N_j}{N-N_i}\right)^\alpha$. Notice that $h_2=h'\cdot\frac{(N-N_i)^\alpha}{N^\alpha}$. So, $\sum_{u_j\in u(P\cap R)\setminus\{u_1\}}\left(\frac{N_j}{N}\right)^\alpha\leq h_2\leq (1+\eps_2)\sum_{u_j\in u(P\cap R)\setminus\{u_1\}}\left(\frac{N_j}{N}\right)^\alpha$. Hence, we have that
$\frac{1}{1+\eps_1}\left(1-\sum_{u_j\in u(P\cap R)}\left(\frac{N_j}{N}\right)^\alpha\right)\leq \bar{h}\leq (1+\eps_1)\left(1-\sum_{u_j\in u(P\cap R)}\left(\frac{N_j}{N}\right)^\alpha\right)$.
Next, we have $$\frac{\bar{h}}{\hat{h}}\leq \frac{(1+\eps_1)\left(1-\sum_{u_j\in u(P\cap R)}\left(\frac{N_j}{N}\right)^\alpha\right)}{\sum_{u_j\in u(P\cap R)}\left(\frac{N_j}{N}\right)^\alpha}=(1+\eps_1)(t-1)\leq (1+\eps_0)(t-1),$$
and
$$\frac{\bar{h}}{\hat{h}}\geq\frac{\frac{1}{1+\eps_1}\left(1-\sum_{u_j\in u(P\cap R)}\left(\frac{N_j}{N}\right)^\alpha\right)}{(1+\eps_1)\sum_{u_j\in u(P\cap R)}\left(\frac{N_j}{N}\right)^\alpha}=\frac{1}{(1+\eps_1)^2}(t-1)\geq \frac{1}{1+\eps_0}(t-1).$$
Hence, we conclude $\frac{1}{1+\eps_0}(t-1)\leq \frac{\bar{h}}{\hat{h}}\leq (1+\eps_0)(t-1)$, and the result follows.

Next, we show that the analysis also holds for $\alpha>2$. Recall that $\eps_2=\eps_1/C_2$.
For any $x\in(0,1/3]$, we have $\frac{x^\alpha}{x}\leq \left(\frac{1}{3}\right)^{\alpha-1}\leq 1-\frac{2}{3}$. Hence,
$$\frac{\sum_{u_j\in u(P\cap R)\setminus\{u_i\}}\left(\frac{N_j}{N}\right)^\alpha}{1-\left(\frac{N_i}{N}\right)^\alpha}\leq \frac{\sum_{u_j\in u(P\cap R)\setminus\{u_i\}}\left(\frac{N_j}{N}\right)^\alpha}{1-\frac{N_i}{N}}\leq \frac{\sum_{u_j\in u(P\cap R)\setminus\{u_i\}}\frac{N_j}{N}(1-2/3)}{1-\frac{N_i}{N}}=1-\frac{2}{3}.$$
This implies that if we compute a multiplicative $(1+\eps_2)$-approximation to $1-\left(\frac{N_i}{N}\right)^\alpha$ and a multiplicative $(1+\eps_2)$-approximation to $\sum_{u_j\in u(P\cap R)\setminus\{u_i\}}\left(\frac{N_j}{N}\right)^\alpha$, we can compute a multiplicative $(1+\eps_1)$-approximation to $1-\sum_{u_j\in u(P\cap R)}\left(\frac{N_j}{N}\right)^\alpha$. The result follows by repeating the same analysis as for $\alpha\in(1,2]$.
\end{proof}

\paragraph{Analysis.}
We compute $P_S$ and identify whether there exists color $u_i$ with $N_i/N>2/3$ in $O(\log^{d+1} n)$ time. If there is no color $u_i$ with $N_i/N<2/3$ then the additive approximation query from Theorem~\ref{thm:Renadd-approx} runs in $O\left(\frac{\alpha}{(\alpha-1)^2\eps^2}\cdot n^{1-1/\alpha}\log^{d+1} n\right)$ time if $\alpha\in (1,2]$, and $O\left(\frac{\alpha}{\eps^2}\cdot n^{1-1/\alpha}\log^{d+1} n\right)$ time if $\alpha>2$. If there is a color $u_i$ with $N_i/N>2/3$ then we run a query from Lemma~\ref{lem:moments} and a query from Lemma~\ref{lem:moments2} in $O\left(\frac{\alpha\cdot n^{1-1/\alpha}}{\eps^2}\log^{d+1}n\right)$ time. In total, the query procedure takes $O\left(\frac{\alpha}{(\alpha-1)^2\eps^2}\cdot n^{1-1/\alpha}\log^{d+1} n\right)$ time if $\alpha\in (1,2]$, and $O\left(\frac{\alpha}{\eps^2}\cdot n^{1-1/\alpha}\log^{d+1} n\right)$ time if $\alpha>2$.

\begin{thm}
\label{thm:Renmult-approx}
Let $P$ be a set of $n$ points in $\Re^d$, where each point is associated with a color. A data structure of $O(n\log^{d}n)$ size can be constructed in $O(n\log^{d} n)$ time, such that given a query hyper-rectangle $R$, a real parameter $\alpha>1$ and a real parameter $\eps\in(0,1)$, a value $h$ can be computed such that $\frac{1}{1+\eps}\ren_{\alpha}(P\cap R)\leq h\leq (1+\eps)\ren_{\alpha}(P\cap R)$, with high probability. The query time is $O\left(\frac{\alpha}{(\alpha-1)^2\eps^2}\cdot n^{1-1/\alpha}\log^{d+1} n\right)$ time if $\alpha\in (1,2]$, and $O\left(\frac{\alpha}{\eps^2}\cdot n^{1-1/\alpha}\log^{d+1} n\right)$ time if $\alpha>2$.
\end{thm}
This data structure can be made dynamic under arbitrary insertions and deletions of points using well known techniques~\cite{ bentley1980decomposable, erickson2011static, overmars1983design, overmars1981worst}. The update time is $O(\log^d n)$.
\section{Partitioning Using the (Expected) Shannon Entropy}
\label{sec:partition}
The new data structures can be used to accelerate some known partitioning algorithms with respect to the (expected) Shannon entropy.
\newcommand{\DS}{\textsf{DS}}
Let $\DS$ be one of our new data structures over $n$ items that can be constructed in $O(P(n))$ time, has $O(S(n))$ space, and given a query range $R$, returns a value $h$ in $O(Q(n))$ time such that $\frac{1}{\alpha}H-\beta \leq h\leq \alpha\cdot H+\beta$, where $H$ is the Shannon entropy of the items in $R$, and $\alpha\geq 1$, $\beta\geq 0$ two error thresholds.
On the other hand, the straightforward way to compute the (expected) entropy without using any data structure has preprocessing time $O(1)$, query time $O(n)$ and it returns the exact Shannon entropy in a query range.

\newcommand{\MaxPart}{\textsf{MaxPart}}
\newcommand{\SumPart}{\textsf{SumPart}}
\new{In most cases, we use the expected entropy to partition the dataset, as this is standard in entropy-based partitioning and clustering algorithms. Aside from being a useful quantity that bounds both the uncertainty and the size of a bucket, it is also monotone.} All our data structures can work for both the Shannon entropy and expected Shannon entropy quantity almost verbatim.
We define two optimization problems. Let $\MaxPart$ be the problem of constructing a partitioning with $k$ buckets that maximizes/minimizes the maximum (expected) entropy in a bucket. Let $\SumPart$ be the problem of constructing a partitioning with $k$ buckets that maximizes/minimizes the sum of (expected) entropies over the buckets.
For simplicity, in order to compare the running times, we skip the $\log(n)$ factors from the running times. \new{We use $\O(\cdot)$ to hide $\polylog (n)$ factors from the running time.}

\newcommand{\Err}{\textsf{Error}}
\newcommand{\DP}{\textsf{DP}}

\paragraph{Partitioning for $d=1$.}
We can easily solve $\MaxPart$ using dynamic programming: $\DP[i,j]=\min_{\ell<i}\max\{\DP[i-\ell,j-1], \Err[i-\ell+1,i])\}$, where $\DP[i,j]$ is the minimum max entropy of the first $i$ items using $j$ buckets, and $\Err[i,j]$ is the expected entropy among the items $i$ and $j$.
Since $\Err$ is monotone, we can find the optimum $\DP[i,j]$ running a binary search on $\ell$, i.e., we do not need to visit all indexes $\ell<i$ one by one to find the optimum. Without using any data structure the running time to find $\DP[n,k]$ is $\O(kn^2)$. Using $\DS$, the running time for partitioning is $\O(P(n)+knQ(n))$. If we use the data structure from Section~\ref{subsec:DS1} for $t=0.5$, then the running time is $\O\left(kn\sqrt{n}\right)=o(kn^2)$.

Next we consider approximation algorithms for the $\MaxPart$ and $\SumPart$ problems.

It is easy to observe that the maximum value and the minimum non-zero value of the optimum solution of $\MaxPart$ are bounded polynomially on $n$. Let $[l_M, r_M]$ be the range of the optimum values.
We discretize the range $[l_M, r_M]$ by a multiplicative factor $(1+\eps)$. We run a binary search on the discrete values. For each value $e\in[l_M, r_M]$ we consider, we construct a new bucket by running another binary search on the input items, trying to expand the bucket until its expected entropy is at most $e$. We repeat the same for all buckets and we decide if we should increase or decrease the error $e$ in the next iteration. In the end, the solution we find is within an $(1+\eps)$ factor far from the max expected entropy in the optimum partitioning.
Without using any data structure, we need $\O(n\log\frac{1}{\eps})$ time to construct the partitioning. If we use $\DS$ we need time $\O\left(P(n)+kQ(n)\log\frac{1}{\eps}\right)$. If we use the data structure in Subsection~\ref{subsec:ApproxMult} we have partition time $\O\left(n+\frac{k}{\eps^2}\log\frac{1}{\eps}\right)=o\left(n\log\frac{1}{\eps}\right)$. If we allow a $\Delta$ additive approximation in addition to the $(1+\eps)$ multiplicative approximation, we can use the data structure in Subsection~\ref{subsec:ApproxAdd} having partition time $\O\left(n+\frac{k}{\Delta^2}\log\frac{1}{\eps}\right)=o\left(n\log\frac{1}{\eps}\right)$.

Next, we focus on the $\SumPart$ problem.
It is known from~\cite{guha2006approximation} (Theorems 5, 6) that if the error function is monotone (such as the expected entropy) then we can get a partitioning with $(1+\eps)$-multiplicative approximation in $\O\left(P(n)+\frac{k^3}{\eps^2}Q(n)\right)$ time.
Hence, the straightforward solution without using a data structure returns an $(1+\eps)$-approximation of the optimum partitioning in $\O\left(\frac{k^3}{\eps^2}n\right)$ time.
If we use the data structure from Subsection~\ref{subsec:ApproxMult} we have running time $\O\left(n+\frac{k^3}{\eps^4}\right)=o\left(\frac{k^3}{\eps^2}n\right)$ with multiplicative error $(1+\eps)^2$. If we set $\eps\leftarrow \eps/3$ then in the same asymptotic running time we have error $(1+\eps)$.
If we also allow $\Delta\cdot n$ additive approximation, we can use the additive approximation $\DS$ from Subsection~\ref{subsec:ApproxAdd}. The running time will be $\O\left(n+\frac{k^3}{\eps^2\Delta^2}\right)=o\left(\frac{k^3}{\eps^2}n\right)$.

\paragraph{Partitioning for $d>1$.}
Partitioning and constructing histograms in high dimensions is usually a challenging task, since most of the known algorithms with theoretical guarantees are very expensive~\cite{cormode2011synopses}. However, there is a practical method with some conditional error guarantees, that works very well in any constant dimension $d$ and it has been used in a few papers~\cite{baltrunas2006multi, liang2022janusaqp, liang2021combining}.
The idea is to construct a tree having a rectangle containing all points in the root. In each iteration of the algorithm, we choose to split (on the median in each coordinate or find the best split) the (leaf) node with the minimum/maximum (expected) entropy.
As stated in previous papers, let make the assumption that an optimum algorithm for either $\MaxPart$ or $\SumPart$ is an algorithm that always chooses to split the leaf node with the smallest/largest expected entropy.
Using the straightforward solution without data structures, we can construct an ``optimum'' partitioning in $O(kn)$ time by visiting all points in every newly generated rectangle.
Using $\DS$, the running time of the algorithm is $O(P(n)+kQ(n))$.
In order to get an optimum solution we use $\DS$ from Subsection~\ref{subsec:DSd}. The overall running time is $O(n^{(2d-1)t+1}+kn^{1-t})$. This is minimized for $n^{(2d-1)t+1}=kn^{1-t}\Leftrightarrow t=t^*=\frac{\log k}{2d\log n}$, so the overall running time is $O(kn^{1-t^*})=o(kn)$.
If we allow $(1+\eps)$-multiplicative approximation we can use the $\DS$ from Subsection~\ref{subsec:ApproxMult}. The running time will be $\O\left(n+\frac{k}{\eps^2}\right)=o(kn)$.
If we allow a $\Delta$-additive approximation, then we can use the $\DS$ from Subsection~\ref{subsec:ApproxAdd} with running time $\O\left(n+\frac{k}{\Delta^2}\right)=o(kn)$.
\section{Conclusion}
\label{sec:conclusion}
In this work, we presented efficient data structures for computing (exactly and approximately) the Shannon and \renyi entropy of the points in a rectangular query in sub-linear time.
Using our new data structures we can accelerate partitioning algorithms for columnar compression (Example~\ref{ex1}) and histogram construction (Example~\ref{ex2}). Furthermore, we can accelerate the exploration of high uncertainty regions for data cleaning (Example~\ref{ex3}).

There are multiple interesting open problems derived from this work. i) Our approximate data structures are dynamic
but our exact data structures are static. Is it possible to design dynamic data structures for returning the exact entropy? 
ii) There remains a gap between the proposed lower and upper bounds of our exact data structures, and closing this gap is an interesting open problem.
iii) Can we extend the faster deterministic approximation data structures from Subsection~\ref{subsec:ApproxAddMult1} and Subsection~\ref{subsec:RenApprox2} in higher dimensions?

\bibliographystyle{alphaurl}
\bibliography{ref}

\newpage
\appendix

\new{
\section{Updating the \renyi entropy}
\label{sec:appndx:renyi}
\begin{lem}
    Let $P_1, P_2\subset P$ such that $u(P_1)\cap u(P_2)=\emptyset$. It holds that,
\begin{equation}
\ren_{\alpha}(P_1\cup P_2)=\frac{1}{\alpha-1}\log\left(\frac{(|P_1|+|P_2|)^{\alpha}}{|P_1|^{\alpha}\cdot 2^{(1-\alpha)\ren_{\alpha}(P_1)}+|P_2|^{\alpha}\cdot 2^{(1-\alpha)\ren_{\alpha}(P_2)}}\right).  
\end{equation}
\end{lem}
\begin{proof}

We have,
\begin{align*}
    \frac{1}{\alpha-1}&\log\left(\frac{(|P_1|+|P_2|)^{\alpha}}{|P_1|^{\alpha}\cdot 2^{(1-\alpha)\ren_{\alpha}(P_1)}+|P_2|^{\alpha}\cdot 2^{(1-\alpha)\ren_{\alpha}(P_2)}}\right)
    \\&=
    \frac{1}{\alpha-1}\log\left(\frac{(|P_1|+|P_2|)^{\alpha}}{|P_1|^{\alpha}\cdot \sum_{i=1}^m\left(\frac{|P_1(u_i)|}{|P_1|}\right)^\alpha+|P_2|^{\alpha}\cdot \sum_{i=1}^m\left(\frac{|P_2(u_i)|}{|P_2|}\right)^\alpha}\right)
    \\&=
    \frac{1}{\alpha-1}\log\left(\frac{(|P_1|+|P_2|)^{\alpha}}{ \sum_{i=1}^m|P_1(u_i)|^\alpha+\sum_{i=1}^m|P_2(u_i)|^\alpha}\right).
\end{align*}
Since $u(P_1)\cap u(P_2)=\emptyset$, for every $i\in[m]$ either $|P_1(u_i)|=0$ or $|P_2(u_i)|=0$, so it holds that $|P_1(u_i)|^\alpha+|P_2(u_i)|^\alpha=(|P_1(u_i)|+|P_2(u_i)|)^\alpha$. Hence,
\begin{align*}
     \frac{1}{\alpha-1}&\log\left(\frac{(|P_1|+|P_2|)^{\alpha}}{ \sum_{i=1}^m|P_1(u_i)|^\alpha+\sum_{i=1}^m|P_2(u_i)|^\alpha}\right)
     \\&=
     \frac{1}{\alpha-1}\log\left(\frac{(|P_1|+|P_2|)^{\alpha}}{ \sum_{i=1}^m\left(|P_1(u_i)|+|P_2(u_i)|\right)^\alpha}\right)
     \\&=\frac{1}{\alpha-1}\log\left(\frac{1}{\sum_{i=1}^m\left(\frac{|P_1(u_i)|+|P_2(u_i)|}{|P_1\cup P_2|}\right)^\alpha}\right)=\ren_\alpha(P_1\cup P_2).\qedhere
\end{align*}
\end{proof}
}

\new{
\begin{lem}
    Let $P_3\subset P_1\subset P$ such that $|u(P_3)|=1$ and $u(P_1\setminus P_3)\cap u(P_3)=\emptyset$. It holds that,
\begin{equation}
\ren_{\alpha}(P_1\setminus P_3)=\frac{1}{\alpha-1}\log\left(\frac{(|P_1|-|P_3|)^{\alpha}}{|P_1|^{\alpha}\cdot 2^{(1-\alpha)\ren_{\alpha}(P_1)}-|P_3|^\alpha}\right).  
\end{equation}
\end{lem}
\begin{proof}
    We have,
    \begin{align*}
        \frac{1}{\alpha-1}&\log\left(\frac{(|P_1|-|P_3|)^{\alpha}}{|P_1|^{\alpha}\cdot 2^{(1-\alpha)\ren_{\alpha}(P_1)}-|P_3|^\alpha}\right)
        \\&=
        \frac{1}{\alpha-1}\log\left(\frac{(|P_1|-|P_3|)^{\alpha}}{|P_1|^{\alpha}\cdot \sum_{i=1}^m\left(\frac{|P_1(u_i)|}{|P_1|}\right)^\alpha-|P_3|^\alpha}\right)
        \\&=
        \frac{1}{\alpha-1}\log\left(\frac{(|P_1|-|P_3|)^{\alpha}}{ \sum_{i=1}^m\left(|P_1(u_i)|^\alpha\right)-|P_3|^\alpha}\right)
    \end{align*}
    Since $|u(P_3)|=1$ and $u(P_1)\cap u(P_3)=\emptyset$ it holds that $$\sum_{i=1}^m\left(|P_1(u_i)|^\alpha\right)-|P_3|^\alpha = \sum_{i=1}^m\left(|P_1(u_i)|-|P_3|\right)^\alpha=\sum_{i=1}^m\left(|P_1(u_i)|-P_3(u_i)\right)^\alpha.$$
    Hence,
    \begin{align*}
        \frac{1}{\alpha-1}&\log\left(\frac{(|P_1|-|P_3|)^{\alpha}}{ \sum_{i=1}^m\left(|P_1(u_i)|^\alpha\right)-|P_3|^\alpha}\right)
        \\&=
        \frac{1}{\alpha-1}\log\left(\frac{(|P_1|-|P_3|)^{\alpha}}{ \sum_{i=1}^m\left(|P_1(u_i)|-|P_3(u_i)|\right)^\alpha}\right)=\ren_\alpha(P_1\setminus P_3).\qedhere
    \end{align*}
\end{proof}
}
\end{document}